\documentclass[a4paper,UKenglish,cleveref,autoref,thm-restate,colorlinks]{lipics-v2021}

\pdfoutput=1 \hideLIPIcs

\title{Taming Infinity one Chunk at a Time: Concisely Represented Strategies in One-Counter MDPs} 

\titlerunning{Concisely-Represented Strategies in OC-MDPs}

\author{Michal Ajdar\'{o}w}{Masaryk University, Czech Republic}{}{}{}

\author{{James C.~A.} Main}{F.R.S.-FNRS and UMONS -- Université de Mons, Belgium}{}{}{Research Fellow of the Fonds de la Recherche Scientifique -- FNRS and member of the TRAIL institute.}

\author{Petr Novotn\'{y}}{Masaryk University, Czech Republic}{}{}{}

\author{Mickael Randour}{F.R.S.-FNRS and UMONS -- Université de Mons, Belgium}{}{}{Research Associate of the Fonds de la Recherche Scientifique -- FNRS and member of the TRAIL institute.}

\authorrunning{M.~Ajdar\'{o}w, {J.~C.~A.} Main, P.~Novotn\'{y}, M.~Randour} 

\Copyright{Michal Ajdar\'{o}w, {James C.~A.} Main, Petr Novotn\'{y}, Mickael Randour} 

\ccsdesc[100]{Theory of computation~Probabilistic computation}
\ccsdesc[100]{Theory of computation~Logic and verification} 

\keywords{one-counter Markov decision processes, randomised strategies, termination, reachability} 

\category{} 

\relatedversion{}

\funding{This work has been supported by the Fonds de la Recherche Scientifique - FNRS under Grant n° T.0188.23 (PDR ControlleRS).} 

\acknowledgements{We thank Sougata Bose for pointing out the reference~\cite{BPR2006}.}

\nolinenumbers

\EventEditors{John Q. Open and Joan R. Access}
\EventNoEds{2}
\EventLongTitle{42nd Conference on Very Important Topics (CVIT 2016)}
\EventShortTitle{CVIT 2016}
\EventAcronym{CVIT}
\EventYear{2016}
\EventDate{December 24--27, 2016}
\EventLocation{Little Whinging, United Kingdom}
\EventLogo{}
\SeriesVolume{42}
\ArticleNo{23}

\usepackage[utf8]{inputenc}
\usepackage[T1]{fontenc}
\usepackage{amsmath, amssymb, amsthm, verbatim, mathtools, stmaryrd}
\usepackage{microtype}
\usepackage{IEEEtrantools}
\usepackage{array}
\usepackage{breakcites}
\usepackage{arydshln}
\usepackage[ruled, vlined]{algorithm2e}

\usepackage{tikz}
\usetikzlibrary{automata,positioning,arrows}

\tikzset{
  >=stealth,
  initial text=,
}

\mathchardef\mhyphen="2D

\newcommand{\init}{\mathsf{init}}

 \newcommand{\bigo}{\mathcal{O}}
\newcommand{\integerInterval}[1]{\llbracket{}#1\rrbracket{}}

\newcommand{\ooInt}[2]{\left]#1, #2\right[}
\newcommand{\ccInt}[2]{\left[#1, #2\right]}
\newcommand{\ocInt}[2]{\left]#1, #2\right]}
\newcommand{\coInt}[2]{\left[#1, #2\right[}

\newcommand{\ptime}{\textsf{P}}
\newcommand{\pspace}{\textsf{PSPACE}}
\newcommand{\npspace}{\textsf{NPSPACE}}
\newcommand{\np}{\textsf{NP}}
\newcommand{\coNP}{\textsf{co-NP}}
\newcommand{\exptime}{\textsf{EXPTIME}}

\newcommand{\posSLP}{\textsf{PosSLP}}
\newcommand{\etr}{\textsf{ETR}}
\newcommand{\coetr}{\textsf{co-ETR}}
  
\newcommand{\sqsx}{x} \newcommand{\sqsxVect}{\bar{x}} \newcommand{\sqsm}{m} \newcommand{\sqsy}{y} \newcommand{\sqsi}{i} \newcommand{\sqsn}{n}

\newcommand{\IR}{\mathbb{R}}

\newcommand{\IN}{\mathbb{N}}
\newcommand{\INpos}{\IN_{>0}}
\newcommand{\IQ}{\mathbb{Q}}
\newcommand{\IC}{\mathbb{C}}

\newcommand{\INbar}{\bar{\IN}}
\newcommand{\INposBar}{\bar{\IN}_{>0}}

\newcommand{\subsets}[1]{2^{#1}}
\newcommand{\proba}{\mathbb{P}}

\newcommand{\dist}[1]{\mathcal{D}(#1)}
\newcommand{\cyl}[1]{\mathsf{Cyl}(#1)}
\newcommand{\cylVerb}[2]{\mathsf{Cyl}_{#1}(#2)}
\newcommand{\supp}[1]{\mathsf{supp}(#1)}

\newcommand{\graph}{G}
\newcommand{\vertexSet}{V}
\newcommand{\vertex}{v}

\newcommand{\edgeSet}{E}

\newcommand{\probaMC}[1]{\proba_{#1}} \newcommand{\probaMCverb}[2]{\proba_{#1, #2}} \newcommand{\probaG}[2]{\probaMC{#1}^{#2}} \newcommand{\probaGverb}[3]{\probaMCverb{#1}{#2}^{#3}} 

\newcommand{\mdp}{\mathcal{M}}
\newcommand{\mdpStateSpace}{S}
\newcommand{\mdpState}{s}

\newcommand{\mdpActionSpace}{A}
\newcommand{\mdpAction}{a}

\newcommand{\mdpTrans}{\delta}
\newcommand{\mdpTuple}{(\mdpStateSpace, \mdpActionSpace, \mdpTrans)}

\newcommand{\playSet}[1]{\mathsf{Plays}(#1)}
\newcommand{\play}{\pi}
\newcommand{\playPrefix}[2]{#1_{|#2}}
\newcommand{\histSet}[1]{\mathsf{Hist}(#1)}
\newcommand{\histPart}{\mathcal{H}}
\newcommand{\hist}{h}
\newcommand{\histConcat}[2]{#1\cdot{}#2}
\newcommand{\last}[1]{\mathsf{last}(#1)}
\newcommand{\first}[1]{\mathsf{first}(#1)}

\newcommand{\indexPosition}{\ell}
\newcommand{\subindexPosition}{\iota}
\newcommand{\indexLast}{n}

\newcommand{\weight}{w}
\newcommand{\weightVal}{u}
\newcommand{\ocmdp}{\mathcal{Q}}

\newcommand{\ocmdpFin}[2]{\mdp^{\leq #2}(#1)}

\newcommand{\ocStateSpace}{Q}
\newcommand{\ocState}{q}
\newcommand{\ocStateB}{p}
\newcommand{\ocStateC}{t}
\newcommand{\ocStateD}{t'}
\newcommand{\ocCount}{k}
\newcommand{\ocCountB}{k'}
\newcommand{\ocConfig}{s}
\newcommand{\ocConfigPair}{(\ocState, \ocCount)}
\newcommand{\ocActionSpace}{A}
\newcommand{\ocAction}{a}
\newcommand{\ocActionB}{b}
\newcommand{\ocActionC}{c}
\newcommand{\ocTrans}{\delta}
\newcommand{\ocTransFin}{\delta^{\leq \counterUB}}
\newcommand{\ocTransInf}{\delta^{\leq \infty}}

\newcommand{\ocTuple}{(\ocStateSpace, \ocActionSpace, \ocTrans, \weight)}
\newcommand{\ocTupleWeightless}{(\ocStateSpace, \ocActionSpace, \ocTrans)}

\newcommand{\counterUB}{r} 

\newcommand{\period}{\rho}

\newcommand{\mchain}{\mathcal{C}}
\newcommand{\mchainTuple}{(\mdpStateSpace, \mdpTrans)}
\newcommand{\ocChain}{\mathcal{P}}
\newcommand{\ocChainTuple}{(\ocStateSpace, \ocTrans)}
\newcommand{\ocChainFin}[2]{\mchain^{\leq #2}(#1)}

\newcommand{\intPart}{\mathcal{I}}
\newcommand{\intPartB}{\mathcal{J}}
\newcommand{\intPartC}{\mathcal{K}}

\newcommand{\interval}{I}
\newcommand{\intervalB}{J}
\newcommand{\intNum}{d}
\newcommand{\intSize}{n}
\newcommand{\intIndex}{j}
\newcommand{\intBound}{b}
\newcommand{\intLB}{\intBound^-}
\newcommand{\intUB}{\intBound^+}

\newcommand{\powerIndex}{\alpha}
\newcommand{\powerMax}{\beta}
\newcommand{\compressChainStrat}[1]{\mchain^{#1}_{\intPart}}
\newcommand{\compressChain}{\compressChainStrat{\strat}}
\newcommand{\compressChainVerbose}{\mchain^{\strat}_{\intPart}(\ocmdp)}
\newcommand{\compressChainB}{\mchain^{\strat}_{\intPart}({\ocmdp'})}
\newcommand{\compressChainStateSpace}{\mdpStateSpace_{\intPart}}
\newcommand{\compressChainStateSpaceJ}{\mdpStateSpace_{\interval}}
\newcommand{\compressChainStateSpaceStar}{\mdpStateSpace_{\intPart}^{\bot}}
\newcommand{\compressChainTransTemplate}[2]{\mdpTrans^{#1}_{#2}}
\newcommand{\compressChainTrans}{\compressChainTransTemplate{\strat}{\intPart}}

\newcommand{\cisChainStrat}[1]{\ocChain^{#1}_{\intPartB}}
\newcommand{\cisChain}{\cisChainStrat{\strat}}

\newcommand{\cisChainStateSpace}{P_{\intPartB}}
\newcommand{\cisChainStateSpaceStar}{P_{\intPartB}^{\top}}
\newcommand{\cisChainTransTemplate}[2]{\ocTrans_{#1}^{#2}}
\newcommand{\cisChainTrans}{\cisChainTransTemplate{\intPartB}{\strat}}

\newcommand{\compressCis}{\mchain_{\intPartC}({\cisChain})}
\newcommand{\compressCisStateSpace}{\mdpStateSpace_{\intPartC}({\cisChainStateSpace})}
\newcommand{\compressCisTrans}{\ocTrans_{\intPartC}[{\cisChain}]}

\newcommand{\cisConfig}{\bar{\ocConfig}}
\newcommand{\cisConfigB}{\bar{\ocConfig}'}

\newcommand{\succHist}[2]{\histPart_{\mathsf{succ}}(#1, #2)}
\newcommand{\absHist}[1]{\histPart_{\mathsf{abs}}(#1)}

\newcommand{\termProbaVar}[2]{\langle#1\searrow{}#2\rangle}
\newcommand{\upProba}[5]{[(#1,#2)\nearrow{}(#3,#4)]_{#5}}
\newcommand{\downProba}[5]{[(#1,#2)\searrow{}(#3,#4)]_{#5}}
\newcommand{\upProbaVar}[5]{\langle (#1,#2)\nearrow{}(#3,#4)\rangle_{#5}}
\newcommand{\downProbaVar}[5]{\langle (#1,#2)\searrow{}(#3,#4)\rangle_{#5}}

\newcommand{\upHistSet}[5]{\histPart_{#5}((#1,#2)\nearrow{}(#3,#4))}
\newcommand{\downHistSet}[5]{\histPart_{#5}((#1,#2)\searrow{}(#3,#4))}

\newcommand{\upPart}[1]{U_{#1}}
\newcommand{\downPart}[1]{D_{#1}}
\newcommand{\mcHist}{\bar{\hist}}

\newcommand{\varTrans}{x}
\newcommand{\varTransTuple}{\mathbf{\varTrans}}
\newcommand{\solTrans}{\varTrans^\star}
\newcommand{\solTransTuple}{\mathbf{\varTrans}^{\star}}

\newcommand{\varObj}{y}
\newcommand{\varObjTuple}{\mathbf{\varObj}}
\newcommand{\solObj}{\varObj^{\star}}
\newcommand{\solObjTuple}{\mathbf{\varObj}^{\star}}

\newcommand{\varStrat}{z}
\newcommand{\varStratI}{\mathbf{\varStrat}^{\interval}}
\newcommand{\varStratIstar}{\mathbf{\varStrat}^{\interval_\intIndex^{\star}}}
\newcommand{\varStratIprime}{\mathbf{\varStrat}^{\interval'}}
\newcommand{\varStratTuple}{\mathbf{\varStrat}}
\newcommand{\solStrat}{\varStrat^{\star}}
\newcommand{\solStratTuple}{\mathbf{\varStrat}^{\star}}

\newcommand{\vectStratTuple}{\varStratTuple^\strat}

\newcommand{\varCis}{v}
\newcommand{\varCisTuple}{\mathbf{\varCis}}
\newcommand{\solCis}{\varCis^\star}
\newcommand{\solCisTuple}{\mathbf{\varCis}^\star}

\newcommand{\compressChainSymbolicVerbose}{\compressChainStrat{\stratB_{\varStratTuple}}}
\newcommand{\compressChainSymbolic}{\compressChainStrat{\varStratTuple}}
\newcommand{\compressChainSymSol}{\compressChainStrat{\solStratTuple}}

\newcommand{\compressChainTransSymbolicVerbose}{\compressChainTransTemplate{\stratB_{\varStratTuple}}{\intPart}}
\newcommand{\compressChainTransSymbolic}{\compressChainTransTemplate{\varStratTuple}{\intPart}}
\newcommand{\compressChainTransSymSol}{\compressChainTransTemplate{\solStratTuple}{\intPart}}

\newcommand{\cisChainSymbolic}{\cisChainStrat{\varStratTuple}}
\newcommand{\cisChainSymbolicVerbose}{\cisChainStrat{\stratB_{\varStratTuple}}}
\newcommand{\cisChainTransSymbolic}{\cisChainTransTemplate{\intPartB}{\stratB_{\varStratTuple}}}
\newcommand{\cisChainTupleSymbolic}{(\cisChainStateSpace, \cisChainTransSymbolic)}
\newcommand{\compressCisSymbolic}{\mchain_{\intPartC}({\cisChainSymbolic})}
\newcommand{\compressCisTransSymbolic}{\ocTrans_{\intPartC}[{\cisChainSymbolic}]}
\newcommand{\compressCisTupleSymbolic}{(\compressCisStateSpace, \compressCisTransSymbolic)}

\newcommand{\cisChainSymbSol}{\cisChainStrat{\solStratTuple}}
\newcommand{\compressCisSymbSol}{\mchain_{\intPartC}({\cisChainSymbSol})}
\newcommand{\cisChainTransSymbSol}{\cisChainTransTemplate{\intPartB}{\solStratTuple}}
\newcommand{\compressCisTransSymbSol}{\ocTrans_{\intPartC}[{\cisChainSymbSol}]}

\newcommand{\formulaTransBase}{\Phi_{\ocTrans}}
\newcommand{\formulaObjBase}{\Phi_{\objective}}
\newcommand{\formulaCisBase}{\Psi_{\ocTrans}}
\newcommand{\formulaStratBase}{\Phi_{\strat}}

\newcommand{\formulaTrans}{\formulaTransBase^\intPart}
\newcommand{\formulaTransI}{\formulaTransBase^\interval}
\newcommand{\formulaObj}{\formulaObjBase^\intPart}
\newcommand{\formulaCis}{\formulaCisBase^\intPartB}
\newcommand{\formulaCisI}{\formulaCisBase^\interval}

\newcommand{\formulaCisTrans}{\formulaTransBase^\intPartC}
\newcommand{\formulaCisObj}{\formulaObjBase^\intPartC}

\newcommand{\formulaStrat}{\formulaStratBase^{\intPart,\intPart'}}
\newcommand{\formulaStratCis}{\formulaStratBase^{\intPartB,\intPartB'}}

\newcommand{\suppBounded}{\mathcal{B}}
\newcommand{\formulaStratB}{\formulaStratBase^{\intPart,\intPart',\suppBounded}}
\newcommand{\formulaTransB}{\formulaTransBase^{\intPart,\intPart',\suppBounded}}
\newcommand{\formulaTransBI}{\formulaTransBase^{\interval,\intPart', \suppBounded}}
\newcommand{\formulaObjB}{\formulaObjBase^{\intPart,\intPart',\suppBounded}}

\newcommand{\chainX}{\ocmdp_{\sqsxVect}}

\newcommand{\eleError}[1]{\varepsilon_{#1}}
\newcommand{\eleB}[1]{\eleError{#1}^{(\sqsi)}}

\newcommand{\eleD}[1]{\eta^{(\sqsi)}_{#1}}
\newcommand{\seqD}{(\eleD{\counterUB})_{\counterUB\geq 2}}

\newcommand{\objective}{\Omega}

\newcommand{\reach}[1]{\mathsf{Reach}(#1)}
\newcommand{\reachVerb}[2]{\mathsf{Reach}_{#1}(#2)}

\newcommand{\target}{T}

\newcommand{\termination}{\mathsf{Term}}
\newcommand{\selectiveTermination}[1]{\termination({#1})}

\newcommand{\thresProba}{\theta}

\newcommand{\emptyword}{\varepsilon}

\newcommand{\mealy}{\mathfrak{M}}
\newcommand{\mealyStateSpace}{M}
\newcommand{\mealyState}{m}

\newcommand{\mealyStateInit}{\mealyState_\init}

\newcommand{\mealyUpdate}{\mathsf{up}}
\newcommand{\mealyNext}{\mathsf{nxt}}

\newcommand{\mealyUpdateHat}{\widehat{\mealyUpdate}}

\newcommand{\mealyTuple}{(\mealyStateSpace, \mealyStateInit, \mealyUpdate, \mealyNext)}

\newcommand{\stratGeneric}[1]{{\sigma_{#1}}}
\newcommand{\strat}{\stratGeneric{}}

\newcommand{\stratBGeneric}[1]{{\tau_{#1}}}
\newcommand{\stratB}{\stratBGeneric{}}

\newcommand\vartextvisiblespace[1][.5em]{\makebox[#1]{\kern.07em
    \vrule height.3ex
    \hrulefill
    \vrule height.3ex
    \kern.07em
  }}

\tikzset{
  >=stealth,
  left sided/.style={
    draw=none,
    append after command={
      [shorten <= -0.5\pgflinewidth]
      (\tikzlastnode.north west) edge[dashed](\tikzlastnode.south west)
    }
  },
  two sided/.style={
    draw=none,
    append after command={
      [shorten <= -0.5\pgflinewidth]
      (\tikzlastnode.north west) edge[dashed](\tikzlastnode.south west)
      (\tikzlastnode.north east) edge[dashed](\tikzlastnode.south east)
    }
  },
  right sided/.style={
    draw=none,
    append after command={
      [shorten <= -0.5\pgflinewidth]
      (\tikzlastnode.north east) edge[dashed](\tikzlastnode.south east)
    }
  }
}

\tikzstyle{stochasticc} = [fill, circle, minimum size=0.1cm, inner sep=0.05cm, outer sep=0cm]
\tikzstyle{stochastics} = [fill, rectangle, minimum size=0.1cm, inner sep=0.05cm, outer sep=0cm]

\makeatletter
\def\squarecorner#1{
\pgf@x=\the\wd\pgfnodeparttextbox \pgfmathsetlength\pgf@xc{\pgfkeysvalueof{/pgf/inner xsep}}\advance\pgf@x by 2\pgf@xc \pgfmathsetlength\pgf@xb{\pgfkeysvalueof{/pgf/minimum width}}\ifdim\pgf@x<\pgf@xb \pgf@x=\pgf@xb \fi \pgf@y=\ht\pgfnodeparttextbox \advance\pgf@y by\dp\pgfnodeparttextbox \pgfmathsetlength\pgf@yc{\pgfkeysvalueof{/pgf/inner ysep}}\advance\pgf@y by 2\pgf@yc \pgfmathsetlength\pgf@yb{\pgfkeysvalueof{/pgf/minimum height}}\ifdim\pgf@y<\pgf@yb \pgf@y=\pgf@yb \fi \ifdim\pgf@x<\pgf@y \pgf@x=\pgf@y \else
        \pgf@y=\pgf@x \fi
\pgf@x=#1.5\pgf@x \advance\pgf@x by.5\wd\pgfnodeparttextbox \pgfmathsetlength\pgf@xa{\pgfkeysvalueof{/pgf/outer xsep}}\advance\pgf@x by#1\pgf@xa \pgf@y=#1.5\pgf@y \advance\pgf@y by-.5\dp\pgfnodeparttextbox \advance\pgf@y by.5\ht\pgfnodeparttextbox \pgfmathsetlength\pgf@ya{\pgfkeysvalueof{/pgf/outer ysep}}\advance\pgf@y by#1\pgf@ya }
\makeatother

\pgfdeclareshape{square}{
    \savedanchor\northeast{\squarecorner{}}
    \savedanchor\southwest{\squarecorner{-}}

    \foreach \x in {east,west} \foreach \y in {north,mid,base,south} {
        \inheritanchor[from=rectangle]{\y\space\x}
    }
    \foreach \x in {east,west,north,mid,base,south,center,text} {
        \inheritanchor[from=rectangle]{\x}
    }
    \inheritanchorborder[from=rectangle]
    \inheritbackgroundpath[from=rectangle]
}

\newtheorem{assumption}[theorem]{Assumption}

\begin{document}
\maketitle

\begin{abstract}
  Markov decision processes (MDPs) are a canonical model to reason about decision making within a stochastic environment. We study a fundamental class of \textit{infinite} MDPs: one-counter MDPs (OC-MDPs). They extend finite MDPs via an associated counter taking natural values, thus inducing an infinite MDP over the set of configurations (current state and counter value). We consider two characteristic objectives: reaching a target state (state-reachability), and reaching a target state with counter value zero (selective termination). The synthesis problem for the latter is not known to be decidable and connected to major open problems in number theory. Furthermore, even seemingly simple strategies (e.g., memoryless ones) in OC-MDPs might be impossible to build in practice (due to the underlying infinite configuration space): we need finite, and preferably small, representations.

To overcome these obstacles, we introduce two natural classes of \textit{concisely represented} strategies based on a (possibly infinite) partition of counter values in intervals. For both classes, and both objectives, we study the verification problem (does a given strategy ensure a high enough probability for the objective?), and two synthesis problems (does there exist such a strategy?): one where the interval partition is fixed as input, and one where it is only parameterized. We develop a generic approach based on a compression of the induced infinite MDP that yields decidability in all cases, with all complexities within $\pspace$.
\end{abstract}

\section{Introduction}\label{section:introduction}
A \textit{Markov decision process} (MDP) models the continuous interaction of a controllable system with a stochastic environment.
MDPs are used notably in the fields of formal methods (e.g.,~\cite{BK08,gog23}) and reinforcement learning (e.g.,~\cite{SuttonB18}).
In each step of an MDP, the system selects an action that is available in the current state, and the state is updated randomly following a distribution depending only on the current state and chosen action.
This interaction goes on forever and yields a \textit{play} (i.e., an execution) of the MDP.
The resolution of non-determinism in an MDP is done via a \textit{strategy}: strategies prescribe choices depending on the current history of the play and may use randomisation.
Strategies are the formal counterpart of controllers for the modelled system~\cite{rECCS,DBLP:reference/mc/BloemCJ18}.
A strategy is \textit{pure} if it does not use randomisation.

A classical problem in MDPs is \textit{synthesis}, which asks to construct a good strategy with respect to a specification.
Specifications can be defined, e.g., via \textit{payoff functions} that quantify the quality of plays, and via \textit{objectives}, which are sets of well-behaved plays.
The goal is thus to synthesise a strategy that, in the former case, has a high expected payoff, and, in the latter case, enforces the objective with high probability.
If possible, it is desirable to construct an \textit{optimal strategy}, i.e., a strategy whose performance is not (strictly) exceeded by another.
The decision variant of the synthesis problem, which asks whether a well-performing strategy exists, is called \textit{realisability}.
We focus on variants of \textit{reachability objectives}, asking that a target set of states be visited. These canonical objectives are central in synthesis~\cite{DBLP:conf/fsttcs/BrihayeGMR23}.

\subparagraph*{Strategies and their representations.}
Traditionally, strategies are represented by \textit{Mealy machines}, i.e., finite automata with outputs (see, e.g.,~\cite{DBLP:journals/lmcs/BouyerRORV22,DBLP:journals/iandc/MainR24}).
Strategies that admit such a representation are called \textit{finite-memory strategies}.
A special subclass is that of \textit{memoryless strategies}, i.e., strategies whose decisions depend only on the current state.

In finite MDPs, pure memoryless strategies often suffice to play optimally with a single payoff or objective, e.g., for reachability~\cite{BK08} and parity objectives~\cite{DBLP:conf/soda/ChatterjeeJH04}, and mean~\cite{bierth1987expected} and discounted-sum payoffs~\cite{Sha53}.
In countable MDPs, for reachability objectives, pure memoryless strategies are also as powerful as general ones, though optimal strategies need not exist in general~\cite{Orn69,DBLP:conf/icalp/KieferMSTW20}.
In the multi-objective setting, memory and randomisation are necessary in general (e.g.,~\cite{DBLP:journals/fmsd/RandourRS17,DBLP:conf/fsttcs/BrihayeGMR23}).
Nonetheless, finite memory often suffices in finite MDPs, e.g., when considering several $\omega$-regular objectives~\cite{DBLP:journals/lmcs/EtessamiKVY08}.

The picture is more complicated in infinite-state MDPs, where even pure memoryless strategies need not admit finite representations.
Our contribution focuses on small (and particularly, finite) counter-based representations of memoryless strategies in a fundamental class of infinite-state MDPs: \emph{one-counter MDPs.}

\subparagraph*{One-counter Markov decision processes.}
\textit{One-counter MDPs} (OC-MDPs,~\cite{DBLP:conf/soda/BrazdilBEKW10}) are finite MDPs augmented with a counter that can be incremented (by one), decremented (by one) or left unchanged on each transition.\footnote{
  Considering such counter updates is not restrictive for modelling: any integer counter update can be obtained with several transitions.
  However, this impacts the complexity of decision problems.}
Such a finite OC-MDP induces an infinite MDP over a set of \textit{configurations} given by states of the underlying MDP and counter values.
In this induced MDP, any play that reaches counter value zero is interrupted; this event is called \textit{termination}.
We consider two variants of the model: \textit{unbounded OC-MDPs}, where counter values can grow arbitrarily large, and \textit{bounded OC-MDPs}, in which plays are interrupted when a fixed counter upper bound is reached.

The counter in OC-MDPs can, e.g., model resource consumption along plays~\cite{DBLP:conf/soda/BrazdilBEKW10}, or serve as an abstraction of unbounded data types and structures~\cite{DBLP:conf/cav/BrazdilKK11}. It can also model a passage of time: indeed, OC-MDPs generalise \textit{finite-horizon MDPs}, in which a bound is imposed on the number of steps (see, e.g.,~\cite{DBLP:conf/icalp/BalajiK0PS19}).
OC-MDPs can be also seen as an extension of \textit{one-counter Markov chains} with non-determinism.
One-counter Markov chains are equivalent to (discrete-time) quasi-birth-death processes~\cite{DBLP:journals/pe/EtessamiWY10}, a model studied in queuing theory.

Termination is the canonical objective in OC-MDPs~\cite{DBLP:conf/soda/BrazdilBEKW10}.
Also relevant is the more general \textit{selective termination} objective, which requires terminating in a target set of states.
In this work, we study both the selective termination objective and the \textit{state-reachability} objective, which requires visiting a target set of states regardless of the counter value.

Optimal strategies need not exist in unbounded OC-MDPs for these objectives~\cite{DBLP:journals/iandc/BrazdilBEK13}.
The general synthesis problem in unbounded OC-MDPs for selective termination is not known to be decidable, and it is at least as hard as the positivity problem for linear recurrence sequences~\cite{DBLP:journals/theoretics/PiribauerB24}, whose decidability would yield a major breakthrough in number theory~\cite{DBLP:conf/soda/OuaknineW14}.
Optimal strategies exist in bounded OC-MDPs: the induced MDP is finite and we consider reachability objectives.
However, constructing optimal strategies is already $\exptime$-hard for reachability in finite-horizon MDPs~\cite{DBLP:conf/icalp/BalajiK0PS19}.

In this work, we propose to tame the inherent complexity of analysing OC-MDPs by restricting our analysis to a class of succinctly representable (yet natural and expressive) strategies called \emph{interval strategies.}

\subparagraph*{Interval strategies.}
The monotonic structure of OC-MDPs makes them amenable to analysis through strategies of special structure. For instance, in their study of \emph{solvency games} (stateless variant of OC-MDPs with binary counter updates), Berger et al.~\cite{DBLP:conf/fsttcs/BergerKSV08} identify a natural class of \emph{rich man's strategies,} which, whenever the counter value is above some threshold, always select the same action in the same state. They argue that the question of existence of rich man's optimal strategies ``\emph{gets at a real phenomenon which \dots should be reflected in any good model of risk-averse investing.''} While~\cite{DBLP:conf/fsttcs/BergerKSV08} shows that optimal rich man's strategies do not always exist, if they do, their existence substantially simplifies the model analysis.

In this work, we make such finitely-represented strategies the primary object of study. We consider so-called \emph{interval} strategies: each such strategy is based on some (finite or infinite but finitely-representable) partitioning of \( \mathbb{N} \) into intervals, and the strategy's decision depends on the current state and on the partition containing the current counter value.

More precisely,
we focus on two classes of these strategies:
On the one hand, in bounded and unbounded OC-MDPs, we consider \textit{open-ended interval strategies} (OEISs): here, the underlying partitioning is finite, and thus contains an open-ended interval \( [n,\infty) \) on which the strategy behaves memorylessly, akin to rich man's strategies.

On the other hand, in unbounded OC-MDPs, we also consider \textit{cyclic interval strategies} (CISs): strategies for which there exists a (positive integer) period such that, for any two counter values that differ by the period, we take the same decisions.
A CIS can be represented similarly to an OEIS using a partition of the set of counter values up to the period.

We collectively refer to OEISs and CISs as \textit{interval strategies}.
While interval strategies are not sufficient to play optimally in unbounded OC-MDPs~\cite{DBLP:conf/fsttcs/BergerKSV08}, they can be used to approximate the supremum probability for the objectives we consider~\cite{DBLP:journals/iandc/BrazdilBEK13}. 

\subparagraph*{Our contributions.}
We formally introduce the class of interval strategies. We show that
OEISs in bounded OC-MDPs and CISs in unbounded OC-MDPs can be exponentially more concise than equivalent Mealy machines (Section~\ref{section:problems}), and unbounded OEISs can even represent infinite-memory strategies.

For selective termination and state reachability, we consider the interval strategy verification problem and two realisability problems for structure-constrained interval strategies.
On the one hand, the \textit{fixed-interval realisability} problem asks, given an interval partition, whether there is an interval strategy built on this partition that ensures the objective with a probability greater than a given threshold.
Intuitively, in this case, the system designer specifies the desired structure of the controller.
On the other hand, the \textit{parameterised realisability} problem (for interval strategies), asks whether there exists a well-performing strategy built on a partition of size no more than a parameter $\intNum$ such that no finite interval is larger than a second parameter $\intSize$.
We consider two variants of the realisability problem: one for checking the existence of a suitable pure strategy and another for randomised strategies.
Randomisation allows for better performance when imposing structural constraints on strategies (Example~\ref{example:pure vs random}), but pure strategies are however often preferred for synthesis~\cite{DBLP:conf/tacas/DelgrangeKQR20}.

\begin{table}
  \centering
  \bgroup
  \def\arraystretch{1.2}
  \scalebox{0.95}{
  \begin{tabular}{|c|c|c|c|c|c|c|}
    \hline
    Semantics
    & \multicolumn{2}{c|}{Bounded} & \multicolumn{4}{c|}{Unbounded}\\
    \hline
    \textit{Strategy type}
    & \multicolumn{2}{c|}{\textit{Open-ended}} & \multicolumn{2}{c|}{\textit{Open-ended}} & \multicolumn{2}{c|}{\textit{Cyclic}} \\
    \hline
    \hline
    \multirow{4}{9em}{\centering Verification}
    & \multicolumn{2}{c|}{$\ptime^\posSLP$} & \multicolumn{2}{c|}{$\coetr$} & \multicolumn{2}{c|}{$\coetr$} \\
    & \multicolumn{2}{c|}{Thm.~\ref{theorem:verification:bounded}} & \multicolumn{2}{c|}{Thm.~\ref{theorem:verification:oeis:complexity}} & \multicolumn{2}{c|}{Thm.~\ref{theorem:verification:cis:complexity}} \\
    \cdashline{2-7}
    & \multicolumn{2}{c|}{sqrt-sum-hard} & \multicolumn{4}{c|}{\multirow{2}{*}{sqrt-sum-hard~\cite{DBLP:journals/pe/EtessamiWY10}}} \\
    & \multicolumn{2}{c|}{Thm.~\ref{verification:hardness:bounded}} & \multicolumn{4}{c|}{}\\
    \hline
    \multirow{4}{9em}{\centering Realisability (both fixed-interval and parameterised)}
    & \textit{Pure} & \textit{Random} & \textit{Pure} & \textit{Random} & \textit{Pure} & \textit{Random} \\
    & $\np^\posSLP$  & $\np ^\etr$ & $\np^\etr$  & $\pspace$  & $\np^\etr$  & $\pspace$   \\
    & Thm.~\ref{theorem:realisability:bounded:pure}  & Thm.~\ref{theorem:realisability:bounded:randomised} & Thm.~\ref{theorem:realisability:oeis:pure}  & Thm.~\ref{theorem:realisability:oeis:randomised}  & Thm.~\ref{theorem:realisability:cis:pure}  & Thm.~\ref{theorem:realisability:cis:randomised}  \\
    \cdashline{2-7}
    & \multicolumn{6}{c|}{$\np$-hard (termination, Thm.~\ref{theorem:realisability:np-hardness}) and sqrt-sum-hard (\cite{DBLP:journals/pe/EtessamiWY10}, Thm.~\ref{verification:hardness:bounded})} \\
    \hline
  \end{tabular}
  }
  \egroup
  \caption{Complexity bound summary for our problems.
    All bounds are below $\pspace$.
    Square-root-sum-hardness results are derived from instances of the form $\sum\sqrt{\sqsx_\sqsi}\geq\sqsy$.
  }\label{table:complexity}
\end{table}

Our complexity results are summarised in Table~\ref{table:complexity}.
Analysing the performance of a memoryless strategy amounts to studying the Markov chain it induces on the MDP.
Our results rely on the analysis of a \textit{compressed Markov chain} derived from the (potentially infinite) Markov chain induced by an interval strategy.
We remove certain configurations and aggregate several transitions into one (Section~\ref{section:abstraction}).
This compressed Markov chain preserves the probability of selective termination and of hitting counter upper bounds (Theorem~\ref{theorem:ocmdp:probability matching}).
However, its transition probabilities may require exponential-size representations or even be irrational.
To represent these probabilities, we characterise them as the least solutions of quadratic systems of equations (Theorems~\ref{theorem:equations:termination} and~\ref{theorem:equations:transitions}), similarly to the result of~\cite{DBLP:journals/lmcs/KuceraEM06} for termination probabilities in probabilistic pushdown automata.
Compressed Markov chains are finite for OEISs and are induced by a one-counter Markov chain for CISs (Section~\ref{section:abstraction:finiteness}).

The crux of our algorithmic results is the aforementioned compression. For verification, we reduce the problem to checking the validity of a universal formula in the theory of the reals, by exploiting our characterisation of transition probabilities in compressed Markov chains. This induces a $\pspace$ upper bound. For bounded OC-MDPs, we can do better: verification can be solved in polynomial time in the unit-cost arithmetic RAM model of computation of Blum, Shub and Smale~\cite{BSS1989}, by computing transition and reachability probabilities of the compressed Markov chain. This yields a $\ptime^\posSLP$ complexity in the Turing model (see~\cite{DBLP:journals/siamcomp/AllenderBKM09}).

Both realisability variants exploit the verification approach through the theory of the reals. For fixed-interval realisability for pure strategies, we exploit non-determinism to select good strategies and then verify them with the above. In the randomised case, in essence, we build on the verification formulae and existentially quantify over the probabilities of actions under the sought strategy. Finally, for parameterised realisability, we build on our algorithms for the fixed-interval case by first non-deterministically building an appropriate partition.

We also provide complexity lower bounds.
We show that all of our considered problems are hard for the square-root-sum problem, a problem that is not known to be solvable in polynomial time but that is solvable in polynomial time in the Blum-Shub-Smale model~\cite{DBLP:journals/jc/Tiwari92}.
We also prove $\np$-hardness for our realisability problems for selective termination, already when checking the existence of good single-interval strategies.

\subparagraph*{Impact.} Our results provide a natural class of strategies for which realisability is decidable (whereas the general case remains open and known to be difficult~\cite{DBLP:journals/theoretics/PiribauerB24}), and with arguably low complexity (for synthesis). Furthermore, the class of interval strategies is of practical interest thanks to their concise representation and their inherently understandable structure (in contrast to the corresponding Mealy machine representation).

\subparagraph*{Related work.} In addition to the main references cited previously, we mention some (non-exhaustive) related work. The closest is~\cite{DBLP:conf/cav/BlahoudekB0OTT20}: interval strategies are similar to counter selector strategies, studied in \textit{consumption} MDPs. These differ from OC-MDPs in key aspects: all transitions consume resources (i.e., bear negative weights), and recharge can only be done in special reload states, where it is considered as an atomic action up to a given capacity. Consumption and counter-based (or energy) models have different behaviors (e.g.,~\cite{DBLP:conf/cav/BrazdilCKN12}). The authors of~\cite{DBLP:conf/cav/BlahoudekB0OTT20} also study incomparable objectives: almost-sure (i.e., probability one) B\"uchi objectives.

Restricting oneself to subclasses of strategies that prove to be of practical interest is a common endeavor in synthesis: e.g., strategies represented by decision trees~\cite{DBLP:conf/cav/BrazdilCCFK15,DBLP:conf/tacas/BrazdilCKT18,DBLP:conf/tacas/AshokJKWWY21,DBLP:journals/sttt/JungermannKW23}, pure strategies~\cite{DBLP:conf/stacs/Gimbert07,DBLP:conf/tacas/DelgrangeKQR20,DBLP:journals/lmcs/BouyerORV23}, finite-memory strategies~\cite{DBLP:journals/acta/ChatterjeeRR14,DBLP:conf/fsttcs/BouyerRV22,DBLP:journals/lmcs/BouyerORV23}, or strategies using limited forms of randomness~\cite{DBLP:journals/iandc/MainR24,MR25}.

\subparagraph*{Outline.}
In Section~\ref{section:preliminaries}, we introduce all prerequisite notions.
We define interval strategies and formalise our decision problems in Section~\ref{section:problems}.
The compressed Markov chain construction, which is central in our verification and realisability algorithms, is presented in Section~\ref{section:abstraction}.
Section~\ref{section:verification} focuses on the verification problem via compressed Markov chains.
We study the fixed-interval and parameterised realisability problems in Section~\ref{section:realisability}.
Finally, we present our complexity lower bounds in Section~\ref{section:hardness}.
Appendix~\ref{appendix:square root sum} presents additional details used in the correctness proof of our reduction from the square-root-sum-problem to the verification problem in bounded OC-MDPs.

\section{Preliminaries}\label{section:preliminaries}
\subparagraph*{Set and probability theory.} We write $\IN$, $\IQ$ and $\IR$ for the sets of non-negative integers, rational numbers and real numbers respectively.
We let  $\INpos$ denote the set of positive integers.
We let $\INbar = \IN\cup\{\infty\}$ and $\INposBar = \INpos\cup\{\infty\}$.
Given $n, n'\in\INbar$, we let $\integerInterval{n, n'}$ denote the set $\{\ocCount\in\IN\mid n\leq\ocCount\leq n'\}$, and if $n=0$, we shorten the notation to $\integerInterval{n'}$.

Let $A$ be a finite or countable set.
We write $\dist{A}$ for the set of distributions over $A$, i.e., of functions $\mu\colon A\to [0, 1]$ such that $\sum_{a\in A}\mu(a) = 1$.
Given a distribution $\mu\in\dist{A}$, we let $\supp{\mu} = \{a\in A\mid \mu(a) > 0\}$ denote the support of $\mu$.

\subparagraph*{Markov decision processes.}
Formally, a \textit{Markov decision process} (MDP) is a tuple $\mdp = \mdpTuple$ where $\mdpStateSpace$ is a countable set of states, $\mdpActionSpace$ is a finite set of actions and $\mdpTrans\colon \mdpStateSpace\times \mdpActionSpace\to \dist{\mdpStateSpace}$ is a (partial) probabilistic transition function.
Given $\mdpState\in\mdpStateSpace$, we write $\mdpActionSpace(\mdpState)$ for the set of actions $\mdpAction\in\mdpActionSpace$ such that $\mdpTrans(\mdpState,\mdpAction)$ is defined.
We say that an action $\mdpAction$ is enabled in $\mdpState$ if $\mdpAction\in\mdpActionSpace(\mdpState)$.
We require, without loss of generality, that there are no deadlocks in MDPs, i.e., that for all states $\mdpState$, $\mdpActionSpace(\mdpState)$ is non-empty.
We say that $\mdp$ is \textit{finite} if its state space is finite.

A \textit{Markov chain} can be seen as an MDP with a single enabled action per state.
We omit actions in Markov chains and denote them as tuples $\mchain = (\mdpStateSpace, \mdpTrans)$ with $\mdpTrans\colon\mdpStateSpace\to\dist{\mdpStateSpace}$.

Let $\mdp = \mdpTuple$.
A \textit{play} of $\mdp$ is an infinite sequence $\play=\mdpState_0\mdpAction_0\mdpState_1\mdpAction_1\ldots\in (\mdpStateSpace\mdpActionSpace)^\omega$ such that for all $\indexPosition\in\IN$, $\mdpState_{\indexPosition+1}\in \supp{\mdpTrans(\mdpState_\indexPosition, \mdpAction_\indexPosition)}$.
A \textit{history} is a finite prefix of a play ending in a state.
We let $\playSet{\mdp}$ and $\histSet{\mdp}$ respectively denote the set of plays and of histories of $\mdp$.
Given a play $\play=\mdpState_0\mdpAction_0\mdpState_1\mdpAction_1\ldots$ and $\indexLast\in\IN$, we let $\playPrefix{\play}{\indexLast} = \mdpState_0\mdpAction_0\ldots \mdpAction_{\indexLast-1}\mdpState_\indexLast$ and extend this notation to histories.
For any given history $\hist = \mdpState_0\mdpAction_0\ldots \mdpAction_{\indexLast-1}\mdpState_\indexLast$, we let $\first{\hist} = \mdpState_0$ and $\last{\hist} = \mdpState_\indexLast$.
Given two histories $\hist_1 = \mdpState_0\mdpAction_0\ldots \mdpAction_{\indexLast-1}\mdpState_\indexPosition$ and $\hist_2 = \mdpState_\indexPosition\mdpAction_\indexPosition\ldots \mdpAction_{\indexLast-1}\mdpState_\indexLast$ such that $\last{\hist_1} = \first{\hist_2}$, we let $\histConcat{\hist_1}{\hist_2} = \mdpState_0\mdpAction_0\ldots\mdpAction_{\indexLast-1}\mdpState_\indexLast$ denote their concatenation without repetition of state $\mdpState_\indexPosition$; we abusively call $\histConcat{\hist_1}{\hist_2}$ the \textit{concatenation} of $\hist_1$ and $\hist_2$.
The \textit{cylinder} of a history $\hist = \mdpState_0\mdpAction_0\ldots\mdpAction_{\indexLast-1}\mdpState_\indexLast$ is the set $\cyl{\hist} = \{\play\in\playSet{\mdp}\mid\playPrefix{\play}{\indexLast}=\hist\}$ of plays extending $\hist$.

Let $\histPart$ be a set of histories.
We let $\cyl{\histPart}=\bigcup_{\hist\in\histPart}\cyl{\hist}$.
We say that $\histPart$ is \textit{prefix-free} if there are no distinct histories $\hist,\hist'\in\histPart$ such that $\hist$ is a prefix of $\hist'$.
Whenever $\histPart$ is prefix-free, the union $\cyl{\histPart}=\bigcup_{\hist\in\histPart}\cyl{\hist}$ is disjoint.

In a Markov chain, due to the lack of actions, plays and histories are simply sequences of states coherent with transitions.
We extend the notation introduced above for plays and histories of MDPs to plays and histories of Markov chains.

\subparagraph*{Strategies.}
A strategy describes the action choices made throughout a play of an MDP.
In general, strategies can use randomisation and all previous knowledge to make decisions.
Formally, a \textit{strategy} is a function $\strat\colon\histSet{\mdp}\to\dist{\mdpActionSpace}$ such that for any history $\hist\in\histSet{\mdp}$, $\supp{\strat{}(\hist)}\subseteq \mdpActionSpace(\last{\hist})$, i.e., a strategy can only prescribe actions that are available in the last state of a history.
A play $\play=\mdpState_0\mdpAction_0\mdpState_1\ldots$ is \textit{consistent} with a strategy $\strat$ if for all $\indexPosition\in\IN$, $\strat(\playPrefix{\play}{\indexPosition})(\mdpAction_\indexPosition)>0$.
Consistency of a history with respect to a strategy is defined similarly.

A strategy is \textit{pure} if it does not use any randomisation, i.e., if it maps histories to Dirac distributions.
We view pure strategies as functions $\strat{}\colon\histSet{\mdp}\to\mdpActionSpace$.
A strategy $\strat$ is \textit{memoryless} if the distribution it suggests depends only on the last state of a history, i.e., if for all histories $\hist$ and $\hist'$, $\last{\hist} = \last{\hist'}$ implies $\strat{}(\hist) = \strat{}(\hist')$.
We view memoryless strategies and pure memoryless strategies as functions $\mdpStateSpace\to \dist{\mdpActionSpace}$ and $\mdpStateSpace\to\mdpActionSpace$ respectively.

Fix a strategy $\strat{}$ and an initial state $\mdpState_\init\in\mdpStateSpace$.
We obtain a purely stochastic process when following $\strat{}$ from $\mdpState_\init$, i.e., a countable Markov chain over the set of histories with the initial state $\mdpState_\init$.
Rather than defining this Markov chain, we instead define a distribution, denoted by $\probaGverb{\mdp}{\mdpState_\init}{\strat}$, over the $\sigma$-algebra induced by cylinder sets.
For all histories $\hist = \mdpState_0\mdpAction_0\mdpState_1\ldots\mdpState_\indexLast\in\histSet{\mdp}$, if $\mdpState_0 = \mdpState_\init$, we let
\(\probaGverb{\mdp}{\mdpState_\init}{\strat}(\cyl{\hist}) =
  \prod_{\indexPosition= 0}^{\indexLast-1}
  \strat(\playPrefix{\hist}{\indexPosition})(\mdpAction_\indexPosition)\cdot
  \mdpTrans(\mdpState_\indexPosition, \mdpAction_\indexPosition)(\mdpState_{\indexPosition+1}),\)
and otherwise, if $\mdpState_0\neq\mdpState_\init$, we let $\probaGverb{\mdp}{\mdpState_\init}{\strat}(\cyl{\hist}) = 0$.
Carath\'{e}odory's extension theorem~\cite[Thm.~A.1.3.]{Dur19} guarantees that this probability distribution extends in a unique fashion to the $\sigma$-algebra induced by cylinder sets.
For a Markov chain $\mchain$, due to the absence of non-determinism, we drop the strategy from the notation and write $\proba_{\mchain, \mdpState_\init}$.
If the relevant MDP or Markov chain is clear from the context, we omit them from the notation above.

In the following, we focus on memoryless strategies.
If $\strat$ is a memoryless strategy of $\mdp$, we consider the \textit{Markov chain induced by $\strat$ on $\mdp$} as a Markov chain over $\mdpStateSpace$.
Formally, it is defined as the Markov chain $(\mdpStateSpace, \mdpTrans')$ such that for all $\mdpState$, $\mdpState'\in\mdpStateSpace$, $\mdpTrans'(\mdpState)(\mdpState') = \sum_{\mdpAction\in\mdpActionSpace(\mdpState)}\strat(\mdpState)(\mdpAction)\cdot \mdpTrans(\mdpState, \mdpAction)(\mdpState')$.

\subparagraph*{Finite-memory strategies.}
In general, strategies can use unlimited knowledge of the past and thus may not admit a finite representation.
The classical finite representation for strategies is based on Mealy machines, i.e., finite automata with outputs.
At each step of a play, an action is chosen according to the outputs of the Mealy machine and the state of the Mealy machine is updated according to the current memory state, the current MDP state and chosen action.
A strategy is \textit{finite-memory} if it can be represented by a Mealy machine.

Let $\mdp=\mdpTuple$ be an MDP.
Formally, a \textit{Mealy machine} is a tuple $\mealy = \mealyTuple$ where $\mealyStateSpace$ is a finite set of memory states, $\mealyStateInit\in\mealyStateSpace$ is an initial memory state, $\mealyUpdate\colon\mealyStateSpace\times\mdpStateSpace\times\mdpActionSpace\to\mealyStateSpace$ is a memory update function and $\mealyNext\colon\mealyStateSpace\times\mdpStateSpace\to\dist{\mdpActionSpace}$ is a next-move function.
We formalise the finite-memory strategy $\strat^\mealy$ induced by a Mealy machine $\mealy=\mealyTuple$ as follows.
First, we define the iterated update function $\mealyUpdateHat\colon(\mdpStateSpace\mdpActionSpace)^*\to\mealyStateSpace$ by induction.
We let $\mealyUpdateHat(\emptyword) = \mealyStateInit$ (where $\emptyword$ denotes the empty word), and let, for all $u\mdpState\mdpAction\in(\mdpStateSpace\mdpActionSpace)^*$, $\mealyUpdateHat(u\mdpState\mdpAction) = \mealyUpdate(\mealyUpdateHat(u), \mdpState, \mdpAction)$.
The strategy $\strat^\mealy$ is then defined, for all histories $\hist = u\mdpState\in\histSet{\mdp}$, by $\strat^\mealy(\hist) = \mealyNext(\mealyUpdateHat(u), \mdpState)$.
Memoryless strategies are induced by one-state Mealy machines.

\begin{remark}
  There exist more general definitions of finite-memory strategies, e.g., where the considered Mealy machines have a randomised initialisation and a randomisation update function.
  However, these models are not equivalent in general (see, e.g.,~\cite{DBLP:journals/iandc/MainR24}.
  The model presented here is sufficient to compare interval strategies introduced in Section~\ref{section:problems} to strategies that rely on memory instead of observing counter values.
\end{remark}

\subparagraph*{One-counter Markov decision processes.}
One-counter MDPs (OC-MDPs) extend MDPs with a counter that can be incremented, decremented or left unchanged on each transition.
Formally, a \textit{one-counter MDP} is a tuple $\ocmdp = \ocTuple$ where $\ocTupleWeightless$ is a finite MDP and $\weight\colon \mdpStateSpace\times\mdpActionSpace\to \{-1, 0, 1\}$ is a (partial) weight function that assigns an integer weight from \(\{-1,0,1 \} \) to state-action pairs.
For all $\ocState\in\ocStateSpace$ and $\ocAction\in\ocActionSpace$, we require that $\weight(\ocState, \ocAction)$ be defined whenever $\ocAction\in\ocActionSpace(\ocState)$.
A \textit{configuration} of $\ocmdp$ is a pair $\ocConfigPair$ where $\ocState\in\ocStateSpace$ and $\ocCount\in\IN$.
In the sequel, by plays, histories, strategies etc.~of $\ocmdp$, we refer to the corresponding notion with respect to the MDP $\ocTupleWeightless$ underlying $\ocmdp$.

We remark that the weight function is not subject to randomisation, i.e., the weight of any transition is determined by the outgoing state and chosen action.
In particular, counter values can be inferred from histories and be taken in account to make decisions in $\ocmdp$.

Let $\ocmdp = \ocTuple$ be an OC-MDP.
The OC-MDP $\ocmdp$ induces an MDP over the infinite countable space of configurations.
Transitions in this induced MDP are defined using $\ocTrans$ for the probability of updating the state and $\weight$ for the deterministic change of counter value.
We interrupt any play whenever a configuration with counter value $0$ is reached.
Intuitively, such configurations can be seen as situations in which we have run out of an energy resource.
We may also impose an upper bound on the counter value, and interrupt any plays that reach this counter upper bound. We refer to OC-MDPs with a finite upper bound on counter values as \textit{bounded OC-MDPs} and OC-MDPs with no upper bounds as \textit{unbounded OC-MDP}.

We provide a unified definition for both semantics.
Let $\ocmdp = \ocTuple$ be an OC-MDP and let $\counterUB\in\INposBar$ be an upper bound on the counter value.
We define the MDP $\ocmdpFin{\ocmdp}{\counterUB} = (\ocStateSpace \times \integerInterval{\counterUB}, \ocActionSpace, \ocTransFin)$ where $\ocTransFin$ is defined, for all configurations $\ocConfig = \ocConfigPair\in\ocStateSpace\times\integerInterval{\counterUB}$, actions $\ocAction\in\ocActionSpace(\ocState)$ and states $\ocStateB\in\ocStateSpace$, by $\ocTransFin(\ocConfig, \ocAction)(\ocStateB, \ocCount + \weight(\ocState, \ocAction)) = \ocTrans(\ocState, \ocAction)(\ocStateB)$ if $\ocCount\notin\{0, \counterUB\}$, and $\ocTransFin(\ocConfig, \ocAction)(\ocConfig) = 1$ otherwise.
The state space of $\ocmdpFin{\ocmdp}{\counterUB}$ is finite if and only if $\counterUB\neq\infty$.

For technical reasons, we also introduce one-counter Markov chains.
In one-counter Markov chains, we authorise stochastic counter updates, i.e., counter updates are integrated in the transition function.
This contrasts with OC-MDPs where deterministic counter updates are used to allow strategies to observe counter updates.
We only require the unbounded semantics in this case.
Formally, a \textit{one-counter Markov chain} is a tuple $\ocChain = \ocChainTuple$ where $\ocStateSpace$ is a finite set of states and $\ocTrans\colon\ocStateSpace\to\dist{\ocStateSpace\times\{-1, 0, 1\}}$ is a probabilistic transition and counter update function.
The one-counter Markov chain $\ocChain$ induces a Markov chain $\ocChainFin{\ocChain}{\infty}=(\ocStateSpace\times\IN, \ocTransInf)$ such that for any configuration $\ocConfig = (\ocState, \ocCount)\in\ocStateSpace\times\IN$, any $\ocStateB\in\ocStateSpace$ and any $\weightVal\in\{-1, 0, 1\}$, we have $\ocTransInf(\ocConfig)((\ocStateB, \ocCount+\weightVal)) = \ocTrans(\ocState)(\ocStateB, \weightVal)$ if $\ocCount \neq 0$ and $\ocTransInf(\ocConfig)(\ocConfig)=1$ otherwise.

\subparagraph*{Objectives.}
Let $\mdp = \mdpTuple$ be an MDP.
An \textit{objective} is a measurable set of plays that represents a specification.
We focus on variants of reachability.
Given a set of target states $\target\subseteq \mdpStateSpace$, the \textit{reachability} objective for $\target$ is defined as $\reach{\target} = \{\mdpState_0\mdpAction_0\mdpState_1\mdpAction_1\ldots\in\playSet{\mdp}\mid \exists\,\indexPosition\in\IN,\, \mdpState_\indexPosition\in\target\}$.
If $\target=\{\mdpState\}$, we write $\reach{\mdpState}$ for $\reach{\{\mdpState\}}$.

Let $\ocmdp = \ocTuple$ be an OC-MDP, $\counterUB\in\INposBar$ be a counter upper bound and $\target\subseteq\ocStateSpace$.
We consider two variants of reachability on $\ocmdpFin{\ocmdp}{\counterUB}$.
The first goal we consider requires visiting $\target$ regardless of the counter value and the second requires visiting $\target$ with a counter value of zero.
The first objective is called \textit{state-reachability} and we denote it by $\reach{\target}$ instead of $\reach{\target\times\integerInterval{\counterUB}}$.
The second objective is called \textit{selective termination}, which is formalised as $\selectiveTermination{\target} = \reach{\target\times\{0\}}$.
Like above, if $\target=\{\ocState\}$, we write $\reach{\ocState}$ and $\selectiveTermination{\ocState}$ instead of $\reach{\{\ocState\}}$ and $\selectiveTermination{\{\ocState\}}$.

\subparagraph*{Reachability probabilities in Markov chains.}
Let $\mchain = \mchainTuple$ be a finite Markov chain and let $\target\subseteq\mdpStateSpace$ be a set of targets.
The probabilities $(\probaMC{\mdpState}(\reach{\target}))_{\mdpState\in\mdpStateSpace}$ are the least solution of a system of linear equations (e.g.,~\cite[Thm.~10.15]{BK08}).
We recall this system.

Let $\{\mdpStateSpace_{=0}, \mdpStateSpace_{=1}, \mdpStateSpace_{?}\}$ be a partition of $\mdpStateSpace$ such that $\mdpStateSpace_{=0} \subseteq\{\mdpState\in\mdpStateSpace\mid \probaMC{\mdpState}(\reach{\target}) = 0\}$ and $\target\subseteq\mdpStateSpace_{=1} \subseteq\{\mdpState\in\mdpStateSpace\mid \probaMC{\mdpState}(\reach{\target}) = 1\}$.
We consider the system defined by $x_\mdpState=0$ for all $\mdpState\in\mdpStateSpace_{=0}$, $x_\mdpState=1$ for all $\mdpState\in\mdpStateSpace_{=1}$ and $x_\mdpState = \sum_{\mdpState'\in\mdpStateSpace}\mdpTrans(\mdpState)(\mdpState')\cdot x_{\mdpState'}$ for all $\mdpState\in\mdpStateSpace_{?}$.
The least non-negative solution of this system is obtained by letting $x_\mdpState=\probaMC{\mdpState}(\reach{\target})$ for all $\mdpState\in\mdpStateSpace$.
Furthermore, this system has a unique solution when $\mdpStateSpace_{=0} =\{\mdpState\in\mdpStateSpace\mid \probaMC{\mdpState}(\reach{\target}) = 0\}$.
The set $\{\mdpState\in\mdpStateSpace\mid \probaMC{\mdpState}(\reach{\target}) = 0\}$ only depends on the topology of the Markov chains (i.e., which states are connected by a transition) and not the transition probabilities; a state is in this set if and only if there are no histories starting in this state and ending in $\target$.

\subparagraph*{The Blum-Shub-Smale model of computation.}
Some of our complexity bounds rely on a model of computation introduced by Blum, Shub and Smale~\cite{BSS1989}.
A Blum-Shub-Smale (BSS) machine, intuitively, is a random-access memory machine with registers storing real numbers.
Arithmetic computations on the content of registers are in constant time in this model.

The class of decision problem that can be solved in polynomial time in the BSS model coincides with the class of decision problems that can be solved, in the Turing model, in polynomial time with a $\posSLP$ oracle~\cite{DBLP:journals/siamcomp/AllenderBKM09}.
The $\posSLP$ problem asks, given a division-free straight-line program (intuitively, an arithmetic circuit), whether its output is positive.
The $\posSLP$ problem lies in the counting hierarchy and can be solved in polynomial space~\cite{DBLP:journals/siamcomp/AllenderBKM09}.

\subparagraph*{Theory of the reals.}
The theory of the reals refers to the set of sentences (i.e., fully quantified first-order logical formulae) in the signature of ordered fields that hold in $\IR$.
The problem of deciding whether a sentence is in the theory of the reals is decidable; if the number of quantifier blocks is fixed, this can be done in $\pspace$~\cite[Rmk.~13.10]{BPR2006}.
Furthermore, the problem of deciding the validity of an existential (resp.~universal) formula is $\np$-hard (resp.~$\coNP$-hard)~\cite[Rmk.~13.9]{BPR2006}.
Our complexity bounds refer to the complexity classes $\etr$ (existential theory of the reals) and $\coetr$, which contain the problems that can be reduced in polynomial time to checking the membership of an \textit{existential sentence} and \textit{universal sentence} in the theory of the reals respectively.

\section{Interval strategies and decision problems}\label{section:problems}
We now introduce interval strategies.
Definitions are given in Section~\ref{section:problems:definition}.
We compare interval strategies to finite-memory strategies of the MDP underlying the OC-MDP in Section~\ref{section:problems:representation}.
We close this section by formulating our decision problems in Section~\ref{section:problems:statements}.

We fix an OC-MDP $\ocmdp = \ocTuple$ and a counter upper bound $\counterUB\in\INposBar$ for this whole section.

\subsection{Interval strategies}\label{section:problems:definition}
Interval strategies are a subclass of memoryless strategies of $\ocmdpFin{\ocmdp}{\counterUB}$ that admit finite compact representations based on families of intervals.
Intuitively, a strategy is an interval strategy if there exists an \textit{interval partition} (i.e., a partition containing only intervals) of the set of counter values such that decisions taken in a state for two counter values in the same interval coincide.
We require that this partition has a finite representation to formulate verification and synthesis problems for interval strategies.

The set $\integerInterval{1, \counterUB-1}$ contains all counter values for which decisions are relevant.
Let $\intPart$ be an interval partition of $\integerInterval{1, \counterUB-1}$.
A memoryless strategy $\strat$ of $\ocmdpFin{\ocmdp}{\counterUB}$ is \textit{based on the partition $\intPart$} if for all $\ocState\in\ocStateSpace$, all $\interval\in\intPart$ and all $\ocCount, \ocCount'\in\interval$, we have $\strat(\ocState, \ocCount) = \strat(\ocState, \ocCount')$.
All memoryless strategies are based on the trivial partition of $\integerInterval{1, \counterUB-1}$ into singleton sets.
In practice, we are interested in strategies based on partitions with a small number of large intervals.

We define two classes of interval strategies: strategies that are based on finite partitions and, in unbounded OC-MDPs, strategies that are based on periodic partitions.
An interval partition $\intPart$ of $\INpos$ is \textit{periodic} if there exists a period $\period\in\INpos$ such that for all $\interval\in\intPart$, $\interval+\period = \{\ocCount+\period\mid\ocCount\in\interval\}\in\intPart$.
A periodic interval partition $\intPart$ with period $\period$ \textit{induces} the interval partition $\intPartB = \{\interval\in\intPart\mid\interval\subseteq\integerInterval{1, \period}\}$ of $\integerInterval{1, \period}$.
Conversely, for any $\period\in\INpos$, an interval partition $\intPartB$ of $\integerInterval{1, \period}$ \textit{generates} the periodic partition $\intPart = \{\interval+\ocCount\cdot\period\mid\interval\in\intPartB, \ocCount\in\IN\}$.

Let $\strat$ be a memoryless strategy of $\ocmdpFin{\ocmdp}{\counterUB}$.
The strategy $\strat$ is an \textit{open-ended interval strategy} (OEIS) if it is based on a finite interval partition of $\integerInterval{1, \counterUB-1}$.
The qualifier open-ended follows from there being an unbounded interval in any finite interval partition of $\INpos$ (for the unbounded case).
When $\counterUB=\infty$, $\strat$ is a \textit{cyclic interval strategy} (CIS) if there exists a period $\period\in\INpos$ such that for all $\ocState\in\ocStateSpace$ and all $\ocCount\in\INpos$, we have $\strat(\ocState, \ocCount) = \strat(\ocState, \ocCount+\period)$.
A strategy is a CIS with period $\period$ if and only if it is based on a periodic interval partition of $\INpos$ with period $\period$.\begin{remark}
  We do not consider strategies based on ultimately periodic interval partitions.
  However, our techniques can be adapted to analyse such strategies.
  The compressed Markov chain construction of Section~\ref{section:abstraction} can be defined for any memoryless strategy of $\ocmdpFin{\ocmdp}{\counterUB}$.
  By combining our approaches for the analysis of OEISs and CISs via compressed Markov chains, we can analyse such strategies.
  Furthermore, we can show that our complexity bounds for the decision problems we study extend to these strategies.\hfill$\lhd$
\end{remark}

We represent interval strategies as follows.
First, assume that $\strat$ is an OEIS and let $\intPart$ be the coarsest finite interval partition of $\integerInterval{1, \counterUB-1}$ on which $\strat$ is based.
We can represent $\strat$ by a table that lists, for each $\interval\in\intPart$, the bounds of $\interval$ and a memoryless strategy of $\ocmdp$ dictating the choices to be made when the current counter value lies in $\interval$.
Next, assume that $\counterUB=\infty$ and that $\strat$ is a CIS with period $\period$.
Let $\intPartB$ be an interval partition of $\integerInterval{1, \period}$ such that $\strat$ is based on the partition $\intPart$ generated by $\intPartB$.
We represent $\strat$ by $\period$ and an OEIS of $\ocmdpFin{\ocmdp}{\period+1}$ based on $\intPartB$ that specifies the behaviour of $\strat$ for counter values up to $\period$.

\begin{remark}
  In practice, in the bounded setting, it is not necessary to encode the counter upper bound $\counterUB$ in the representation of an OEIS; it is implicit from $\ocmdpFin{\ocmdp}{\counterUB}$.
  Nonetheless, we assume that $\counterUB$ is part of the strategy representation for the sake of convenience: this allows us to treat all bounded intervals uniformly in complexity analyses, as though all interval bounds are in the encoding of considered OEIS.
  This has no impact on our complexity results, as $\counterUB$ is part of the description of $\ocmdpFin{\ocmdp}{\counterUB}$.
  \hfill$\lhd$
\end{remark}

Interval strategies subsume \textit{counter-oblivious strategies}, i.e., memoryless strategies that make choices based only on the state and disregard the current counter value.
Counter-oblivious strategies can be viewed as memoryless strategies $\strat\colon\ocStateSpace\to\dist{\ocActionSpace}$ of $\ocmdp$.

\subsection{Substituting counters with memory}\label{section:problems:representation}
Memoryless strategies of $\ocmdpFin{\ocmdp}{\counterUB}$ can be seen as strategies of $\ocmdp$ when an initial counter value is fixed.
The idea is to use memory to keep track of the counter instead of observing it.
We can thus compare our representation of interval strategies to the classical Mealy machine representation of the corresponding strategies of $\ocmdp$.

We first formalise how to derive a strategy of $\ocmdp$ from a memoryless strategy of $\ocmdpFin{\ocmdp}{\counterUB}$ and an initial counter value.
Let $\ocCount_\init\in\integerInterval{1, \counterUB-1}$ be an initial counter value and let $\strat$ be a memoryless strategy of $\ocmdpFin{\ocmdp}{\counterUB}$.
We build a (partially-defined) strategy $\stratB_{\ocCount_\init}^{\strat}$ of $\ocmdp$ from $\strat$.
Intuitively, instead of having the counter value as an input of the strategy, we store the current counter value in memory.
For any history $\hist_\ocmdp = \ocState_1\ocAction_1\ldots\ocAction_{\indexLast-1}\ocState_\indexLast\in\histSet{\ocmdp}$, we define $\weight(\hist_\ocmdp) = \sum_{\indexPosition=0}^{\indexLast-1}\weight(\ocState_\indexPosition, \ocAction_\indexPosition)$.
The strategy $\stratB_{\ocCount_\init}^{\strat}$ is defined for any history $\hist_\ocmdp\in\histSet{\ocmdp}$ by $\stratB_{\ocCount_\init}^{\strat}(\hist_\ocmdp) = \strat((\last{\hist_\ocmdp}, \ocCount_\init+\weight(\hist_\ocmdp)))$ when $\ocCount_\init+\weight(\hist_\ocmdp)\in\integerInterval{1, \counterUB-1}$ and is left undefined otherwise.

Similarly, the state-reachability and selective termination objectives of $\ocmdpFin{\ocmdp}{\counterUB}$ can be translated into objectives of $\ocmdp$ when an initial counter value $\ocCount_\init$ is specified. 
We can show that the probability in $\ocmdpFin{\ocmdp}{\counterUB}$ of such an objective under $\strat$ from a configuration $(\ocState, \ocCount_\init)$ matches the probability of the counterpart objective in $\ocmdp$ under the strategy $\stratB^{\strat}_{\ocCount_\init}$ from $\ocState$.

If $\counterUB\in\IN$, the counterpart in $\ocmdp$ of any OEIS has finite memory: there are only finitely many counter values.
Assume that $\counterUB\in\IN$ and let $\strat$ be an OEIS.
The strategy $\stratB^{\strat}_{\ocCount_\init}$ is induced by the following Mealy machine: we let $\mealy = \mealyTuple$ where $\mealyStateSpace = \integerInterval{1, \counterUB -1}$, $\mealyStateInit = \ocCount_\init$, and for all $\ocState\in\ocStateSpace$, $\ocCount\in\mealyStateSpace$ and $\ocAction\in\ocActionSpace$, we let $\mealyUpdate(\ocCount, \ocState, \ocAction) = \ocCount + \weight(\ocState, \ocAction)$ and $\mealyNext(\ocCount, \ocState)(\ocAction) = \strat((\ocState, \ocCount))$.
Intuitively, $\mealy$ induces $\stratB^{\strat}_{\ocCount_\init}$ because it keeps track of the weight of the current history and this weight determines the choice prescribed by $\stratB^{\strat}_{\ocCount_\init}$.

This construction yields a Mealy machine whose size is exponential in the binary encoding size of $\counterUB$.
The following example illustrates that such exponential-size Mealy machines may be required, even for OEISs based on the partition $\{\{1\}, \integerInterval{2, \counterUB-1}\}$, i.e., OEISs that only have to distinguish $1$ from other counter values.
Additionally, this same example shows that the counterpart in $\ocmdp$ of an OEIS may require infinite memory in the unbounded setting.

\begin{example}\label{example:interval vs memory}
  We consider the OC-MDP $\ocmdp$ illustrated in Figure~\ref{figure:ocmdp:memory}.
  Let $\counterUB\in\INposBar$, $\counterUB \geq 3$.
  We consider the OEIS $\strat$ defined by $\strat(\ocState_1, \ocCount) = \ocAction$ for all $\ocCount\in\integerInterval{2, \counterUB-1}$ and $\strat(\ocState_1, 1) = \ocActionB$.
  The strategy $\strat$ maximises the probability of terminating in $\ocState_2$ from the configuration $(\ocState_0, 1)$.
  \begin{figure}
    \centering
        \begin{tikzpicture}
      \node[state, align=center] (q0) {$\ocState_0$};
      \node[stochasticc, right = of q0] (q0a) {};
      \node[state, right = of q0a] (q1) {$\ocState_1$};
      \node[state, right = of q1] (q2) {$\ocState_2$};

      \path[-] (q0) edge node[above] {$\ocAction\mid 1$} (q0a);
      \path[->] (q0a) edge[bend left] node[below] {$\frac{1}{2}$} (q0);
      \path[->] (q0a) edge node[above] {$\frac{1}{2}$} (q1);
      \path[->] (q1) edge node[above] {$\ocActionB\mid -1$} (q2);
      \path[->] (q1) edge[loop above] node[above] {$\ocAction\mid -1$} (q1);
      \path[->] (q2) edge[loop above] node[above] {$\ocAction\mid 0$} (q2);
    \end{tikzpicture}
    \caption{An OC-MDP.
      Edge splits following actions indicate probabilistic transitions.
      We indicate the weight of a state-action pair next to each action.
    }\label{figure:ocmdp:memory}
  \end{figure}
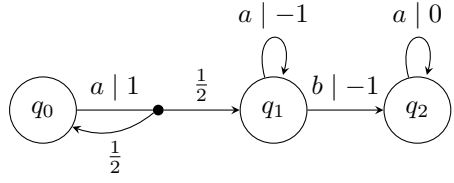

  We let $\stratB^{\strat}_1$ denote the counterpart of $\strat$ in $\ocmdp$ for the initial counter value $1$.
  We show that for all $2\leq\ocCount<\counterUB$, a Mealy machine inducing $\stratB^{\strat}_1$ must have at least $\ocCount$ states.
  In particular, if $\counterUB$ is finite, it means that any such Mealy machine must have $\counterUB-1$ states, and if $\counterUB$ is infinite, it means that there is no (finite) Mealy machine inducing $\stratB^\strat_1$.
  
  Let $2\leq \ocCount < \counterUB$.
  We proceed by contradiction.
  Assume that there exists a Mealy machine $\mealy = \mealyTuple$ inducing $\stratB^\strat_1$ such that $|\mealyStateSpace|\leq\ocCount-1$.
  For all $\indexPosition\in\integerInterval{\ocCount-1}$, we let $\mealyState_\indexPosition = \mealyUpdateHat((\ocState_0\ocAction)^{\ocCount-1}(\ocState_1\ocAction)^\indexPosition)$ and let $\hist_\indexPosition = (\ocState_0\ocAction)^{\ocCount-1}(\ocState_1\ocAction)^{\indexPosition}\ocState_1$.
  For all $\indexPosition\in\integerInterval{\ocCount-1}$, $\hist_\indexPosition$ is a history of $\ocmdp$ in the domain of $\stratB^\strat_1$ (because $\ocCount < \counterUB$ and the initial counter value is $1$) that is consistent with $\stratB^\strat_1$.
  Since $\mealy$ induces $\stratB^{\strat}_1$, it follows that for all $\indexPosition\in\integerInterval{\ocCount-2}$, we have $\stratB^{\strat}_1(\hist_\indexPosition) = \mealyNext(\mealyState_\indexPosition, \ocState_1) = \ocAction$ and that $\stratB^{\strat}_1(\hist_{\ocCount-1}) = \mealyNext(\mealyState_{\ocCount-1}, \ocState_1) = \ocActionB$.
  However, $\mealyState_{\ocCount-1}$ must occur more than once in the sequence $(\mealyState_\indexPosition)_{\indexPosition\in\integerInterval{\ocCount-1}}$ because this sequence is obtained by repeating the update rule $\mealyUpdate(\cdot, \ocState_1, \ocAction)$ from $\mealyState_0$ and there are no more than $\ocCount-1$ memory states in $\mealy$.
  This is a contradiction.
  \hfill$\lhd$
\end{example}

We now assume that $\counterUB=\infty$.
The counterpart in $\ocmdp$ of any CIS is a finite-memory strategy: it suffices to keep track of the remainder of division of the counter value by a period.
Let $\strat$ be a CIS, $\period$ be a period of $\strat$ and $\ocCount_\init$ be an initial counter value.
Formally, the strategy $\stratB^{\strat}_{\ocCount_\init}$ is induced by the Mealy machine $\mealy = \mealyTuple$ defined as follows.
We let $\mealyStateSpace = \integerInterval{\period -1}$ and $\mealyStateInit = \ocCount_\init\bmod\period$.
Updates are defined, for all $\ocState\in\ocStateSpace$, $\ocCount\in\mealyStateSpace$ and $\ocAction\in\ocActionSpace$, by $\mealyUpdate(\ocCount, \ocState, \ocAction) = (\ocCount + \weight(\ocState, \ocAction))\bmod\period$.
The next-move function is defined differently following whether the memory state is zero or not.
We let, for all $\ocState\in\ocStateSpace$, $\ocCount\in\mealyStateSpace$ and $\ocAction\in\ocActionSpace$, $\mealyNext(\ocCount, \ocState)(\ocAction) = \strat((\ocState, \ocCount))$ if $\ocCount\neq 0$ and $\mealyNext(0, \ocState)(\ocAction) = \strat((\ocState, \period))$.
Intuitively, $\mealy$ induces $\stratB^\strat_{\ocCount_\init}$ because  the remainder of the current counter value for its division by the period is sufficient to mimic $\strat$.

By adapting Example~\ref{example:interval vs memory}, we can show that a CIS representation may be exponentially more succinct than any Mealy machine for its counterpart in $\ocmdp$.

\subsection{Decision problems}\label{section:problems:statements}

We formulate the decision problems studied in following sections.
The common inputs of these problems are an OC-MDP $\ocmdp$ with rational transition probabilities, a counter bound $\counterUB\in\INposBar$ (encoded in binary if it is finite), a target $\target\subseteq\ocStateSpace$, an objective $\objective\in\{\reach{\target}, \selectiveTermination{\target}\}$, an initial configuration $\ocConfig_\init = (\ocState_\init, \ocCount_\init)$ and a threshold $\thresProba\in\ccInt{0}{1}\cap\IQ$ against which we compare the probability of $\objective$.
Problems that are related to CISs assume that $\counterUB=\infty$ as all strategies in the bounded case are OEISs.
We lighten the probability notation below by omitting $\ocmdpFin{\ocmdp}{\counterUB}$ from it.

First, we are concerned with the verification of interval strategies.
The \textit{interval strategy verification problem} asks to decide, given an interval strategy $\strat$, whether $\probaG{\ocConfig_\init}{\strat}(\objective)\geq\thresProba$.
We assume that the encoding of the input interval strategy matches the description of Section~\ref{section:problems:definition}.

The other two problems relate to interval strategy synthesis.
The corresponding decision problem is called \textit{realisability}.
We provide algorithms checking the existence of well-performing structurally-constrained interval strategies.
We formulate two variants of this problem.

For the first variant, we fix an interval partition $\intPart$ of $\integerInterval{1, \counterUB-1}$ beforehand and ask to check if there is a good strategy based on $\intPart$.
Formally, the \textit{fixed-interval OEIS realisability problem} asks whether, given a finite interval partition $\intPart$ of $\integerInterval{1, \counterUB-1}$, there exists an OEIS $\strat$ based on $\intPart$ such that $\probaG{\ocConfig_\init}{\strat}(\objective)\geq\thresProba$.
The variant for CISs is defined similarly: the \textit{fixed-interval CIS realisability problem} asks whether, given a period $\period\in\INpos$ and an interval partition $\intPartB$ of $\integerInterval{1, \period}$, there exists a CIS $\strat$ based on the partition generated by $\intPartB$ such that $\probaG{\ocConfig_\init}{\strat}(\objective)\geq\thresProba$.

For the second variant, we parameterise the number of intervals in the partition and the size of bounded intervals.
The  \textit{parameterised OEISs realisability problem} asks, given a bound $\intNum\in\INpos$ on the number of intervals and a bound $\intSize\in\INpos$ on the size of bounded intervals, whether there exists an OEIS $\strat$ such that $\probaG{\ocConfig_\init}{\strat}(\objective)\geq\thresProba$ and $\strat$ is based on an interval partition $\intPart$ of $\integerInterval{1, \counterUB-1}$ with $|\intPart|\leq\intNum$ and, for all bounded $\interval\in\intPart$, $|\interval|\in\integerInterval{1, \intSize}$ (in particular, the parameter $\intSize$ does not constrain the required infinite interval in the unbounded setting $\counterUB=\infty$).
The  \textit{parameterised CIS realisability problem} asks, given a bound  $\intNum\in\INpos$ on the number of intervals and a bound $\intSize\in\INpos$ on the size of intervals, whether there exists a CIS $\strat$ such that $\probaG{\ocConfig_\init}{\strat}(\objective)\geq\thresProba$ and $\strat$ is based on an interval partition $\intPart$ of $\INpos$ with period $\period$ such that $|\interval|\leq \intSize$ for all $\interval\in\intPart$ and $\intPart$ induces a partition of $\integerInterval{1, \period}$ with at most $\intNum$ intervals.
In both cases, we assume that the number $\intNum$ is given in unary.
This ensures that witness strategies, when they exist, are based on interval partitions that have a representation of size polynomial in the size of the inputs.

\begin{remark}\label{remark:oeis realisability:trivial parameters}
  In bounded OC-MDPs, instances of the parameterised OEIS realisability problem such that no partitions are compatible with the input parameters $\intNum$ and $\intSize$ are trivially negative.
  If $\counterUB\in\IN$, then there are no interval partitions of $\integerInterval{1, \counterUB-1}$ compatible with $\intNum$ and $\intSize$ whenever $\counterUB-1>\intNum\cdot\intSize$.
  This issue does not arise for OEISs in unbounded OC-MDPs or for CISs, as counter-oblivious strategies are always possible witnesses no matter the  parameters.
  \hfill$\lhd$
\end{remark}

For both realisability problems, we consider two variants, depending on whether we want the answer with respect to the set of \textit{pure} or \textit{randomised} interval strategies.
For many objectives in MDPs (e.g., reachability, parity objectives~\cite{BK08,DBLP:journals/lmcs/BouyerORV23}), the maximum probability of satisfying the objective is the same for pure and randomised strategies.
As highlighted by the following example, when we restrict the structure of the sought interval strategy, there may exist randomised strategies that perform better than all pure ones.
Intuitively, randomisation can somewhat alleviate the loss in flexibility caused by structural restrictions~\cite{DBLP:journals/acta/ChatterjeeRR14}.

\begin{example}\label{example:pure vs random}
  The fixed-interval and parameterised realisability problems subsume the realisability problem for counter-oblivious strategies.
  The goal of this example is to provide an OC-MDP in which there exists a randomised counter-oblivious strategy that performs better than any pure counter-oblivious strategy from a given configuration.

  We consider the OC-MDP $\ocmdp$ depicted in Figure~\ref{figure:example:pure vs random} in which all weights are $-1$, the objective $\reach{\ocStateC_\top}=\selectiveTermination{\ocStateC_\top}$, a counter bound $\counterUB\geq 3$ and the initial configuration $(\ocState, 2)$.
  To maximise the probability of reaching $\ocStateC_\top$ from $(\ocState, 2)$ with an unrestricted strategy, we must select action $\ocAction$ in $(\ocState, 2)$ and then $\ocActionB$ in $(\ocState, 1)$.
  Intuitively, $\ocAction$ is used when two steps remain as it allows another attempt and we switch to $\ocActionB$ when one step is left because it is the action that maximises the probability of reaching $\ocStateC_\top$ in a single step.

  \begin{figure}
    \centering
    \begin{tikzpicture}
      \node[state, align=center] (q) {$\ocState$};
      \node[right = of q] (mid) {};
      \node[stochasticc, node distance=4mm, below = of mid] (qa) {};
      \node[stochasticc, node distance=4mm, above = of mid] (qb) {};
      \node[state, right = of qb] (bot) {$\ocStateC_\bot$};
      \node[state, right = of qa] (top) {$\ocStateC_\top$};

      \path[-] (q) edge node[above] {$\ocAction$} (qa);
      \path[-] (q) edge node[above] {$\ocActionB$} (qb);
      \path[->] (qa) edge[bend left] node[below] {$\frac{1}{2}$} (q);
      \path[->] (qa) edge node[below] {$\frac{1}{2}$} (top);
      \path[->] (qb) edge node[above] {$\frac{1}{4}$} (bot);
      \path[->] (qb) edge node[below] {$\frac{3}{4}$} (top);
      \path[->] (top) edge[loop right] node[right] {$\ocAction$} (top);
      \path[->] (bot) edge[loop right] node[right] {$\ocAction$} (bot);
    \end{tikzpicture}
    \caption{An OC-MDP where all weights are $-1$ and are omitted from the figure.
    }\label{figure:example:pure vs random}
  \end{figure}
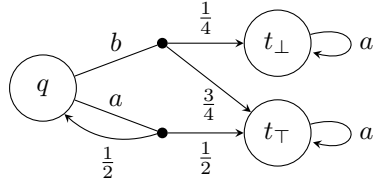

  We now limit our attention to counter-oblivious strategies.
  For pure strategies, no matter whether action $\ocAction$ or $\ocActionB$ is chosen in $\ocState$, $\ocStateC_\top$ is reached with probability $\frac{3}{4}$ from $(\ocState, 2)$.
  However, when playing both actions uniformly at random in $\ocState$, the resulting reachability probability from $(\ocState, 2)$ is $\frac{25}{32} > \frac{3}{4}$.
  This shows that randomised counter-oblivious strategies can achieve better reachability (resp.~selective termination) probabilities than pure strategies.
  \hfill$\lhd$
\end{example}

\section{Compressing induced Markov chains}\label{section:abstraction}
This section introduces the construction underlying our algorithms for the decision problems presented in Section~\ref{section:problems}.
To analyse the (possibly infinite) Markov chains induced by memoryless strategies over the space of configurations of an OC-MDP, we build a \textit{compressed Markov chain}, i.e., a Markov chain with fewer configurations and adjusted transitions.
A compressed Markov chain is defined with respect to a memoryless strategy and an interval partition on which the strategy is based.
This construction is generic, in the sense that it can formally be defined not only for interval strategies, but for all memoryless strategies.

We introduce the compression construction in Section~\ref{section:abstraction:definition}.
Section~\ref{section:compression:refinement} explains how to efficiently enforce a technical property on interval partitions such that compressed Markov chains are well-defined.
In Section~\ref{section:abstraction:validation}, we prove that termination probabilities and the probability of hitting a counter upper bound are preserved following compression.
For algorithmic reasons, we show that the transition probabilities of the compressed Markov chain can be represented as solutions of systems of quadratic equations in Section~\ref{section:abstraction:transitions}.
Although compressed Markov chains are finite for OEISs, they are not for CISs; we close the section by proving that compressed Markov chains for CISs are induced by one-counter Markov chains.

For the whole section, we fix an OC-MDP $\ocmdp = \ocTuple$, a bound $\counterUB\in\INposBar$ on counter values and a memoryless strategy $\strat$ of $\ocmdpFin{\ocmdp}{\counterUB}$ based on an interval partition $\intPart$ of $\integerInterval{1,\counterUB-1}$.

\subsection{Defining compressed Markov chains}\label{section:abstraction:definition}
We introduce the \textit{compressed Markov chain} $\compressChainVerbose$ derived from the Markov chain induced on $\ocmdpFin{\ocmdp}{\counterUB}$ by $\strat$ and the partition $\intPart$.
We write $\compressChain$ instead of $\compressChainVerbose$ whenever $\ocmdp$ is clear from the context. Intuitively, we keep some configurations in the state space of $\compressChain$ and, to balance this, transitions of $\compressChain$ aggregate several histories of the induced Markov chain.
We also apply compression for bounded intervals: interval bounds are represented in binary and thus the size of an interval is exponential in its encoding size.
We open with an example.

\begin{example}\label{example:compression:main}
  We consider the OC-MDP depicted in Figure~\ref{figure:example:compression:main:ocmdp} and counter upper bound $\counterUB=+\infty$.
  Let $\strat$ denote the randomised OEIS based on $\intPart = \{\integerInterval{1, 7}, \integerInterval{8, \infty}\}$ such that $\strat(\ocState, 1)(\ocAction) = \strat(\ocStateB, 1)(\ocAction) =\strat(\ocState, 8)(\ocActionC) = 1$ and $\strat(\ocStateB, 8)(\ocAction) = \strat(\ocStateB, 8)(\ocActionB) = \frac{1}{2}$.
  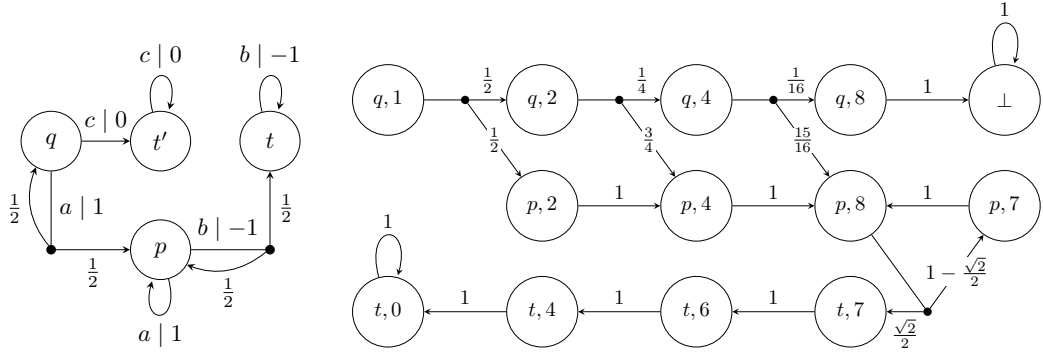
\begin{figure}
    \begin{subfigure}[t]{0.31\textwidth}
      \centering
      \scalebox{0.9}{
        \begin{tikzpicture}[every state/.style={on grid},node distance=1.6cm]
          \node[state, align=center] (q) {$\ocState$};
          \node[on grid,stochasticc, below = of q] (qmid) {};
          \node[state, right = of qmid] (p) {$\ocStateB$};
          \node[on grid,stochasticc, right = of p] (pmid) {};
          \node[state, above = of pmid] (t) {$\ocStateC$};
          \node[state, right = of q] (bot) {$\ocStateD$};
          
          \path[->] (q) edge node[above] {$\ocActionC\mid 0$} (bot);
          \path[-] (q) edge node[right] {$\ocAction\mid 1$} (qmid);
          \path[->] (qmid) edge[bend left] node[left] {$\frac{1}{2}$} (q);
          \path[->] (qmid) edge node[below] {$\frac{1}{2}$} (p);
          \path[-] (p) edge node[above] {$\ocActionB\mid -1$} (pmid);
          \path[->] (pmid) edge node[right] {$\frac{1}{2}$} (t);
          \path[->] (pmid) edge[bend left] node[below] {$\frac{1}{2}$} (p);
          \path[->] (p) edge[loop below] node[below] {$\ocAction\mid 1$} (p);
          \path[->] (t) edge[loop above] node[above] {$\ocActionB\mid -1$} (t);
          \path[->] (bot) edge[loop above] node[above] {$\ocActionC\mid 0$} (bot);
        \end{tikzpicture}
      }
      \caption{An OC-MDP. Weights are written next to actions.}\label{figure:example:compression:main:ocmdp}
    \end{subfigure}\hfill
    \begin{subfigure}[t]{0.67\textwidth}
      \centering
      \scalebox{0.78}{
        \begin{tikzpicture}[every state/.style = {on grid,minimum size = 12mm},node distance=1.3cm]
\node[state] (q1) {$\ocState, 1$};
          \node[on grid,stochasticc, right= of q1] (qmid1) {};
          \node[state, right = of qmid1] (q2) {$\ocState, 2$};
          \node[on grid,stochasticc, right= of q2] (qmid2) {};
          \node[state, right = of qmid2] (q4) {$\ocState, 4$};
          \node[on grid,stochasticc, right= of q4] (qmid4) {};
          \node[state, right = of qmid4] (q8) {$\ocState, 8$};
          \node[on grid,white,stochasticc, right= of q8] (qmid8) {};
          \node[state, right = of qmid8] (bot) {$\bot$};
\node[state, below = 1.8cm of q2] (p2) {$\ocStateB, 2$};
          \node[on grid,white,stochasticc, right= of p2] (pmid2) {};
          \node[state, right = of pmid2] (p4) {$\ocStateB, 4$};
          \node[on grid,white,stochasticc, right= of p4] (pmid4) {};
          \node[state, right = of pmid4] (p8) {$\ocStateB, 8$};
\node[state, below = 1.8cm of p8] (t7) {$\ocStateC, 7$};
          \node[on grid,white,stochasticc, left= of t7] (tmid7) {};
          \node[on grid,white,stochasticc, right= of p8] (p7anchor) {};
          \node[state, right = of p7anchor] (p7) {$\ocStateB, 7$};
          \node[on grid,stochasticc, right= of t7] (pmid8) {};
          \node[state, left = of tmid7] (t6) {$\ocStateC, 6$};
          \node[on grid,white,stochasticc, left= of t6] (tmid6) {};
          \node[state, left = of tmid6] (t4) {$\ocStateC, 4$};
          \node[on grid, white,stochasticc, left= of t4] (tmid4) {};
          \node[state, left = of tmid4] (t0) {$\ocStateC, 0$};

          \path[-] (q1) edge (qmid1);
          \path[->] (qmid1) edge node[above] {$\frac{1}{2}$} (q2);
          \path[->] (qmid1) edge node[inner sep=0, fill=white] {$\frac{1}{2}$} (p2);
          
          \path[-] (q2) edge (qmid2);
          \path[->] (qmid2) edge node[above] {$\frac{1}{4}$} (q4);
          \path[->] (qmid2) edge node[inner sep=0, fill=white] {$\frac{3}{4}$} (p4);

          \path[-] (q4) edge (qmid4);
          \path[->] (qmid4) edge node[above]  {$\frac{1}{16}$} (q8);
          \path[->] (qmid4) edge node[inner sep=0, fill=white]  {$\frac{15}{16}$} (p8);

          \path[->] (q8) edge node[above] {$1$} (bot);

          \path[->] (p2) edge node[above] {$1$} (p4);
          \path[->] (p4) edge node[above] {$1$} (p8);
          \path[->] (p7) edge node[above] {$1$} (p8);

          \path[-] (p8) edge (pmid8);
          \path[->] (pmid8) edge node[inner sep=0,fill=white] {$1-\frac{\sqrt{2}}{2}$} (p7);
          \path[->] (pmid8) edge node[below] {$\frac{\sqrt{2}}{2}$} (t7);

          \path[->] (t7) edge node[above] {$1$} (t6);
          \path[->] (t6) edge node[above] {$1$} (t4);
          \path[->] (t4) edge node[above] {$1$} (t0);

          \path[->] (t0) edge[loop above] node[above] {$1$} (t0);
          \path[->] (bot) edge[loop above] node[above] {$1$} (bot);
        \end{tikzpicture}
      }
      \caption{Fragment of the compressed Markov chain of Ex.~\ref{example:compression:main} reachable from $(\ocState, 1)$.
        Configuration parentheses are omitted to lighten the figure.
      }\label{figure:example:compression:main:chain}
    \end{subfigure}
    \caption{An illustration of an OC-MDP and its compression for a specific strategy.}
  \end{figure}

  We define the compressed Markov chain $\compressChain$ depicted in Figure~\ref{figure:example:compression:main:chain} by processing each interval individually.
  First, we consider the bounded interval $\interval = \integerInterval{1, 7}$.
  When we enter $\interval$ from its minimum or maximum, we only consider counter jumps by powers of $2$, starting with $1=2^0$.
  If a counter value in $\interval$ is reached by jumping by $2^\powerMax$, we consider counter updates of $2^{\powerMax+1}$ from it; Figure~\ref{figure:mc:counter update scheme} illustrates this counter update rule.
  This explains the counter progressions from $(\ocState, 1)$ to $(\ocState, 8)$ and from $(\ocStateC, 7)$ to $(\ocStateC, 0)$.
  The state space of $\compressChain$ with respect to $\interval$ contains the configurations whose counter values can be reached via this scheme.
  Transitions aggregate several histories of $\ocmdpFin{\ocmdp}{\infty}$, e.g., the probability from $\ocConfig=(\ocState, 2)$ to $\ocConfig' =(\ocStateB, 4)$ is the probability under $\strat$ of all histories of $\ocmdpFin{\ocmdp}{\infty}$ from $\ocConfig$ to $\ocConfig'$ along which counter values elsewhere than in $\ocConfig'$ remain between $\min\interval=1$ and $3$ (i.e., the counter value before the next step).
  The encoding of transition probabilities may be exponential in the number of retained configurations; this is highlighted by the progression of probability denominators between $(\ocState, 1)$ and $(\ocState, 8)$.

  \begin{figure}
  \centering
  \begin{tikzpicture}
    \node[dotted, state, align=center] (q0) {$0$};
    \node[state, right = of q0] (q1) {$1$};
    \node[state, right = of q1] (q2) {$2$};
    \node[state, right = of q2] (q4) {$4$};
    \node[state, right = of q4] (q6) {$6$};
    \node[state, right = of q6] (q7) {$7$};
    \node[dotted, state, right = of q7] (q8) {$8$};

    \path[->] (q1) edge (q0);
    \path[->] (q1) edge (q2);
    \path[->] (q2) edge (q4);
    \path[->, bend left] (q2) edge (q0);
    \path[->, bend left] (q4) edge (q0);
    \path[->] (q7) edge (q8);
    \path[->] (q7) edge (q6);
    \path[->] (q6) edge (q4);
    \path[->,bend right] (q6) edge (q8);
    \path[->, bend right] (q4) edge (q8);
  \end{tikzpicture}
  \caption{An illustration of the counter update scheme in a compressed Markov chain for the interval $\integerInterval{1, 7}$.}\label{figure:mc:counter update scheme}
\end{figure}
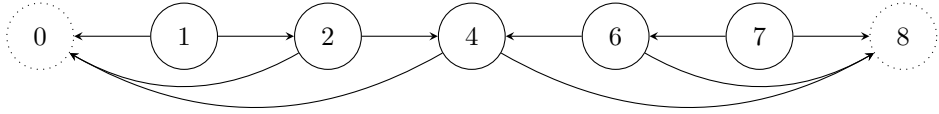

  For the unbounded interval $\interval = \integerInterval{8, \infty}$, we only consider configurations with counter value $\min\interval=8$ and consider transitions to configurations with counter value $\min\interval-1=7$.
  In this case, for instance, the transition probability from $(\ocStateB, 8)$ to $(\ocStateB, 7)$, corresponds the probability under $\strat$ in $\ocmdpFin{\ocmdp}{\infty}$ of hitting counter value $7$ for the first time in $\ocStateB$ from $(\ocStateB, 8)$.
  This example illustrates that this probability can be irrational.
  Here, the probability of moving from $(\ocStateB, 8)$ to $(\ocStateB, 7)$ is a solution of the quadratic equation $\varTrans = \frac{1}{4}+\frac{1}{2}\varTrans^2$ (see~\cite{DBLP:journals/lmcs/KuceraEM06}): $\frac{1}{4}$ is the probability of directly moving from $(\ocStateB, 8)$ to $(\ocStateB, 7)$ and $\frac{1}{2}\varTrans^2$ is the probability of moving from $(\ocStateB, 8)$ to $(\ocStateB, 7)$ by first going through the intermediate configuration $(\ocStateB, 9)$.

  Finally, we comment on the absorbing state $\bot$.
  The rules making up transitions of $\compressChain$ outlined above require a change in counter value.
  We redirect the probability of never seeing such a change to $\bot$.
  In this example, $\strat$ does not allow a counter decrease from $(\ocState, 8)$.
  \hfill$\lhd$
\end{example}

Example~\ref{example:compression:main} outlines the main ideas to construct $\compressChain$.
To ensure that compressed Markov chains are well-defined, we impose the following assumption on $\intPart$ which guarantees that, in general, bounded intervals of $\intPart$ can only be entered by one of their bounds.
\begin{assumption}\label{assumption:interval size}
  For all bounded $\interval\in\intPart$, $\log_2(|\interval|+1)\in\IN$, i.e., $|\interval| = 2^{\powerMax_{\interval}} - 1$ for some $\powerMax_\interval\in\IN$.
\end{assumption}
Assumption~\ref{assumption:interval size} is not prohibitive: we prove in Section~\ref{section:compression:refinement} that, given a bounded interval, we can partition it into sub-intervals satisfying the required size constraint in polynomial time.
We assume that Assumption~\ref{assumption:interval size} is satisfied for $\intPart$ for the remainder of the section.

We now formalise $\compressChain = (\compressChainStateSpace, \compressChainTrans)$.
We start by defining its state space $\compressChainStateSpace$ which does not depend on $\strat$.
We first formalise the configurations that are retained for each interval.

Let $\interval\in\intPart$.
First, let us assume that $\interval$ is unbounded and let $\intLB_\interval$ denote its minimum.
We set $\compressChainStateSpaceJ = \ocStateSpace\times\{\intLB_\interval\}$, i.e., we only retain the configurations with minimal counter value in $\interval$.

Next, let us assume that $\interval$ is bounded and of the form $\integerInterval{\intLB_\interval, \intUB_\interval}$.
Let $\powerMax_\interval=\log_2(|\interval|+1)$ (this is an integer by Assumption~\ref{assumption:interval size}).
We retain the counter values that can be reached via a generalisation of the scheme depicted in Figure~\ref{figure:mc:counter update scheme}.
The set of retained configurations for $\interval$ is given by
\[ \compressChainStateSpaceJ = \ocStateSpace \times \left(
  \{\intLB_\interval + 2^{\powerIndex} - 1 \mid
  \powerIndex\in\integerInterval{\powerMax_\interval-1}\}\cup
  \{\intUB_\interval - (2^{\powerIndex} - 1) \mid
  \powerIndex\in\integerInterval{\powerMax_\interval-1}\}
\right).\]

Finally, we consider absorbing configurations and the new state $\bot$.
We let $\compressChainStateSpaceStar = \{\bot\}\cup (\ocStateSpace\times\{0, \counterUB\})$ if $\counterUB\in\IN$ and $\compressChainStateSpaceStar = \{\bot\}\cup (\ocStateSpace\times\{0\})$ otherwise.
We define $\compressChainStateSpace = \compressChainStateSpaceStar\cup\bigcup_{\interval\in\intPart}\compressChainStateSpaceJ$.

We now define the transition function $\compressChainTrans$.
For all $\ocConfig\in\compressChainStateSpaceStar$, we let $\compressChainTrans(\ocConfig)(\ocConfig)=1$.
For configurations whose counter value lies in one of the intervals $\interval\in\intPart$, we provide a unified definition based on a notion of \textit{successor counter values}, generalising the ideas of Example~\ref{example:compression:main} and Figure~\ref{figure:mc:counter update scheme}.

Let $\interval\in\intPart$.
If $\interval$ is unbounded, we define the successor of $\intLB_\interval=\min\interval$ to be $\intLB_\interval-1$.
We now assume that $\interval = \integerInterval{\intLB_\interval, \intUB_\interval}$ is bounded and let $\powerMax_\interval=\log_2(|\interval|+1)$ and $\powerIndex\in\integerInterval{\powerMax_\interval-1}$.
The successors of $\intLB_\interval + 2^{\powerIndex}-1$ are $\intLB_\interval-1$ and $\intLB_\interval + 2^{\powerIndex+1}-1$.
Symmetrically, for $\intUB_\interval - (2^{\powerIndex}-1)$, its successors are $\intUB_\interval+1$ and $\intUB_\interval - (2^{\powerIndex+1}-1)$.
Both cases entail a counter change by $2^\powerIndex$.
Assumption~\ref{assumption:interval size} ensures that all successor counter values appear in $\compressChainStateSpace$.

Let $\ocConfig = (\ocState, \ocCount)\in\compressChainStateSpace\setminus\compressChainStateSpaceStar$ and $\ocConfig' = (\ocState', \ocCount')\in\compressChainStateSpace\setminus\{\bot\}$.
If $\ocCount'$ is not a successor of $\ocCount$, we set $\compressChainTrans(\ocConfig)(\ocConfig') = 0$.
Assume now that $\ocCount'$ is a successor of $\ocCount$.
We let $\succHist{\ocConfig}{\ocConfig'}\subseteq\histSet{\ocmdpFin{\ocmdp}{\counterUB}}$ be the set of histories $\hist$ such that $\first{\hist} = \ocConfig$, $\last{\hist}=\ocConfig'$ and for all configurations $\ocConfig''$ along $\hist$ besides $\ocConfig'$, the counter value of $\ocConfig''$ is not a successor of $\ocCount$; outside of $\ocConfig'$ along $\hist$, the counter remains, in the bounded case, strictly between the two successors of $\ocCount$, and, in the unbounded case, strictly above the successor $\ocCount-1$ of $\ocCount$.
We set $\compressChainTrans(\ocConfig)(\ocConfig') = \probaGverb{\ocmdpFin{\ocmdp}{\counterUB}}{\ocConfig}{\strat}(\cyl{\succHist{\ocConfig}{\ocConfig'}})$.
To ensure that $\compressChainTrans(\ocConfig)$ is a distribution we let $\compressChainTrans(\ocConfig)(\bot) = 1 - \sum_{\ocConfig''\neq\bot}\compressChainTrans(\ocConfig)(\ocConfig'')$; this transition captures the probability of the counter never hitting a successor of $\ocCount$.

\begin{remark}\label{remark:compressing ocmcs}
  Although we have formalised compressed Markov chains for OC-MDPs, the construction can be applied to one-counter Markov chains.
  In particular, the properties outlined below transfer to the compression of a one-counter Markov chain.
  \hfill$\lhd$
\end{remark}

In the following, we differentiate histories of $\compressChain$ from histories of $\ocmdpFin{\ocmdp}{\counterUB}$ by denoting them with a bar, e.g., $\mcHist$ indicates a history of $\compressChain$.

\subsection{Efficiently refining interval partitions}\label{section:compression:refinement}
To define a compressed Markov chain with respect to an interval partition $\intPartB$ of $\integerInterval{1, \counterUB-1}$, we require that the size constraints of Assumption~\ref{assumption:interval size} hold, i.e., that for all bounded $\interval\in\intPartB$, $\log_2(|\interval|+1)\in\IN$.
We present a refinement procedure for interval partitions that enforces this property while generating few intervals.
At the end of this section, we provide an additional procedure that can be used to retain specific configurations in the state space of compressed Markov chains.

To refine an interval partition, we subdivide its bounded intervals one by one.
Breaking up these intervals into singleton sets is not a valid approach for complexity reasons; any input interval partition is such that its interval bounds are represented in binary, i.e., the size of intervals is exponential in the size of their representation.
We provide a polynomial-time refinement procedure that divides an interval into sub-intervals of the required size in Algorithm~\ref{algorithm:refine:intervals}.
To refine an interval, we determine a largest sub-interval of a suitable size and then continue by recursively partitioning its complement.
This algorithm enables us, in the context of verification, to modify the interval partition from the representation of an interval strategy into one suitable for compressed Markov chains.

\begin{algorithm}[t]
  \caption{Procedure \textsf{Refine} to split an interval into intervals of size in $\{2^\powerMax-1\mid\powerMax\in\IN\}$.} \label{algorithm:refine:intervals}
  \KwData{A bounded interval $\interval = \integerInterval{\intLB, \intUB}$.}
  $\ell\leftarrow \lfloor\log_2(\intUB-\intLB+2)\rfloor (=\lfloor\log_2(|\interval|+1)\rfloor$)\;
  \If{$|\interval| = 2^\ell-1$}{
    \Return $\{\interval\}$\;
  }
  \Else{
    $\interval' \leftarrow \integerInterval{\intLB, \intLB + 2^\ell -2}$; $\interval'' \leftarrow \integerInterval{\intLB + 2^\ell -1, \intUB}$\;
    \Return $\{\interval'\}\cup\textsf{Refine}(\interval'')$\;
  }
\end{algorithm}

We show that, for all bounded intervals $\interval$ of $\INpos$, the partition $\mathsf{Refine}(\interval)$ (from Algorithm~\ref{algorithm:refine:intervals}) has a polynomial size (with respect to the binary encoding of the bounds of $\interval$) and all of its elements $\intervalB$ satisfy $\log_2(|\intervalB|+1)\in\IN$.
\begin{lemma}\label{lemma:ocmpd:interval size}
  Let $\interval=\integerInterval{\intLB, \intUB}$ be a bounded interval of $\INpos$.
  The interval partition $\mathsf{Refine}(\interval)$ of $\interval$ obtained via Algorithm~\ref{algorithm:refine:intervals} satisfies $|\mathsf{Refine}(\interval)|\leq\log_2(|\interval|+1)+1\leq\log_2(\intUB)+1$ and, for all $\intervalB\in\mathsf{Refine}(\interval)$, we have $\log_2(|\intervalB|+1)\in\IN$.
\end{lemma}
\begin{proof}
  Let $\ell=\lfloor\log_2(|\interval|+1)\rfloor$.
  We show both statements by induction.

  For the size of the elements in $\mathsf{Refine}(\interval)$, we proceed by induction on $|\interval|$.
  If $|\interval| = 1$, then $\mathsf{Refine}(\interval)=\{\interval\}$ (since $1 = 2^1-1$) and the result follows.
  Now, assume that for all intervals smaller than $\interval$, the statement holds.
  If $\interval = 2^\ell-1$, we have $\mathsf{Refine}(\interval)=\{\interval\}$ which satisfies the condition.
  Otherwise, we let $\interval' = \integerInterval{\intLB, \intLB + 2^\ell -2}$ and $\interval'' = \integerInterval{\intLB + 2^\ell -1, \intUB}$.
  In particular, we have $|\interval'| = \intLB + 2^\ell -2 - \intLB + 1 = 2^\ell-1$ and thus $\log_2(|\interval'|+1)\in\IN$.
  We conclude that all elements of $\mathsf{Refine}(\interval) = \{\interval'\}\cup\mathsf{Refine}(\interval'')$ satisfy the required constraint on their size.

  We now show that $|\mathsf{Refine}(\interval)|\leq\ell+1$.
  We proceed by induction on $\ell$.
  If $\ell=1$, then $|\interval|=1$ and we have $|\mathsf{Refine}(\interval)|=1$.
  This closes the base case.
  
  We assume by induction that for all $\interval'$ such that $\ell'=\lfloor\log_2(|\interval'|+1)\rfloor < \ell$, we have $|\mathsf{Refine}(\interval')|\leq\ell'+1$.
  If $|\interval| = 2^\ell-1$, we have $|\mathsf{Refine}(\interval)|=1\leq\ell+1$.
  We thus assume that $2^\ell-1 <|\interval| < 2^{\ell+1}-1$ (the upper bound follows from the definition of $\ell$).
  We let $\interval' = \integerInterval{\intLB, \intLB + 2^\ell -2}$ and $\interval'' = \integerInterval{\intLB + 2^\ell -1, \intUB}$.
  It remains to show that $|\mathsf{Refine}(\interval'')|\leq\ell$ to conclude.
  It holds that $|\interval''| = |\interval| - (2^{\ell}-1) < 2^{\ell+1}-1 - (2^{\ell}-1) = 2^\ell$.
  We distinguish two cases in light of this.
  First, we assume that $|\interval''| = 2^\ell-1$.
  In this case, we have $|\mathsf{Refine}(\interval'')|=1$, which implies that $|\mathsf{Refine}(\interval)|=2\leq \ell+1$.
  Second, we assume that $|\interval''| < 2^\ell-1$.
  By the induction hypothesis, we obtain that $|\mathsf{Refine}(\interval'')|\leq\ell$, ensuring that $|\mathsf{Refine}(\interval)|\leq\ell+1$ and ending the argument.
\end{proof}

For the sake of conciseness, we extend the $\mathsf{Refine}$ operator to infinite intervals and interval partitions.
For any infinite interval $\interval$ of $\INpos$, we let $\mathsf{Refine}(\interval) = \{\interval\}$.
Let $\intervalB$ be an interval of $\INpos$ and let $\intPartB$ be a partition of $\intervalB$.
We let $\mathsf{Refine}(\intPartB) = \bigcup_{\interval\in\intPartB}\mathsf{Refine}(\interval)$.
Lemma~\ref{lemma:ocmpd:interval size} implies that the constraints of Assumption~\ref{assumption:interval size} are satisfied by $\mathsf{Refine}(\intPartB)$. 
This result also yields bounds on the size of $\mathsf{Refine}(\intPartB)$ when $\intPartB$ is finite.

We remark that if an interval partition $\intPartB$ of $\INpos$ has period $\period$ and is generated by an interval partition $\intPartB'$ of $\integerInterval{1, \period}$, then $\mathsf{Refine}(\intPartB)$ is generated by $\mathsf{Refine}(\intPartB')$.

We now introduce an operator ensuring that a specific counter value is retained in a compression by making it an interval bound.
For any interval $\interval=\integerInterval{\intLB, \intUB}$ of $\INpos$ (not necessarily bounded) and $\ocCount\in\IN$, we let $\mathsf{Isolate}(\interval, \ocCount)$ denote $\{\interval\}$ if $\ocCount\notin\interval$ and $\{\integerInterval{\intLB, \ocCount}, \integerInterval{\ocCount+1, \intUB}\}\setminus\{\emptyset\}$ if $\ocCount\in\interval$.
We extend the $\mathsf{Isolate}$ operator to interval partitions as follows.
For all intervals $\intervalB$ of $\INpos$, interval partitions $\intPartB$ of $\intervalB$ and $\ocCount\in\IN$, we let $\mathsf{Isolate}(\intPartB, \ocCount) = \bigcup_{\interval\in\intPartB}\mathsf{Isolate}(\interval, \ocCount)$.

\subsection{Validity of the compression approach}\label{section:abstraction:validation}

We establish that for all configurations $\ocConfig\in\compressChainStateSpace$ and states $\ocState\in\ocStateSpace$, the probability of terminating or reaching the counter upper bound $\counterUB$ in $\ocState$ coincides in $\compressChain$ and in the Markov chain induced on $\ocmdpFin{\ocmdp}{\counterUB}$ by $\strat$.
There is not such a direct correspondence for state-reachability probabilities.
We prove that for all targets $\target\subseteq\ocStateSpace$, there exists an OC-MDP $\ocmdp'$ with state space $\ocStateSpace$ derived by changing transitions of $\ocmdp$ such that, for all $\ocState\in\target$, the probability of visiting $\target$ for the first time via a configuration with state $\ocState$ in $\ocmdpFin{\ocmdp}{\counterUB}$ under $\strat$ coincides with the probability of terminating or hitting the counter upper bound in $\ocState$ in $\ocmdpFin{\ocmdp'}{\counterUB}$ under $\strat$.

For the first property, we rely on a relation between histories of $\compressChain$ and of $\ocmdpFin{\ocmdp}{\counterUB}$.
Let $\hist = \ocConfig_0\ocAction_0\ldots\ocAction_{\indexLast-1}\ocConfig_\indexLast\in\histSet{\ocmdpFin{\ocmdp}{\counterUB}}$ such that $\last{\hist}\in\ocStateSpace\times\{0, \counterUB\}$ and $\last{\hist}$ occurs only once in $\hist$.
By induction, we identify a sequence of configurations in $\compressChainStateSpace$ along $\hist$ that is a well-formed history of $\compressChain$.
Let $\indexPosition_0 = 0$.
Assume that we have constructed an increasing sequence $\indexPosition_0 <\ldots < \indexPosition_\subindexPosition$ such that $\ocConfig_{\indexPosition_0}\ldots\ocConfig_{\indexPosition_\subindexPosition}\in\histSet{\compressChain}$.
If $\indexPosition_\subindexPosition\neq\indexLast$, we let $\indexPosition_{\subindexPosition+1}$ be the least index $\indexPosition > \indexPosition_{\subindexPosition}$ such that $\ocConfig_\indexPosition\in\compressChainStateSpace$ and  $\compressChainTrans(\ocConfig_{\indexPosition_{\subindexPosition}})(\ocConfig_{\indexPosition}) > 0$ and continue the induction.
Such an index is guaranteed to exist.
Since weights are in $\{-1, 0, 1\}$, we witness all counter values between that of $\ocConfig_{\indexPosition_\subindexPosition}$ and $\ocConfig_\indexLast$ in the suffix $\ocConfig_{\indexPosition_\subindexPosition}\ldots\ocConfig_\indexLast$.
This index necessarily exists: all counter values have a smaller successor, and those from a bounded interval have a greater successor.
If $\indexPosition_{\subindexPosition}=\indexLast$, the induction ends and we let $\mcHist = \ocConfig_0\ocConfig_{\indexPosition_1}\ldots\ocConfig_{\indexPosition_{\indexLast'-1}}\ocConfig_{\indexPosition_{\indexLast'}}$ be the resulting history.
We say that $\mcHist$ \textit{abstracts} $\hist$, and it is the unique history of $\compressChain$ that does so.

We now state the first theorem of this section.
The crux of its proof is to establish that, for all histories $\mcHist$ of $\compressChain$ ending in $\ocStateSpace\times\{0, \counterUB\}$, the probability of its cylinder in $\compressChain$ matches the probability that a history abstracted by $\mcHist$ occurs in the Markov chain induced by $\strat$ on $\ocmdpFin{\ocmdp}{\counterUB}$.
We conclude by writing reachability objectives as countable unions of cylinders.
For the sake of clarity, in the following statement, we indicate the relevant MDP or Markov chain for each objective.
\begin{theorem}\label{theorem:ocmdp:probability matching}
  Let $\ocConfig \in\compressChainStateSpace\setminus\{\bot\}$.
  For all $\ocState\in\ocStateSpace$, $\probaGverb{\ocmdpFin{\ocmdp}{\counterUB}}{\ocConfig}{\strat}(\selectiveTermination{\ocState}) = \probaMCverb{\compressChain}{\ocConfig}(\reachVerb{\compressChain}{\ocState, 0})$ and,
  if $\counterUB\in\IN$, $\probaGverb{\ocmdpFin{\ocmdp}{\counterUB}}{\ocConfig}{\strat}(\reachVerb{\ocmdpFin{\ocmdp}{\counterUB}}{\ocState, \counterUB}) = \probaMCverb{\compressChain}{\ocConfig}(\reachVerb{\compressChain}{\ocState, \counterUB})$.
\end{theorem}
\begin{proof}
    Let $\ocState\in\ocStateSpace$.
  We only prove the result for the target configuration $(\ocState, 0)$.
  The argument is the same when the target is $(\ocState, \counterUB)$.
  To lighten notation, we let $\mdp = \ocmdpFin{\ocmdp}{\counterUB}$ for the remainder of the proof.
  For the sake of clarity, we indicate whether cylinder sets are with respect to $\mdp$ or $\compressChain$ by indexing the notation with whichever is applicable. 

  Let $\mcHist=\ocConfig_0\ocConfig_1\ldots\ocConfig_\indexLast\in\histSet{\compressChain}$ be such that $\first{\mcHist} =\ocConfig$ and $\last{\mcHist}=(\ocState, 0)$.
  We let $\absHist{\mcHist}\subseteq\histSet{\mdp}$ be the set of histories abstracted by $\mcHist$.
  We show that $\probaMCverb{\compressChain}{\ocConfig}(\cylVerb{\compressChain}{\mcHist}) = \probaGverb{\mdp}{\ocConfig}{\strat}(\cylVerb{\mdp}{\absHist{\mcHist}})$.
  By construction of the abstraction relation, all elements of $\absHist{\mcHist}$ are a uniquely defined concatenation of an element of $\succHist{\ocConfig_0}{\ocConfig_1}$ with an element of $\succHist{\ocConfig_1}{\ocConfig_2}$, \ldots, with an element of $\succHist{\ocConfig_{\indexLast-1}}{\ocConfig_\indexLast}$.
  Conversely, any such concatenation is an element of $\absHist{\mcHist}$.
  We obtain: \begingroup
  \allowdisplaybreaks
  \begin{align*}
    \probaMCverb{\compressChain}{\ocConfig}(\cylVerb{\compressChain}{\mcHist})
    &= \prod_{\indexPosition=0}^{\indexLast-1}\probaGverb{\mdp}{\ocConfig}{\strat}(\cylVerb{\mdp}{\succHist{\ocConfig_\indexPosition}{\ocConfig_{\indexPosition+1}}}) \\
    &= \prod_{\indexPosition=0}^{\indexLast-1}
      \left(\sum_{\hist_\indexPosition\in\succHist{\ocConfig_\indexPosition}{\ocConfig_{\indexPosition+1}}}\probaGverb{\mdp}{\ocConfig}{\strat}(\cylVerb{\mdp}{\hist_\indexPosition})\right) \\
    & = \sum_{\hist_0\in\succHist{\ocConfig_0}{\ocConfig_1}}\ldots\sum_{\hist_{\indexLast-1}\in\succHist{\ocConfig_{\indexLast-1}}{\ocConfig_\indexLast}}\left(
      \prod_{\indexPosition=0}^{\indexLast-1}
      \probaGverb{\mdp}{\ocConfig}{\strat}(
      \cylVerb{\mdp}{\hist_\indexPosition})
      \right) \\
    & = \sum_{\hist_0\in\succHist{\ocConfig_0}{\ocConfig_1}}\ldots\sum_{\hist_{\indexLast-1}\in\succHist{\ocConfig_{\indexLast-1}}{\ocConfig_\indexLast}}\left(
      \probaGverb{\mdp}{\ocConfig}{\strat}(
      \cylVerb{\mdp}{\histConcat{\hist_0}{\histConcat{\ldots}{\hist_{\indexLast-1}}}})
      \right) \\
    & = \probaGverb{\mdp}{\ocConfig}{\strat}(\cylVerb{\mdp}{\absHist{\mcHist}}).
  \end{align*}\endgroup
  The first line follows by definition of $\compressChainTrans$ and the definition of the probability distribution over plays of Markov chains.
  For the second line, we first observe that for all $\indexPosition\in\integerInterval{\indexLast-1}$, the set $\succHist{\ocConfig_\indexPosition}{\ocConfig_{\indexPosition+1}}$ is prefix-free and thus the cylinders of elements of $\succHist{\ocConfig_\indexPosition}{\ocConfig_{\indexPosition+1}}$ are pairwise disjoint.
  The third line is a rewriting of the second.
  The fourth line is obtained by definition of the probability distribution induced by a strategy in an MDP, using the fact that $\strat$ is a memoryless strategy.
  The last line is obtained because $\absHist{\mcHist}$ is the set of all concatenations occurring in the previous line.

  We can now end the argument.
  Let $\histPart$ and $\bar\histPart$ denote the set of histories of $\mdp$ and $\compressChain$ respectively that start in $\ocConfig$ and end in $(\ocState, 0)$ with only one occurrence of $(\ocState, 0)$.
  These two sets are prefix-free and we have $\histPart = \bigcup_{\mcHist\in\bar{\histPart}}\absHist{\mcHist}$.
  Using the above, we conclude that:
  \[
    \probaGverb{\mdp}{\ocConfig}{\strat}(\selectiveTermination{\ocState}) =
    \sum_{\hist\in\histPart}\probaGverb{\mdp}{\ocConfig}{\strat}(\cylVerb{\mdp}{\hist}) =
    \sum_{\mcHist\in\bar\histPart}\probaMCverb{\compressChain}{\ocConfig}(\cylVerb{\compressChain}{\mcHist}) =
    \probaMCverb{\compressChain}{\ocConfig}(\reach{(\ocState, 0)}).
  \]
  This is the claim of the theorem.
\end{proof}

We now discuss state-reachability probabilities.
Let $\target\subseteq\ocStateSpace$ be a target.
Transitions of $\compressChain$ group together (possibly infinitely many) transitions of $\ocmdpFin{\ocmdp}{\counterUB}$.
In particular, this compression may result in some visits to $\target$ not being observed in $\compressChain$ despite occurring in $\ocmdpFin{\ocmdp}{\counterUB}$.
By slightly modifying $\ocmdp$, we can guarantee that all of these visits are witnessed in the new compressed Markov chain.

The idea is to make all target states absorbing with self-loops of weight $-1$.
Formally, we let $\ocmdp' = (\ocStateSpace, \ocActionSpace, \ocTrans', \weight')$ be the OC-MDP defined by letting, for all $\ocState\in\ocStateSpace$ and all $\ocAction\in\ocActionSpace(\ocState)$, $\ocTrans'(\ocState, \ocAction) = \ocTrans(\ocState, \ocAction)$ and $\weight'(\ocState, \ocAction) = \weight(\ocState, \ocAction)$ if $\ocState\notin\target$ and, otherwise, $\ocTrans'(\ocState, \ocAction)(\ocState)=1$ and $\weight'(\ocState, \ocAction)=-1$.
We remark that $\strat$ is a well-defined memoryless strategy of $\ocmdpFin{\ocmdp'}{\counterUB}$.

By design, any history of $\ocmdpFin{\ocmdp}{\counterUB}$ that ends in a configuration in $\target\times\integerInterval{\counterUB}$ and that does not visit this set before is also a history of $\ocmdpFin{\ocmdp'}{\counterUB}$.
The cylinders of these histories in both MDPs have the same probability under $\strat$, as transitions are the same in states outside of $\target$.
This implies that, under $\strat$, the probability of terminating or hitting the counter upper bound in $\target$ in $\ocmdpFin{\ocmdp'}{\counterUB}$ is equal to the probability of reaching $\target$ in $\ocmdpFin{\ocmdp}{\counterUB}$.
We conclude by Theorem~\ref{theorem:ocmdp:probability matching} that the compressed Markov chain $\compressChainB$ captures the state-reachability probabilities for the target $\target$ in $\ocmdpFin{\ocmdp}{\counterUB}$ under $\strat$.
We formalise this by the following theorem, in which, for the sake of clarity, we indicate the relevant MDP or Markov chain for each objective.

\begin{theorem}\label{theorem:ocmdp:probability:reach}
    Let $\target\subseteq\ocStateSpace$.
  Let $\ocmdp' = (\ocStateSpace, \ocActionSpace, \ocTrans', \weight')$ be defined as above. For all $\ocConfig\in\compressChainStateSpace$, we have
  $\probaGverb{\ocmdpFin{\ocmdp}{\counterUB}}{\ocConfig}{\strat}(\reachVerb{\ocmdpFin{\ocmdp}{\counterUB}}{\target}) =
  \probaMCverb{\compressChainB}{\ocConfig}(\reachVerb{\compressChainB}{\target\times\{0, \counterUB\}})$.
\end{theorem}
\begin{proof}
    To lighten notation, we let $\mdp = \ocmdpFin{\ocmdp}{\counterUB}$ and $\mdp' = \ocmdpFin{\ocmdp'}{\counterUB}$ for the remainder of the proof.
  We also indicate whether cylinder sets are with respect to $\mdp$ or $\mdp'$ by indexing the notation by the applicable MDP.
  By Theorem~\ref{theorem:ocmdp:probability matching}, it suffices to show that $\probaGverb{\mdp}{\ocConfig}{\strat}(\reachVerb{\mdp}{\target}) =\probaGverb{\mdp'}{\ocConfig}{\strat}(\reachVerb{\mdp'}{\target\times\{0, \counterUB\}})$.
  
  Let $\histPart\subseteq\histSet{\mdp}$ be the set of histories $\hist\in\histSet{\mdp}$ such that $\last{\hist}\in\target\times\integerInterval{\counterUB}$ and no prior configuration is in $\target\times\integerInterval{\counterUB}$.
  The state-reachability objective $\reachVerb{\mdp}{\target}$ can be written as $\cylVerb{\mdp}{\histPart}$.
  Since $\histPart$ is prefix-free, we have $\probaGverb{\mdp}{\ocConfig}{\strat}(\reachVerb{\mdp}{\target}) = \sum_{\hist\in\histPart}\probaGverb{\mdp}{\ocConfig}{\strat}(\cylVerb{\mdp}{\hist})$.
  Furthermore, for all $\hist\in\histPart$, by definition of $\ocTrans'$, since no configuration with a state in $\target$ occurs along $\hist$ besides the last one, we have $\hist\in\histSet{\mdp'}$ and $\probaGverb{\mdp}{\ocConfig}{\strat}(\cylVerb{\mdp}{\hist}) =\probaGverb{\mdp'}{\ocConfig}{\strat}(\cylVerb{\mdp'}{\hist})$.

  To end the proof, it suffices to show that $\cylVerb{\mdp'}{\histPart} = \reachVerb{\mdp'}{\target\times\{0, \counterUB\}}$.
  We show both inclusions.
  Let $\play\in\cylVerb{\mdp'}{\histPart}$.
  By definition of $\histPart$, there exists a configuration of $\play$ with a state in $\target$.
  If the counter value of this configuration is $\counterUB$, then we have $\play\in\reachVerb{\mdp'}{\target\times\{0, \counterUB\}}$.
  If not, we are guaranteed to have a configuration in $\target\times\{0\}$ along $\play$ because states of $\target$ are absorbing in $\ocmdp'$ and their self-loops have weight $-1$.
  Conversely, let $\play\in\reachVerb{\mdp'}{\target\times\{0, \counterUB\}}$.
  By definition of the reachability objective, there must be a configuration with a state in $\target$ along $\play$.
  The earliest occurrence of a state of $\target$ (regardless of the counter value) witnesses that $\play\in\cylVerb{\mdp'}{\histPart}$.
  This ends the proof.
\end{proof}
\subsection{Characterising transition probabilities}\label{section:abstraction:transitions}

Example~\ref{example:compression:main} illustrates that the transition probabilities of a compressed Markov chain may require large representations or be irrational.
This section presents characterisations of these transition probabilities via equation systems.

For the remainder of this section, we fix an interval $\interval\in\intPart$.
We present a system characterising the outgoing transition probabilities from configurations of $\compressChain$ with counter value in $\interval$.
In Section~\ref{section:abstraction:transitions:unbounded}, we assume that $\interval$ is unbounded, and we handle the bounded case in Section~\ref{section:abstraction:transitions:bounded}.
We also provide bounds on the size of the systems.

\subsubsection{Unbounded intervals}\label{section:abstraction:transitions:unbounded}
We assume that $\interval$ is an infinite interval and let $\intLB=\min\interval$.
This implies that $\counterUB=\infty$.
We characterise the transition probabilities from the configurations of $\compressChainStateSpace$ with counter value $\intLB$ via existing results on termination probabilities in one-counter Markov chains.

For any $\ocState$, $\ocStateB\in\ocStateSpace$, the transition probability $\compressChainTrans((\ocState, \intLB))((\ocStateB, \intLB-1))$ can be seen as a termination probability in a one-counter Markov chain.
Let $\stratB$ denote the counter-oblivious strategy $\strat(\cdot, \intLB)$.
We consider the one-counter Markov chain $\ocChain=(\ocStateSpace, \ocTrans^\stratB)$, where, for all $\ocState$, $\ocStateB\in\ocStateSpace$ and all $\weightVal\in\{-1, 0, 1\}$, we let $\ocTrans^\stratB(\ocState)(\ocStateB, \weightVal) = \sum_{\ocAction\in\ocActionSpace(\ocState), \weight(\ocState, \ocAction)=\weightVal}\stratB(\ocState)(\ocAction)\cdot\ocTrans(\ocState, \ocAction)(\ocStateB)$.

Let $\ocState$, $\ocStateB\in\ocStateSpace$ and let $\ocConfig = (\ocState, \intLB)$.
There is a bijection between $\hist\in\succHist{\ocConfig}{(\ocStateB, \intLB-1)}$ and the set of histories of $\ocChainFin{\ocChain}{\infty}$ that start in $(\ocState, 1)$ and end in $(\ocStateB, 0)$: one omits all actions and subtracts $\intLB-1$ to all counter values in the history.
By definition of $\strat$ and $\ocTrans^\stratB$, this bijection preserves the probability of cylinders.
This implies that $\compressChainTrans((\ocState, \intLB))((\ocStateB, \intLB-1))$ is exactly the probability, in $\ocChainFin{\ocChain}{\infty}$, of terminating in $\ocStateB$ from $(\ocState, 1)$.

It follows that the characterisation of termination probabilities in one-counter Markov chains of~\cite{DBLP:journals/lmcs/KuceraEM06} applies to our transition probabilities in this case.
We obtain the following.

\begin{theorem}[\cite{DBLP:journals/lmcs/KuceraEM06}]\label{theorem:equations:termination}
  For each $\ocState, \ocStateB\in\ocStateSpace$, we consider a variable $\termProbaVar{\ocState}{\ocStateB}$ and the system of equations formed by the equations, for all $\ocState, \ocStateB\in\ocStateSpace$,
  \[\termProbaVar{\ocState}{\ocStateB} =
    \ocTrans^\interval(\ocState)(\ocStateB, -1) +
    \sum_{\ocStateC\in\ocStateSpace} \ocTrans^\interval(\ocState)(\ocStateC, 0)\cdot\termProbaVar{\ocStateC}{\ocStateB} +
    \sum_{\ocStateC\in\ocStateSpace} \ocTrans^\interval(\ocState)(\ocStateC, 1)\cdot
    \left(\sum_{\ocStateC'\in\ocStateSpace}\termProbaVar{\ocStateC}{\ocStateC'}\cdot\termProbaVar{\ocStateC'}{\ocStateB}\right),
  \]
  where $\ocTrans^\interval(\ocStateC)(\ocStateC', \weightVal) = \sum_{\ocAction\in\ocActionSpace(\ocStateC), \weight(\ocStateC, \ocAction)=\weightVal}\strat(\ocStateC, \intLB)(\ocAction)\cdot\ocTrans(\ocStateC, \ocAction)(\ocStateC')$ for all $\ocStateC$, $\ocStateC'\in\ocStateSpace$ and all $\weightVal\in\{-1, 0, 1\}$.
  The least non-negative solution of this system is obtained by substituting each variable $\termProbaVar{\ocState}{\ocStateB}$ by  $\compressChainTrans((\ocState, \intLB))((\ocStateB, \intLB-1))$.
\end{theorem}

The equation system of Theorem~\ref{theorem:equations:termination} has one variable of the form $\termProbaVar{\ocState}{\ocStateB}$ for every two states $\ocState$, $\ocStateB\in\ocStateSpace$ and there is one equation per variable.
Furthermore, the equations have length polynomial in the sizes of $|\ocActionSpace|$ and $|\ocStateSpace|$.
Indeed, if we distribute all products in the right-hand sides of the equations to rewrite them as a sum of products, there are at most $|\ocActionSpace|\cdot|\ocStateSpace|^2$ products of at most four variables or constants.
We obtain the following result.

\begin{lemma}\label{lemma:equations:termination:size}
  The equation system of Theorem~\ref{theorem:equations:termination} has $|\ocStateSpace|^2$ variables and equations.
  Its equations have length polynomial in $|\ocStateSpace|$ and $|\ocActionSpace|$.
\end{lemma}

\subsubsection{Bounded intervals}\label{section:abstraction:transitions:bounded}

We now assume that $\interval$ is bounded.
We write $\interval=\integerInterval{\intLB, \intUB}$ and let $\powerMax=\log_2(|\interval|+1)\in\INpos$.
To improve readability, we assume that $\intLB = 1$ and $\intUB = 2^\powerMax-1$.
All results below can be recovered for the general case by adding $\intLB-1$ to the counter values in configurations.

Counter values of $\interval$ that are kept in $\compressChainStateSpace$ can be partitioned in two sets: the set $\{2^{\powerIndex} \mid \powerIndex\in\integerInterval{\powerMax-1}\}$ of values reachable from $\intLB$ and the set $\ocStateSpace\times\{2^{\powerMax} - 2^{\powerIndex} \mid \powerIndex\in\integerInterval{\powerMax-1}\}$ of values reachable from $\intUB$ (in the sense of Figure~\ref{figure:mc:counter update scheme}).
By symmetry of the transition structure of the compressed Markov chain, the outgoing transitions from a configuration $(\ocState, 2^\powerIndex)$ correspond to outgoing transitions from the configuration $(\ocState, 2^\powerMax-2^\powerIndex)$.

\begin{lemma}\label{lemma:abstraction:transitions:symmetry}
  Let $\ocState, \ocStateB\in\ocStateSpace$ and $\powerIndex\in\integerInterval{\powerMax-1}$.
  It holds that $\compressChainTrans(\ocState, 2^\powerIndex)(\ocStateB, 2^{\powerIndex+1}) = \compressChainTrans(\ocState, 2^\powerMax-2^\powerIndex)(\ocStateB, 2^\powerMax)$ and $\compressChainTrans(\ocState, 2^\powerIndex)(\ocStateB, 0) = \compressChainTrans(\ocState, 2^\powerMax-2^\powerIndex)(\ocStateB, 2^ \powerMax - 2^{\powerIndex+1})$
\end{lemma}
\begin{proof}
  We only prove that $\compressChainTrans(\ocState, 2^\powerIndex)(\ocStateB, 2^{\powerIndex+1}) = \compressChainTrans(\ocState, 2^\powerMax-2^\powerIndex)(\ocStateB, 2^\powerMax)$ as the other case is similar.
  We define a bijection $f\colon\succHist{(\ocState, 2^\powerIndex)}{(\ocStateB, 2^{\powerIndex+1})}\to\succHist{(\ocState, 2^\powerMax-2^\powerIndex)}{(\ocStateB, 2^\powerMax)}$ and prove that we have $\probaGverb{\ocmdpFin{\ocmdp}{\counterUB}}{(\ocState, 2^\powerIndex)}{\strat}(\cyl{\hist})=\probaGverb{\ocmdpFin{\ocmdp}{\counterUB}}{(\ocState, 2^\powerMax-2^\powerIndex)}{\strat}(\cyl{f(\hist)})$ for all $\hist\in\succHist{(\ocState, 2^\powerIndex)}{(\ocStateB, 2^{\powerIndex+1})}$.
  This is sufficient to obtain our result.

  Let $\hist\in\succHist{(\ocState, 2^\powerIndex)}{(\ocStateB, 2^{\powerIndex+1})}$.
  We let $f(\hist)$ be the history obtained by adding $2^\powerMax-2^{\powerIndex+1}$ to all counter values along $\hist$.
  We must show that $f(\hist)\in\succHist{(\ocState, 2^\powerMax-2^\powerIndex)}{(\ocStateB, 2^\powerMax)}$.
  For the first and last configurations, we observe that $2^\powerIndex + 2^\powerMax - 2^{\powerIndex+1} = 2^\powerMax - 2^\powerIndex$ and $2^{\powerIndex+1} + 2^\powerMax - 2^{\powerIndex+1} = 2^\powerMax$.
  For the other configurations, their counter values are in the interval $\integerInterval{1, 2^{\powerIndex+1}-1}$, thus their counterparts in $f(\hist)$ have a counter value in $\integerInterval{2^\powerMax-2^{\powerIndex+1} + 1, 2^{\powerMax}-1}$.

  We now establish that $\probaGverb{\ocmdpFin{\ocmdp}{\counterUB}}{(\ocState, 2^\powerIndex)}{\strat}(\cyl{\hist})=\probaGverb{\ocmdpFin{\ocmdp}{\counterUB}}{(\ocState, 2^\powerMax-2^\powerIndex)}{\strat}(\cyl{f(\hist)})$.
  Let $\hist = (\ocState_0, \ocCount_0)\ocAction_0(\ocState_1, \ocCount_1)\ldots\ocAction_\indexLast(\ocState_\indexLast, \ocCount_\indexLast)$.
  Because $\strat$ is memoryless, based on $\intPart$ and $\interval\in\intPart$, we obtain
  \begin{align*}
    \probaGverb{\ocmdpFin{\ocmdp}{\counterUB}}{(\ocState, 2^\powerIndex)}{\strat}(\cyl{\hist})
    & = \prod_{\indexPosition=0}^{\indexLast-1}
      \ocTrans(\ocState_\indexPosition, \ocAction_\indexPosition)(\ocState_{\indexPosition+1})\cdot
      \strat(\ocConfig_\indexPosition, \ocCount_\indexPosition)(\ocAction_\indexPosition) \\
    & = \prod_{\indexPosition=0}^{\indexLast-1}
      \ocTrans(\ocState_\indexPosition, \ocAction_\indexPosition)(\ocState_{\indexPosition+1})\cdot
    \strat(\ocConfig_\indexPosition, \ocCount_\indexPosition+2^\powerMax-2^{\powerIndex+1})(\ocAction_\indexPosition) \\
    & =\probaGverb{\ocmdpFin{\ocmdp}{\counterUB}}{(\ocState, 2^\powerMax-2^\powerIndex)}{\strat}(\cyl{f(\hist)}).
  \end{align*}

  To prove that $f$ is bijective, we define its inverse.
  We let $f^{-1}$ be the function over $\succHist{(\ocState, 2^\powerMax-2^\powerIndex)}{(\ocStateB, 2^\powerMax)}$ that subtracts $2^\powerMax-2^{\powerIndex+1}$ to all counter values along histories.
  It is easy to verify that $f^{-1}$ is well-defined and that it is the inverse of $f$.
\end{proof}

Due to Lemma~\ref{lemma:abstraction:transitions:symmetry}, it is sufficient for us to characterise the outgoing transition probabilities for the configurations in $\ocStateSpace\times\{2^{\powerIndex} \mid \powerIndex\in\integerInterval{\powerMax-1}\}$.
We do so via a quadratic system of equations.
We provide intuition on how to derive this system for our interval $\interval=\integerInterval{1, 2^\powerMax-1}$ by using Markov chains. 

Let us first consider transitions from $\ocStateSpace\times\{1\}$. We illustrate the situation in Figure~\ref{figure:equations:bounded:one}: we consider a Markov chain over $\ocStateSpace\times\{0, 1, 2\}$ where states in $\ocStateSpace\times\{0, 2\}$ are absorbing and transitions from other states are inherited from the Markov chain induced by $\strat$ on $\ocmdpFin{\ocmdp}{\counterUB}$.
We represent transitions in this Markov chain from a configuration $(\ocState, 1)\in\ocStateSpace\times\{1\}$ to configurations with a state $\ocStateB\in\ocStateSpace$.
For any $\ocState\in\ocStateSpace$, the probability of reaching a configuration $\ocConfig'\in\ocStateSpace\times\{0, 2\}$ from $(\ocState, 1)$ in this Markov chain is $\compressChainTrans((\ocState, 1))(\ocConfig')$ by definition.

\begin{figure}
    \centering
    \begin{subfigure}[t]{0.46\textwidth}
      \centering
      \scalebox{0.9}{
        \begin{tikzpicture}[node distance=18mm]
          \node[state, align=center] (q1) {$\ocState, 1$};
          \node[state, align=center, left = of q1] (p0) {$\ocStateB, 0$};
          \node[state, align=center, node distance=15mm, below = of q1] (p1) {$\ocStateB, 1$};
          \node[state, align=center, right = of q1] (p2) {$\ocStateB, 2$};

          \path[->] (q1) edge node[above] {$\ocTrans^\interval(\ocState)(\ocStateB, -1)$} node[below] {$(\weight=-1)$} (p0);
          \path[->] (q1) edge node[align=center,right] {$\ocTrans^\interval(\ocState)(\ocStateB, 0)$\\$(\weight=0)$} (p1);
          \path[->] (q1) edge node[above] {$\ocTrans^{\interval}(\ocState)(\ocStateB, 1)$} node[below] {($\weight=1$)} (p2);
          \path[->] (p0) edge[loop below] (p0);
          \path[->] (p2) edge[loop below] (p2);
        \end{tikzpicture}
      }
      \caption{Markov chain transition scheme for the characterisation of transitions from $\ocStateSpace\times\{1\}$, where $\ocTrans^\interval(\ocState)(\ocStateB, \weightVal) = \sum_{\substack{\ocAction\in\ocActionSpace(\ocState) \\ \weight(\ocState,\ocAction)=\weightVal}}\strat((\ocState, 1))(\ocAction)\cdot\ocTrans(\ocState,\ocAction)(\ocStateB)$.}\label{figure:equations:bounded:one}
    \end{subfigure}\hfill
    \begin{subfigure}[t]{0.5\textwidth}
      \centering
      \scalebox{0.75}{
        \begin{tikzpicture}[node distance=15mm, every state/.style = {minimum size=18mm}]
          \node[state, align=center] (qy) {$\ocState, 2^\powerIndex$};
          \node[state, align=center, left = of qy] (pym) {$\ocStateB, 2^{\powerIndex-1}$};
          \node[state, align=center, right = of qy] (pyp) {$\ocStateB, 3\cdot 2^{\powerIndex-1}$};
          \node[state, align=center, below = of pym] (q0) {$\ocState, 0$};
          \node[state, align=center, below = of pyp] (qyp) {$\ocState, 2^{\powerIndex+1}$};

          \path[->] (qy) edge[bend right] node[above] {$\compressChainTrans(\ocState, 2^{\powerIndex-1})(\ocStateB, 0)$} (pym);
          \path[->] (qy) edge[bend left] node[above] {$\compressChainTrans(\ocState, 2^{\powerIndex-1})(\ocStateB, 2^\powerIndex)$} (pyp);
          \path[->] (pym) edge[bend right] node[below] {$\compressChainTrans(\ocStateB, 2^{\powerIndex-1})(\ocState, 2^\powerIndex)$} (qy);
          \path[->] (pyp) edge[bend left] node[below] {$\compressChainTrans(\ocStateB, 2^{\powerIndex-1})(\ocState, 0)$} (qy);
          \path[->] (pym) edge node[below right] {$\compressChainTrans(\ocStateB, 2^{\powerIndex-1})(\ocState, 0)$} (q0);
          \path[->] (pyp) edge node[below left] {$\compressChainTrans(\ocStateB, 2^{\powerIndex-1})(\ocState, 2^\powerIndex)$} (qyp);

          \path[->] (qyp) edge[loop left] (qyp);
          \path[->] (q0) edge[loop right] (q0);
        \end{tikzpicture}
      }
      \caption{Markov chain transition scheme for the characterisation of transitions from $\ocStateSpace\times\{2^\powerIndex\}$ for $0<\powerIndex<\powerMax$.}
      \label{figure:equations:bounded:induction}
    \end{subfigure}
    \caption{Fragments of Markov chain used to derive a characterisation for transition probabilities of $\compressChain$ for a bounded interval of the form $\integerInterval{1, 2^\powerMax-1}$.
    }\label{figure:equations:bounded}
  \end{figure}
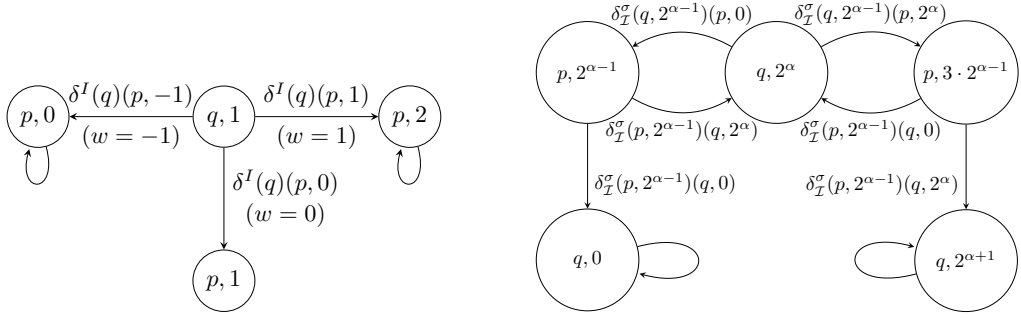

Next, we let $\powerIndex\in\integerInterval{1, \powerMax-1}$ and consider configurations in $\ocStateSpace\times\{2^{\powerIndex}\}$. The situation is depicted in  Fig.~\ref{figure:equations:bounded:induction}.
We divide a counter change by $2^\powerIndex$ into counter changes by $2^{\powerIndex-1}$ and, thus, rely on the transition probabilities  from $\ocStateSpace\times\{2^{\powerIndex-1}\}$ in $\compressChain$.
In this case, we can see transition probabilities from $\ocStateSpace\times\{2^\powerIndex\}$ in $\compressChain$ as reachability probabilities in a Markov chain over $\ocStateSpace\times\{0, 2^{\powerIndex-1}, 2^{\powerIndex},3\cdot2^{\powerIndex-1},2^{\powerIndex+1}\}$.

By putting together the reachability systems for $\ocStateSpace\times\{2^\powerIndex\}$ for all $\powerIndex\in\integerInterval{\powerMax-1}$, we obtain a quadratic system of equations.
To formalise our system and prove its validity, we introduce some notation.
Let $\powerIndex\in\integerInterval{\powerMax-1}$, $\ocState, \ocStateB\in\ocStateSpace$ and $\ocCount\in\integerInterval{1, 2^{\powerIndex+1}-1}$.
We let $\upHistSet{\ocState}{\ocCount}{\ocStateB}{2^{\powerIndex+1}}{\powerIndex}$ (resp.~$\downHistSet{\ocState}{\ocCount}{\ocStateB}{0}{\powerIndex}$) denote the set of histories $\hist$ of $\ocmdpFin{\ocmdp}{\counterUB}$  such that $\first{\hist}=(\ocState, \ocCount)$, $\last{\hist}=(\ocStateB, 2^{\powerIndex+1})$ (resp.~$(\ocStateB, 0)$) and no configuration along $\hist$ besides its last one has a counter value in $\{0, 2^{\powerIndex+1}\}$.
These sets are prefix-free.
We let $\upProba{\ocState}{\ocCount}{\ocStateB}{2^{\powerIndex+1}}{\powerIndex}=
\probaGverb{\ocmdpFin{\ocmdp}{\counterUB}}{(\ocState, \ocCount)}{\strat}(\cyl{\upHistSet{\ocState}{\ocCount}{\ocStateB}{2^{\powerIndex+1}}{\powerIndex}})$ and
$\downProba{\ocState}{\ocCount}{\ocStateB}{2^{\powerIndex+1}}{\powerIndex}=
\probaGverb{\ocmdpFin{\ocmdp}{\counterUB}}{(\ocState, \ocCount)}{\strat}(\cyl{\downHistSet{\ocState}{\ocCount}{\ocStateB}{2^{\powerIndex+1}}{\powerIndex}})$.
The transition probabilities of $\compressChain$ can be written with this notation.
For all $\ocState, \ocStateB\in\ocStateSpace$ and $\powerIndex\in\integerInterval{\powerMax-1}$, we have $\compressChainTrans((\ocState, 2^\powerIndex))((\ocStateB, 0)) = \downProba{\ocState}{2^\powerIndex}{\ocStateB}{0}{\powerIndex}$ and $\compressChainTrans((\ocState,2^\powerIndex))((\ocStateB, 2^{\powerIndex+1})) = \upProba{\ocState}{2^\powerIndex}{\ocStateB}{2^{\powerIndex+1}}{\powerIndex}$.

The following theorem  formalises our characterisation of the transition probabilities of $\compressChain$ for configurations in $\compressChainStateSpace\cap(\ocStateSpace\times\interval)$.
The size of this system is polynomial in $|\ocStateSpace|$ and $\powerMax$.
We provide a self-contained proof that does not refer to the Markov chains described in Figure~\ref{figure:equations:bounded}.
A corollary of this proof is that the Markov chains above yield an accurate characterisation of the transition probabilities.
\begin{theorem}\label{theorem:equations:transitions}
    For each $\ocState, \ocStateB\in\ocStateSpace$, we consider variables $\upProbaVar{\ocState}{1}{\ocStateB}{2}{0}$ and $\downProbaVar{\ocState}{1}{\ocStateB}{0}{0}$, and for all $\powerIndex\in\integerInterval{1, \powerMax-1}$ and $\ocCount\in\{2^{\powerIndex-1}, 2^\powerIndex, 3\cdot 2^{\powerIndex-1}\}$, we consider variables $\upProbaVar{\ocState}{\ocCount}{\ocStateB}{2^{\powerIndex+1}}{\powerIndex}$ and $\downProbaVar{\ocState}{\ocCount}{\ocStateB}{0}{\powerIndex}$.
  For all $\ocState$, $\ocStateB\in\ocStateSpace$ and all $\weightVal\in\{-1, 0, 1\}$, let $\ocTrans^\interval(\ocState)(\ocStateB, \weightVal) = \sum_{\ocAction\in\ocActionSpace(\ocState), \weight(\ocState, \ocAction)=\weightVal}\strat(\ocState, 1)(\ocAction)\cdot\ocTrans(\ocState, \ocAction)(\ocStateB)$.

  Consider the system of equations given by, for all $\ocState, \ocStateB\in\ocStateSpace$:
  \begin{align}
    \upProbaVar{\ocState}{1}{\ocStateB}{2}{0} & =
    \ocTrans^\interval(\ocState)(\ocStateB, 1) +
    \sum_{\ocStateC\in\ocStateSpace} \ocTrans^\interval(\ocState)(\ocStateC, 0)\cdot\upProbaVar{\ocStateC}{1}{\ocStateB}{2}{0}\label{equation:transitions:up:zero} \\
    \downProbaVar{\ocState}{1}{\ocStateB}{0}{0} & =
    \ocTrans^\interval(\ocState)(\ocStateB, -1) +
    \sum_{\ocStateC\in\ocStateSpace} \ocTrans^\interval(\ocState)(\ocStateC, 0)\cdot\downProbaVar{\ocStateC}{1}{\ocStateB}{0}{0}\nonumber
  \end{align}
  and for all $\powerIndex\in\integerInterval{1, \powerMax-1}$,
  \begin{align}
    \begin{split}
      \upProbaVar{\ocState}{2^{\powerIndex-1}}{\ocStateB}{2^{\powerIndex+1}}{\powerIndex} & = 
      \sum_{\ocStateC\in\ocStateSpace}
      \upProbaVar{\ocState}{2^{\powerIndex-1}}{\ocStateC}{2^{\powerIndex}}{\powerIndex-1}
      \cdot
      \upProbaVar{\ocStateC}{2^{\powerIndex}}{\ocStateB}{2^{\powerIndex+1}}{\powerIndex},
    \end{split}\label{equation:transitions:up:first}\\
    \begin{split}
      \upProbaVar{\ocState}{2^{\powerIndex}}{\ocStateB}{2^{\powerIndex+1}}{\powerIndex} & =      
      \sum_{\ocStateC\in\ocStateSpace}\bigg(
      \upProbaVar{\ocState}{2^{\powerIndex-1}}{\ocStateC}{2^{\powerIndex}}{\powerIndex-1}
      \cdot 
      \upProbaVar{\ocStateC}{3\cdot2^{\powerIndex-1}}{\ocStateB}{2^{\powerIndex+1}}{\powerIndex} \\
      &
      +
      \downProbaVar{\ocState}{2^{\powerIndex-1}}{\ocStateC}{0}{\powerIndex-1}
      \cdot 
      \upProbaVar{\ocStateC}{2^{\powerIndex-1}}{\ocStateB}{2^{\powerIndex+1}}{\powerIndex}
      \bigg),
    \end{split}\label{equation:transitions:up:second} \\
    \begin{split}
      \upProbaVar{\ocState}{3\cdot 2^{\powerIndex-1}}{\ocStateB}{2^{\powerIndex+1}}{\powerIndex} & =
      \sum_{\ocStateC\in\ocStateSpace}\bigg(
      \downProbaVar{\ocState}{2^{\powerIndex-1}}{\ocStateC}{0}{\powerIndex-1}
      \cdot
      \upProbaVar{\ocStateC}{2^{\powerIndex}}{\ocStateB}{2^{\powerIndex+1}}{\powerIndex}
      \bigg)
      \\
      &
      + \upProbaVar{\ocState}{2^{\powerIndex-1}}{\ocStateB}{2^{\powerIndex}}{\powerIndex-1},
    \end{split}\label{equation:transitions:up:third}
  \end{align}
  \begin{align*}
    \downProbaVar{\ocState}{3\cdot 2^{\powerIndex-1}}{\ocStateB}{0}{\powerIndex} =
    & 
      \sum_{\ocStateC\in\ocStateSpace}
      \downProbaVar{\ocState}{2^{\powerIndex-1}}{\ocStateC}{0}{\powerIndex-1}
      \cdot
      \downProbaVar{\ocStateC}{2^{\powerIndex}}{\ocStateB}{0}{\powerIndex}, \\
    \downProbaVar{\ocState}{2^{\powerIndex}}{\ocStateB}{0}{\powerIndex} =
    &
      \sum_{\ocStateC\in\ocStateSpace}\bigg(
      \downProbaVar{\ocState}{2^{\powerIndex-1}}{\ocStateC}{0}{\powerIndex-1}
      \cdot 
      \downProbaVar{\ocStateC}{2^{\powerIndex-1}}{\ocStateB}{0}{\powerIndex} \\
    &
      +
      \upProbaVar{\ocState}{2^{\powerIndex-1}}{\ocStateC}{2^{\powerIndex}}{\powerIndex-1}
      \cdot 
      \downProbaVar{\ocStateC}{3\cdot 2^{\powerIndex-1}}{\ocStateB}{0}{\powerIndex}
      \bigg), \\
    \downProbaVar{\ocState}{2^{\powerIndex-1}}{\ocStateB}{0}{\powerIndex} =
    &
      \sum_{\ocStateC\in\ocStateSpace}\bigg(
      \upProbaVar{\ocState}{2^{\powerIndex-1}}{\ocStateC}{2^{\powerIndex}}{\powerIndex-1}
      \cdot
      \downProbaVar{\ocStateC}{2^{\powerIndex}}{\ocStateB}{0}{\powerIndex}
      \bigg) \\
    &
      + \downProbaVar{\ocState}{2^{\powerIndex-1}}{\ocStateB}{0}{\powerIndex-1}.
  \end{align*}
  The least non-negative solution of this system is obtained by substituting each variable $\upProbaVar{\ocState}{\ocCount}{\ocStateB}{2^{\powerIndex+1}}{\powerIndex}$ by $\upProba{\ocState}{\ocCount}{\ocStateB}{2^{\powerIndex+1}}{\powerIndex}$ and $\downProbaVar{\ocState}{\ocCount}{\ocStateB}{0}{\powerIndex}$ by $\downProba{\ocState}{\ocCount}{\ocStateB}{0}{\powerIndex}$.
\end{theorem}

\begin{proof}
  Occurrences of $\proba$ in this proof refer to $\ocmdpFin{\ocmdp}{\counterUB}$, thus we omit it from the notation.
  We show a claim to shorten our arguments.
  Let $\ocConfig$, $\ocConfig'\in\ocStateSpace\times\integerInterval{\counterUB}$.
  Let $\histPart\subseteq\histSet{\ocmdpFin{\ocmdp}{\counterUB}}$ be a prefix-free set of histories starting in $\ocConfig$. Assume that there exist two prefix-free sets of histories $\histPart^{(1)}$ and $\histPart^{(2)}$ such that the last (resp.~first) configuration of all elements of $\histPart^{(1)}$ (resp.~$\histPart^{(2)}$) is $\ocConfig'$ and we have $\histPart = \{\histConcat{\hist_1}{\hist_2}\mid\hist_1\in\histPart^{(1)},\,\hist_2\in\histPart^{(2)}\}$.
  Then it holds that
  \begin{equation}\label{equation:transitions:splitting}
    \probaG{\ocConfig}{\strat}(\cyl{\histPart}) = \probaG{\ocConfig}{\strat}(\cyl{\histPart^{(1)}})\cdot\probaG{\ocConfig'}{\strat}(\cyl{\histPart^{(2)}}).
  \end{equation}
  Equation~\eqref{equation:transitions:splitting} can be proven as follows:
  \begin{align*}
    \probaG{\ocConfig}{\strat}(\cyl{\histPart})
    & = \sum_{\hist_1\in\histPart^{(1)}} \sum_{\hist_2\in\histPart^{(2)}}\probaG{\ocConfig}{\strat}(\cyl{\histConcat{\hist_1}{\hist_2}}) \\
    & = \sum_{\hist_1\in\histPart^{(1)}} \sum_{\hist_2\in\histPart^{(2)}}
      \probaG{\ocConfig}{\strat}(\hist_1)\cdot
      \probaG{\ocConfig'}{\strat}(\hist_2) \\
    & = \bigg(
      \sum_{\hist_1\in\histPart^{(1)}}
      \probaG{\ocConfig}{\strat}(\cyl{\hist_1})
      \bigg)\cdot
      \bigg(
      \sum_{\hist_2\in\histPart^{(2)}}
      \probaG{\ocConfig}{\strat}(\cyl{\hist_2})
      \bigg) \\
    & =\probaG{\ocConfig}{\strat}(\cyl{\histPart^{(1)}})\cdot
      \probaG{\ocConfig'}{\strat}(\cyl{\histPart^{(2)}})
  \end{align*}
  The first line follows from $\histPart$ being prefix-free.
  The second line is obtained from the definition of $\probaG{\ocConfig}{\strat}$, using the fact that $\strat$ is memoryless.
  We obtain the third line by algebraic manipulations and the last one using the fact that $\histPart^{(1)}$ and $\histPart^{(2)}$ are prefix-free.

  We now prove the theorem.
  We start by proving that the asserted solution is a solution of the system.
  We only verify Equations~\eqref{equation:transitions:up:zero}--\eqref{equation:transitions:up:third}, i.e., the equations in which the left-hand side of the equation has a variable with the arrow $\nearrow$.
  Arguments for the others are analogous.
  
  Let $\ocState, \ocStateB\in\ocStateSpace$.
  First, we consider the case $\powerIndex = 0$.
  We recall that $\upHistSet{\ocState}{1}{\ocStateB}{2}{0}$ is prefix-free.
  We partition $\upHistSet{\ocState}{1}{\ocStateB}{2}{0}$ into two sets $\histPart$ and $\histPart'$ such that $\histPart$ is the set of histories starting in $(\ocState, 1)$ whose second configuration is $(\ocStateB, 0)$ and $\histPart'= \upHistSet{\ocState}{1}{\ocStateB}{2}{0}\setminus\histPart$.
  For all histories of $\histPart'$, their second configuration has counter value $1$.
  We rewrite $\probaG{(\ocState, 1)}{\strat}(\cyl{\histPart})$ and $\probaG{(\ocState, 1)}{\strat}(\cyl{\histPart'})$ to prove the desired equality.
  On the one hand, we have
  \[\probaG{(\ocState, 1)}{\strat}(\cyl{\histPart}) =
    \sum_{\substack{\ocAction\in\ocActionSpace(\ocState) \\ \weight(\ocState,\ocAction)=1}}
    \probaG{(\ocState, 1)}{\strat}(\cyl{(\ocState, 1)\ocAction(\ocStateB, 2)}) =
    \ocTrans^\interval(\ocState, 1)(\ocStateB).
  \]
  For the other set, we partition $\histPart'$ according to the second configuration of the histories.
  We further partition the resulting sets following the first action and apply Equation~\eqref{equation:transitions:splitting} to obtain
  \begin{align*}
    \probaG{(\ocState, 1)}{\strat}(\cyl{\histPart'})
    & =
      \sum_{\ocStateC\in\ocStateSpace}
      \sum_{\substack{\ocAction\in\ocActionSpace(\ocState) \\ \weight(\ocState,\ocAction)=0}}
    \sum_{\hist\in\upHistSet{\ocStateC}{1}{\ocStateB}{2}{0}}
    \probaG{(\ocState, 1)}{\strat}(\cyl{\histConcat{(\ocState, 1)\ocAction(\ocStateC, 1)}{\hist}}) \\
    & =
      \sum_{\ocStateC\in\ocStateSpace}
      \bigg(
      \sum_{\substack{\ocAction\in\ocActionSpace(\ocState) \\ \weight(\ocState,\ocAction)=0}}
    \strat(\ocState, 1)(\ocAction)\cdot\delta(\ocState, \ocAction)(\ocStateC)\bigg)\cdot    
    \probaG{(\ocStateC, 1)}{\strat}(\cyl{\upHistSet{\ocStateC}{1}{\ocStateB}{2}{0}}) \\
    & = \sum_{\ocStateC\in\ocStateSpace}
      \ocTrans^\interval(\ocState, 0)(\ocStateC)\cdot\upProba{\ocStateC}{1}{\ocStateB}{2}{0}.
  \end{align*}
  Equation~\eqref{equation:transitions:up:zero} thus follows from the above and
  \[\upProba{\ocState}{1}{\ocStateB}{2}{0} =
    \probaG{(\ocState, 1)}{\strat}(\cyl{\histPart}) +
    \probaG{(\ocState, 1)}{\strat}(\cyl{\histPart'}) .
  \]
  This ends the case where $\powerIndex= 0$.

  Let $\powerIndex\geq 1$.
  We start by considering Equation~\eqref{equation:transitions:up:first}.
  All histories in $\upHistSet{\ocState}{2^{\powerIndex-1}}{\ocStateB}{2^{\powerIndex+1}}{\powerIndex}$ have a configuration with counter value $2^\powerIndex$.
  We let $(\upPart{\ocStateC})_{\ocStateC\in\ocStateSpace}$ be a partition of $\upHistSet{\ocState}{2^{\powerIndex-1}}{\ocStateB}{2^{\powerIndex+1}}{\powerIndex}$ based on the state of the first configuration with counter value $2^\powerIndex$ that is reached.
  For all $\ocStateC\in\ocStateSpace$ and all $\hist\in\upPart{\ocStateC}$, we let $\hist_1$ and $\hist_2$ such that $\hist = \histConcat{\hist_1}{\hist_2}$ where $\hist_1$ is the prefix of $\hist$ up to the first occurrence of $(\ocStateC, 2^{\powerIndex})$, and let, for $i\in\{1, 2\}$, $\upPart{\ocStateC}^{(i)} = \{\hist_i\mid \histConcat{\hist_1}{\hist_2}\in\upPart{\ocStateC}\}$.
  We have $\upPart{\ocStateC}^{(1)} = \upHistSet{\ocState}{2^{\powerIndex-1}}{\ocStateC}{2^\powerIndex}{\powerIndex-1}$ and $\upPart{\ocStateC}^{(2)} = \upHistSet{\ocState}{2^{\powerIndex}}{\ocStateC}{2^{\powerIndex+1}}{\powerIndex}$ by construction.
  We conclude that Equation~\eqref{equation:transitions:up:first} is satisfied by the candidate solution via the following equations (the second line uses Equation~\eqref{equation:transitions:splitting}):
  \begin{align*}
    \upProba{\ocState}{2^{\powerIndex-1}}{\ocStateB}{2^{\powerIndex+1}}{\powerIndex}
    & = \sum_{t\in\ocStateSpace}\probaG{(\ocState, 2^{\powerIndex-1})}{\strat}(\cyl{\upPart{\ocStateC}}) \\
& = \sum_{t\in\ocStateSpace}
      \probaG{(\ocState, 2^{\powerIndex-1})}{\strat}(\cyl{\upPart{\ocStateC}^{(1)}})\cdot
      \probaG{(\ocStateC, 2^{\powerIndex})}{\strat}(\cyl{\upPart{\ocStateC}^{(2)}})
    \\
& =
      \sum_{t\in\ocStateSpace}
      \upProba{\ocState}{2^{\powerIndex-1}}{\ocStateC}{2^{\powerIndex}}{\powerIndex-1}
      \cdot
      \upProba{\ocStateC}{2^{\powerIndex}}{\ocStateB}{2^{\powerIndex+1}}{\powerIndex}.
  \end{align*}
  
  We now move on to Equation~\eqref{equation:transitions:up:second}.
  We partition $\upHistSet{\ocState}{2^\powerIndex}{\ocStateB}{2^{\powerIndex+1}}{\powerIndex}$ as follows.
  Let $\ocStateC\in\ocStateSpace$.
  We let $\upPart{\ocStateC}$ (resp.~$\downPart{\ocStateC}$) denote the subset of $\upHistSet{\ocState}{2^\powerIndex}{\ocStateB}{2^{\powerIndex+1}}{\powerIndex}$ containing the histories such that the first configuration with counter value $3\cdot 2^{\powerIndex-1}$ (resp.~$2^{\powerIndex-1}$) that is visited has state $\ocStateC$ and no prior configuration has a counter value in $\{3\cdot 2^{\powerIndex-1},2^{\powerIndex-1}\}$.
  The sets $\downPart{\ocStateC}$ and $\upPart{\ocStateC}$, $\ocStateC\in\ocStateSpace$ partition, $\upHistSet{\ocState}{2^\powerIndex}{\ocStateB}{2^{\powerIndex+1}}{\powerIndex}$.
  Indeed, all of these sets are disjoint by definition and any history from $(\ocState, 2^\powerIndex)$ to $(\ocStateB, 2^{\powerIndex+1})$ must traverse a configuration with counter value $3\cdot 2^{\powerIndex-1}$.
  Similarly to above (for Equation~\eqref{equation:transitions:up:first}), for all $\ocStateC\in\ocStateSpace$ and all $\hist\in\upPart{\ocStateC}$ (resp.~$\downPart{\ocStateC}$), we let $\hist=\histConcat{\hist_1}{\hist_2}$ such that $\hist_1$ ends in the configuration witnessing that $\hist\in\upPart{\ocStateC}$ (resp.~$\downPart{\ocStateC}$).
  For $i\in\{1, 2\}$, we let $\upPart{\ocStateC}^{(i)}= \{\hist_i\mid\histConcat{\hist_1}{\hist_2}\in\upPart{\ocStateC}\}$ and $\downPart{\ocStateC}^{(i)}= \{\hist_i\mid\histConcat{\hist_1}{\hist_2}\in\downPart{\ocStateC}\}$.
  By applying Equation~\eqref{equation:transitions:splitting}, we obtain:
    \begin{equation*}
    \begin{aligned}
      \upProba{\ocState}{2^{\powerIndex}}{\ocStateB}{2^{\powerIndex+1}}{\powerIndex} =
      &
      \sum_{\ocStateC\in\ocStateSpace}
      \probaG{(\ocState, 2^\powerIndex)}{\strat}(\cyl{\upPart{\ocStateC}^{(1)}})
      \cdot
      \probaG{(\ocStateC, 3\cdot 2^{\powerIndex-1})}{\strat}(\cyl{\upPart{\ocStateC}^{(2)}})
      \\
      & +
      \sum_{\ocStateC\in\ocStateSpace}
      \probaG{(\ocState, 2^\powerIndex)}{\strat}(\cyl{\downPart{\ocStateC}^{(1)}})
      \cdot
      \probaG{(\ocStateC, 2^{\powerIndex-1})}{\strat}(\cyl{\downPart{\ocStateC}^{(2)}})
    \end{aligned}
  \end{equation*}

  We now prove that the cylinder probabilities match the terms in Equation~\eqref{equation:transitions:up:second}.
  Let $\ocStateC\in\ocStateSpace$.
  We have
  $\probaG{(\ocStateC, 3\cdot 2^{\powerIndex-1})}{\strat}(\cyl{\upPart{\ocStateC}^{(2)}}) =
  \upProba{\ocStateC}{3\cdot 2^{\powerIndex-1}}{\ocStateB}{2^{\powerIndex+1}}{\powerIndex}$ and
  $\probaG{(\ocStateC, 2^{\powerIndex-1})}{\strat}(\cyl{\downPart{\ocStateC}^{(2)}}) =
  \upProba{\ocStateC}{2^{\powerIndex-1}}{\ocStateB}{2^{\powerIndex+1}}{\powerIndex}$
  because $\upPart{\ocStateC}^{(2)}$ and $\downPart{\ocStateC}^{(2)}$ are respectively the sets $\upHistSet{\ocStateC}{3\cdot 2^{\powerIndex-1}}{\ocStateB}{2^{\powerIndex+1}}{\powerIndex}$ and
  $\upHistSet{\ocStateC}{2^{\powerIndex-1}}{\ocStateB}{2^{\powerIndex+1}}{\powerIndex}$.

  The sets $\upPart{\ocStateC}^{(1)}$ and $\downPart{\ocStateC}^{(1)}$ do not directly match relevant sets of histories as above.
  However, there are bijections from $\upPart{\ocStateC}^{(1)}$ to $\upHistSet{\ocState}{2^{\powerIndex-1}}{\ocStateC}{2^\powerIndex}{\powerIndex-1}$ and from $\downPart{\ocStateC}^{(1)}$ to $\downHistSet{\ocState}{2^{\powerIndex-1}}{\ocStateC}{0}{\powerIndex-1}$.
  Both bijections map a history to the history obtained by subtracting $2^{\powerIndex-1}$ to the counter values in all configurations along the history.
  All counter values in a history in $\upPart{\ocStateC}^{(1)}$ or $\downPart{\ocStateC}^{(1)}$ and its image lie in the interval $\interval$.
  Therefore, for all $\hist_1\in\upPart{\ocStateC}^{(1)}\cup\downPart{\ocStateC}^{(1)}$ with $\ocConfig$ as its first configuration, if its image by the relevant bijection is $\hist'_1$ and $\hist'_1$ starts in $\ocConfig'$, then $\proba^{\strat}_{\ocConfig}(\cyl{\hist_1}) = \proba^{\strat}_{\ocConfig'}(\cyl{\hist'_1})$.
  We conclude that
  $\probaG{(\ocState, 2^{\powerIndex})}{\strat}(\cyl{\upPart{\ocStateC}^{(1)}}) =
  \upProba{\ocState}{2^{\powerIndex-1}}{\ocStateC}{2^{\powerIndex}}{\powerIndex-1}$ and
  $\probaG{(\ocState, 2^\powerIndex)}{\strat}(\cyl{\downPart{\ocStateC}^{(1)}}) =
  \downProba{\ocState}{2^{\powerIndex-1}}{\ocStateC}{0}{\powerIndex-1}$ (this argument is detailed further in the proof of Lemma~\ref{lemma:abstraction:transitions:symmetry}).
  We have shown that the asserted solution satisfies Equation~\eqref{equation:transitions:up:second}.
  
  We now move on to Equation~\eqref{equation:transitions:up:third}.
  We follow the same scheme as above, i.e.,  we partition $\upHistSet{\ocState}{3\cdot 2^{\powerIndex-1}}{\ocStateB}{2^{\powerIndex+1}}{\powerIndex}$.
  First, we let $\upPart{\ocStateB}$ be the subset with all histories that never hit counter value $2^\powerIndex$.
  For any $\ocStateC\in\ocStateSpace$, we let $\downPart{\ocStateC}$ be the subset of $\upHistSet{\ocState}{3\cdot 2^{\powerIndex-1}}{\ocStateB}{2^{\powerIndex+1}}{\powerIndex}\setminus\upPart{\ocStateB}$ consisting of histories that reach counter value $2^\powerIndex$ for the first time in a configuration with state $\ocStateC$.
  The sets $\downPart{\ocStateC}$, $\ocStateC\in\ocStateSpace$, and $\upPart{\ocStateB}$ partition $\upHistSet{\ocState}{3\cdot 2^{\powerIndex-1}}{\ocStateB}{2^{\powerIndex+1}}{\powerIndex}$.
  As above, for any $\ocStateC\in\ocStateSpace$ and $\hist\in\downPart{\ocStateC}$, we write $\hist=\histConcat{\hist_1}{\hist_2}$ such that $\hist_1$ is the prefix of $\hist$ up to the first occurrence of $(\ocStateC, 2^\powerIndex)$.
  For $i\in\{1, 2\}$, we let $\downPart{\ocStateC}^{(i)}= \{\hist_i\mid\histConcat{\hist_1}{\hist_2}\in\downPart{\ocStateC}\}$.
  Like before, we obtain, from Equation~\eqref{equation:transitions:splitting},
  \begin{equation*}
    \begin{aligned}
      \upProba{\ocState}{3\cdot 2^{\powerIndex-1}}{\ocStateB}{2^{\powerIndex+1}}{\powerIndex} =
      &
      \sum_{\ocStateC\in\ocStateSpace}
      \probaG{(\ocState, 3\cdot 2^{\powerIndex-1})}{\strat}(\cyl{\downPart{\ocStateC}^{(1)}})
      \cdot
      \probaG{(\ocStateC,  2^{\powerIndex})}{\strat}(\cyl{\downPart{\ocStateC}^{(2)}})
      \\
      & +
      \probaG{(\ocState, 3\cdot 2^\powerIndex)}{\strat}(\cyl{\upPart{\ocStateB}}).
    \end{aligned}
  \end{equation*}
  Let $\ocStateC\in\ocStateSpace$.
  It follows from $\downPart{\ocStateC}^{(2)} = \upHistSet{\ocStateC}{2^\powerIndex}{\ocStateB}{2^{\powerIndex+1}}{\powerIndex}$ that
  $\probaG{(\ocStateC,  2^{\powerIndex})}{\strat}(\cyl{\downPart{\ocStateC}^{(2)}}) =
  \upProba{\ocStateC}{2^\powerIndex}{\ocStateB}{2^{\powerIndex+1}}{\powerIndex}$.
  We can adapt the bijection-based argument used for Equation~\eqref{equation:transitions:up:second} to conclude that
  $\probaG{(\ocState, 3\cdot 2^{\powerIndex-1})}{\strat}(\cyl{\downPart{\ocStateC}^{(1)}})=
  \downProba{\ocState}{2^{\powerIndex-1}}{\ocStateC}{0}{\powerIndex-1}$ and
  $\proba_{(\ocState, 3\cdot 2^{\powerIndex-1})}^{\strat}(\cyl{\upPart{\ocStateB}})=
  \upProba{\ocState}{2^{\powerIndex-1}}{\ocStateB}{2^{\powerIndex}}{\powerIndex-1}$.
  This shows that Equation~\eqref{equation:transitions:up:third} is verified by the asserted solution, and ends the argument that all equations hold.

  It remains to show that the asserted solution is the least non-negative solution of the system.
  Once again, we only consider the case of variables with ascending arrows as the other case can be handled similarly.
  We fix an arbitrary non-negative solution of the system.
  We denote its component corresponding to a variable $\varTrans$ by $\varTrans^\star$.

  All probabilities in the asserted solution can be written as the probability of a cylinder of a set of histories.
  In particular, these probabilities can be approximated by only considering the histories with at most $\indexLast$ actions (for $\indexLast\in\IN$).
  It suffices therefore to show that each approximation is lesser or equal to the fixed arbitrary solution to end the proof.
  
  For all $\indexLast\in\IN$, $\ocState, \ocStateB\in\ocStateSpace$, $\powerIndex\in\integerInterval{\powerMax-1}$ and $\ocCount\in\{2^{\powerIndex-1}, 2^{\powerIndex}, 3\cdot 2^{\powerIndex-1}\}$ if $\powerIndex\neq 0$ and $\ocCount = 1$ otherwise, we let $\upProba{\ocState}{\ocCount}{\ocStateB}{2^{\powerIndex+1}}{\powerIndex}^{\leq\indexLast} = \probaG{(\ocState, \ocCount)}{\strat}(\cyl{\histPart^{\leq\indexLast}})$ where $\histPart^{\leq\indexLast}$ is the subset of $\upHistSet{\ocState}{\ocCount}{\ocStateB}{2^{\powerIndex+1}}{\powerIndex}$ containing all histories with at most $\indexLast$ actions
  We define $\downProba{\ocState}{\ocCount}{\ocStateB}{0}{\powerIndex}^{\leq\indexLast}$ similarly.

  Let $\ocState$ and $\ocStateB\in\ocStateSpace$.
  We use nested induction arguments in the remainder of the proof: an outer induction on $\powerIndex$ and an inner induction on $\indexLast$.

  First, we deal with the case $\powerIndex = 0$.
  Let $\indexLast\in\IN$.
  For the base case $\indexLast=0$, we have $\upProba{\ocState}{1}{\ocStateB}{2}{0}^{\leq 0} = 0 \leq \upProbaVar{\ocState}{1}{\ocStateB}{2}{0}^\star$ because we consider a non-negative solution.
  We now assume that $\upProba{\ocState}{1}{\ocStateB}{2}{0}^{\leq \indexLast-1} \leq \upProbaVar{\ocState}{1}{\ocStateB}{2}{0}^\star$ by induction.
  We can apply the reasoning used when considering Equation~\eqref{equation:transitions:up:zero} in the first part of the proof (taking in account the length of histories) and then apply the induction hypothesis to obtain:
  \begin{align*}
    \upProba{\ocState}{1}{\ocStateB}{2}{0}^{\leq\indexLast}
    & =\ocTrans^\interval(\ocState)(\ocStateB, 1) +
    \sum_{\ocStateC\in\ocStateSpace} \ocTrans^\interval(\ocState)(\ocStateC, 0)\cdot
      \upProba{\ocStateC}{1}{\ocStateB}{2}{0}^{\leq\indexLast-1} \\
    & \leq \ocTrans^\interval(\ocState)(\ocStateB, 1) +
      \sum_{\ocStateC\in\ocStateSpace}
      \ocTrans^\interval(\ocState)(\ocStateC, 0)\cdot
      \upProbaVar{\ocStateC}{1}{\ocStateB}{2}{0}^\star \\
    & = \upProbaVar{\ocState}{1}{\ocStateB}{2}{0}^\star.
  \end{align*}
  This closes the proof for the case $\powerIndex=0$.

  Next, let $\powerIndex\geq 1$.
  We assume, by induction on $\powerIndex$, that we have shown that for all $\ocStateC, \ocStateD\in\ocStateSpace$, we have $\upProba{\ocStateC}{2^{\powerIndex-1}}{\ocStateD}{2^{\powerIndex}}{\powerIndex-1}\leq \upProbaVar{\ocStateC}{2^{\powerIndex-1}}{\ocStateD}{0}{\powerIndex-1}^\star$ and $\downProba{\ocStateC}{2^{\powerIndex-1}}{\ocStateD}{2^{\powerIndex}}{\powerIndex-1}\leq \downProbaVar{\ocStateC}{2^{\powerIndex-1}}{\ocStateD}{0}{\powerIndex-1}^\star$.
  The base case $\indexLast = 0$ of the inner induction is direct because for all $\ocCount\in\{2^{\powerIndex-1}, 2^{\powerIndex}, 3\cdot 2^{\powerIndex-1}\}$, we have $\upProba{\ocState}{\ocCount}{\ocStateB}{2^{\powerIndex+1}}{\powerIndex}^{\leq 0}=0$.

  We assume by induction that $\upProba{\ocStateC}{\ocCount}{\ocStateD}{2^{\powerIndex+1}}{\powerIndex}^{\leq\indexLast}\leq\upProbaVar{\ocStateC}{\ocCount}{\ocStateD}{2^{\powerIndex+1}}{\powerIndex}^\star$ for all $\ocStateC, \ocStateD\in\ocStateSpace$ and all $\ocCount\in \{2^{\powerIndex-1}, 2^{\powerIndex}, 3\cdot 2^{\powerIndex-1}\}$.
  All required inequalities are obtained by an adaptation of the argument used in the first part of the proof for Equations~\eqref{equation:transitions:up:first},~\eqref{equation:transitions:up:second} and~\eqref{equation:transitions:up:third}, i.e., partitioning the set of histories while taking in account the length of histories and invoking Equation~\eqref{equation:transitions:splitting}.
  For this reason, we omit some details.
  From configuration $(\ocState, 2^{\powerIndex-1})$, we obtain that
  \begin{align*}
    \upProba{\ocState}{2^{\powerIndex-1}}{\ocStateB}{2^{\powerIndex+1}}{\powerIndex}^{\leq\indexLast} & \leq\sum_{\ocStateC\in\ocStateSpace}
    \upProba{\ocState}{2^{\powerIndex-1}}{\ocStateC}{2^{\powerIndex}}{\powerIndex-1}
    \cdot
    \upProba{\ocStateC}{2^{\powerIndex}}{\ocStateB}{2^{\powerIndex+1}}{\powerIndex}^{\leq\indexLast-1}.
  \end{align*}
  By the induction hypotheses and the fact we are dealing with a solution of the system, we obtain $\upProba{\ocState}{2^{\powerIndex-1}}{\ocStateB}{2^{\powerIndex+1}}{\powerIndex}^{\leq\indexLast} \leq \upProbaVar{\ocState}{2^{\powerIndex-1}}{\ocStateB}{2^{\powerIndex+1}}{\powerIndex}^\star$.
  Next, for configuration $(\ocState, 2^\powerIndex)$, we obtain that 
  \begin{align*}
    \upProba{\ocState}{2^{\powerIndex}}{\ocStateB}{2^{\powerIndex+1}}{\powerIndex}^{\leq\indexLast} \leq
    &
      \sum_{\ocStateC\in\ocStateSpace}\bigg(
      \upProba{\ocState}{2^{\powerIndex-1}}{\ocStateC}{2^{\powerIndex}}{\powerIndex-1}
      \cdot 
      \upProba{\ocStateC}{3\cdot2^{\powerIndex-1}}{\ocStateB}{2^{\powerIndex+1}}{\powerIndex}^{\leq\indexLast-1} \\
    &
      +
      \downProba{\ocState}{2^{\powerIndex-1}}{\ocStateC}{0}{\powerIndex-1}
      \cdot 
      \upProba{\ocStateC}{2^{\powerIndex-1}}{\ocStateB}{2^{\powerIndex+1}}{\powerIndex}^{\leq\indexLast-1}
      \bigg).
  \end{align*}
  It follows from the induction hypotheses and the fact we deal with a solution that $\upProba{\ocState}{2^{\powerIndex}}{\ocStateB}{2^{\powerIndex+1}}{\powerIndex}^{\leq\indexLast} \leq \upProbaVar{\ocState}{2^{\powerIndex}}{\ocStateB}{2^{\powerIndex+1}}{\powerIndex}^\star$.
  Finally, for configuration $(\ocState, 3\cdot 2^{\powerIndex-1})$, we have
  \begin{align*}
    \upProba{\ocState}{3\cdot 2^{\powerIndex-1}}{\ocStateB}{2^{\powerIndex+1}}{\powerIndex}^{\leq\indexLast} \leq
    &
      \sum_{\ocStateC\in\ocStateSpace}\bigg(
      \downProba{\ocState}{2^{\powerIndex-1}}{\ocStateC}{0}{\powerIndex-1}
      \cdot
      \upProba{\ocStateC}{2^{\powerIndex}}{\ocStateB}{2^{\powerIndex+1}}{\powerIndex}^{\leq\indexLast-1}
      \bigg)
    \\
    &
      + \upProba{\ocState}{2^{\powerIndex-1}}{\ocStateB}{2^{\powerIndex}}{\powerIndex-1}.
  \end{align*}
  The induction hypotheses imply that $\upProba{\ocState}{3\cdot 2^{\powerIndex-1}}{\ocStateB}{2^{\powerIndex+1}}{\powerIndex}^{\leq\indexLast}\leq \upProbaVar{\ocState}{3\cdot 2^{\powerIndex-1}}{\ocStateB}{2^{\powerIndex+1}}{\powerIndex}^\star$.

  We have shown that the asserted solution is the least non-negative solution of the system.
\end{proof}

We now analyse the size of the equation system of Theorem~\ref{theorem:equations:transitions}.
There are as many equations as there are variables.
There are $2\cdot|\ocStateSpace|^2$ equations in the system where the variable of the left-hand side is indexed by $0$, and, for all $\powerIndex\in\integerInterval{1, \powerMax-1}$, there are $6\cdot|\ocStateSpace|^2$ equations in the system where the variable of the left-hand side is indexed by $\powerIndex$.
We can also show that these equations have length polynomial in $|\ocStateSpace|$ and $|\ocActionSpace|$.
We obtain the following result.
\begin{lemma}\label{lemma:equations:bounded:size}
  The equation system of Theorem~\ref{theorem:equations:transitions} has $2\cdot|\ocStateSpace|^2\cdot(3\powerMax-2)$ variables and equations.
  Its equations have length polynomial in $|\ocStateSpace|$ and $|\ocActionSpace|$.
\end{lemma}
\begin{proof}
  The argument regarding the number of variables and equations is given above.
  We thus provide an analysis of the length of the equations.
  We analyse Equations~\eqref{equation:transitions:up:zero}--\eqref{equation:transitions:up:third} from Theorem~\ref{theorem:equations:transitions}. A similar analysis applies to the other equations.
  We need only comment on the right-hand side of each equation, as the left-hand side contains a single variable.

  We start with Equation~\eqref{equation:transitions:up:zero}.
  If we rewrite its right-hand side as a sum of products (we substitute references to $\ocTrans^\interval$ by the corresponding sum), we obtain a sum of no more than $|\ocStateSpace|\cdot|\ocActionSpace|$ products of at most three variables or constants.
  For Equations~\eqref{equation:transitions:up:first}--\eqref{equation:transitions:up:third}, we observe that their right-hand side are respectively sums of no more than $2|\ocStateSpace|$ products of two variables.
  This ends the proof.
\end{proof}

Theorem~\ref{theorem:equations:transitions} provides a system of equations that may not have a unique solution.
We describe how to alter this system to have a unique solution based on the supports of the distributions assigned by $\strat$.

We rely on the Markov chains described in Figure~\ref{figure:equations:bounded}.
By Theorem~\ref{theorem:equations:transitions}, the transition probabilities of $\compressChain$ are reachability probabilities in these Markov chains.
More precisely, the system of Theorem~\ref{theorem:equations:transitions} is a collection of systems for reachability probabilities in these Markov chains.
It follows that modifying the equation system of Theorem~\ref{theorem:equations:transitions} by setting all relevant probabilities to zero will ensure uniqueness of the solution.

It remains to determine how to identify the probabilities that are zero in the least solution of the system.
As explained in Section~\ref{section:preliminaries}, the probabilities that are zero only depend on the transitions (with non-zero probability) between configurations.
Therefore, we need not compute the transition probabilities of the Markov chains (which would have an important computational cost, see Example~\ref{example:compression:main}) and need only determine the transitions qualitatively.

The idea of the procedure is to proceed gradually increasing the counter step size.
First, we can study the Markov chain for counter values $\{0, 1, 2\}$, as illustrated in Figure~\ref{figure:equations:bounded:one}, and perform a graph-based analysis to determine which probabilities to set to zero for outgoing transitions from $\ocStateSpace\times\{1\}$ in $\compressChain$.
Then, for all $\powerIndex\in\integerInterval{\powerMax-1}$, assuming that the non-zero transition probabilities in $\compressChain$ have been determined for configurations in $\ocStateSpace\times\{2^{\powerIndex-1}\}$, we can perform another graph-based analysis on the Markov chain described in Figure~\ref{figure:equations:bounded:induction} to determine the non-zero transition probabilities from $\ocStateSpace\times\{2^\powerIndex\}$ in $\compressChain$.

In this way, we obtain a procedure that runs in time polynomial in $|\ocStateSpace|$ and $\powerMax$: we perform a reachability analysis on one graph of size $3\cdot|\ocStateSpace|$ for the base case and on $\powerMax-1$ graphs of size $5\cdot|\ocStateSpace|$ for the other cases.
This analysis does not require the precise probabilities given by $\strat$, and it is sufficient to only know which actions are chosen with positive probabilities in $\ocStateSpace\times\interval$.
When given the precise probabilities, the system resulting from this procedure can be solved in polynomial time in the BSS model; by construction, its unique solution can be computed by solving $\powerMax$ linear systems.
We summarise this result in the following theorem.

\begin{theorem}\label{theorem:equations:transitions:unique}
    There exists an algorithm modifying the system of Theorem~\ref{theorem:equations:transitions} such that (i) the least solution of the original system is the unique solution of the modified one and (ii) the algorithm runs in time polynomial in $\powerMax$ and the representation size of $\ocmdp$.
  This algorithms only relies on the support of the distributions in the image of $\strat$ and not the precise probabilities.
  The resulting system can be solved in polynomial time in the BSS model.
\end{theorem}
\begin{proof}
  We first formalise the Markov chains of Figure~\ref{figure:equations:bounded}.
  The remainder of our argument is based on these $\powerMax$ Markov chains.
  We let $\mchain_0 = (\{\bot\}\cup(\ocStateSpace\times\{0, 1, 2\}), \mdpTrans_0^\interval)$ be the Markov chain such that all states in $\{\bot\}\cup(\ocStateSpace\times\{0, 2\})$ are absorbing and for all $\ocState,\ocStateB\in\ocStateSpace$ and $\weightVal\in\{-1, 0, 1\}$, $\mdpTrans_0^\interval((\ocState, 1))((\ocStateB, 1+\weightVal)) = \sum_{\ocAction\in\ocActionSpace(\ocState), \weight(\ocState, \ocAction)=\weightVal}\strat(\ocState, 1)(\ocAction)\cdot\ocTrans(\ocState, \ocAction)(\ocStateB)$ (unattributed probability goes to $\bot$).
  For all $\powerIndex\in\integerInterval{1, \powerMax-1}$, we let $\mchain_\powerIndex = (\{\bot\}\cup(\ocStateSpace\times\{0, 2^{\powerIndex-1}, 2^\powerIndex, 3\cdot 2^{\powerIndex-1}, 2^{\powerIndex+1}\}), \mdpTrans_\powerIndex^\interval)$ where the states in $\{\bot\}\cup(\ocStateSpace\times\{0, 2^{\powerIndex+1}\})$ are absorbing and, for all $\ocState, \ocStateB\in\ocStateSpace$ and $\ocCount\in\{2^{\powerIndex-1}, 2^\powerIndex, 3\cdot 2^{\powerIndex-1}\}$, we let $\mdpTrans_\powerIndex((\ocState, \ocCount))((\ocStateB, \ocCount - 2^{\powerIndex-1})) = \compressChainTrans((\ocState, 2^{\powerIndex-1}))((\ocStateB, 0))$ and $\mdpTrans_\powerIndex((\ocState, \ocCount))((\ocStateB, \ocCount + 2^{\powerIndex-1})) = \compressChainTrans((\ocState, 2^{\powerIndex-1}))((\ocStateB, 2^{\powerIndex}))$.

  We observe that for all $\ocStateB\in\ocStateSpace$ and all $\powerIndex\in\integerInterval{\powerMax-1}$, the subset of equations from Theorem~\ref{theorem:equations:transitions} with the variables of the form $\upProbaVar{\ocState}{\ocCount}{\ocStateB}{2^{\powerIndex+1}}{\powerIndex}$ (resp.~$\downProbaVar{\ocState}{\ocCount}{\ocStateB}{0}{\powerIndex}$) in the left-hand side coincides with a system for reachability probabilities in $\mchain_\powerIndex$ for target $\{(\ocStateB, 2^{\powerIndex+1})\}$ (resp.~$\{(\ocStateB, 0)\}$) when substituting variables indexed by $\powerIndex-1$ by their assignment in the least solution of the system.
  We devise an algorithm that individually modifies every such system so it has a unique solution.
  Because we are dealing with systems for reachability probabilities, we obtain a system with a unique solution by setting the variables whose assignment in the least solution of the system is zero to zero (see Section~\ref{section:preliminaries}).
  This set of variables can be determined using a qualitative reachability analysis of the Markov chains $\mchain_\powerIndex$.

  We analyse the Markov chains in order, i.e., we start with $\mchain_0$, then continue with $\mchain_1$ and so on.
  This is necessary: the transitions with non-zero probabilities in a Markov chain $\mchain_\powerIndex$ with $\powerIndex\geq 1$ is not known beforehand, but can be inferred from the analysis of $\mchain_{\powerIndex-1}$.
  We prove the following invariant of our procedure: after processing $\mchain_\powerIndex$, the least non-negative solution of the modified system is the least non-negative solution of the original system and all variables indexed by $\powerIndex'\leq\powerIndex$ have a unique valid assignment in any solution of the new system.
  We modify the system by adding constraints that are satisfied by the least non-negative solution of the original system.
  Thus, the first part of the invariant follows directly and we do not comment on it.
  
  The transition structure of the Markov chain $\mchain_0$ can be constructed directly as follows: there exists a transition from a state $(\ocState, 1)$ to a state $(\ocStateB, 1+\weightVal)$ if and only if there exists an action $\ocAction\in\supp{\strat(\ocState, 1)}$ such that $\weight(\ocState, \ocAction)=\weightVal$ and $\ocStateB\in\supp{\ocTrans(\ocState, \ocAction)}$ (in particular, the numerical probabilities do not matter).
  This yields a directed graph $\graph_0$ over $\ocStateSpace\times\{0, 1, 2\}$.
  For all $\ocState, \ocStateB\in\ocStateSpace$, we have $\upProba{\ocState}{1}{\ocStateB}{2}{0}=0$ (resp.~$\downProba{\ocState}{1}{\ocStateB}{0}{0}=0$) if and only if $(\ocStateB, 2)$ (resp.~$(\ocStateB, 0)$) cannot be reached from $(\ocState, 1)$ in $\graph_0$, and in this case, we add $\upProbaVar{\ocState}{1}{\ocStateB}{2}{0}=0$ (resp.~$\downProbaVar{\ocState}{1}{\ocStateB}{0}{0}=0$) to the system.
  After analysing $\mchain_0$, the invariant is satisfied.
  Indeed, following the addition of the new equations, there is only one possible assignment of the variables indexed by $0$ in any solution: all of these variables are involved in a Markov chain reachability probability system with a unique solution.

  We now let $\powerIndex\in\integerInterval{1, \powerMax-1}$ and assume that $\mchain_{\powerIndex-1}$ has been processed.
  We assume that the invariant holds by induction.
  Through the analysis of $\mchain_{\powerIndex-1}$, we know which transitions of $\mchain_{\powerIndex}$ have non-zero probability because these probabilities are reachability probabilities in $\mchain_{\powerIndex-1}$ by Theorem~\ref{theorem:equations:transitions}.
  We construct a directed graph $\graph_\powerIndex$ over the state space of $\mchain_\powerIndex$ similarly to above.
  Let $\ocConfig=(\ocState, \ocCount)\in\ocStateSpace\times\{2^{\powerIndex-1}, 2^\powerIndex, 3\cdot2^{\powerIndex-1}\}$ and $\ocStateB\in\ocStateSpace$.
  In $\graph_\powerIndex$, there is an edge from $\ocConfig$ to $(\ocStateB, \ocCount+2^{\powerIndex-1})$ if $\upProba{\ocState}{2^{\powerIndex-1}}{\ocStateB}{2^\powerIndex}{\powerIndex-1}>0$ and there is an edge from $\ocConfig$ to $(\ocStateB, \ocCount-2^{\powerIndex-1})$ if $\downProba{\ocState}{2^{\powerIndex-1}}{\ocStateB}{0}{\powerIndex-1}>0$.
Whether these probabilities are positive is known from the analysis of $\mchain_{\powerIndex-1}$.
  As above, we have $\upProba{\ocState}{\ocCount}{\ocStateB}{2^{\powerIndex+1}}{\powerIndex}=0$ (resp.~$\downProba{\ocState}{\ocCount}{\ocStateB}{0}{\powerIndex}=0$) if and only if $(\ocStateB, 2^{\powerIndex+1})$ (resp.~$(\ocStateB, 0)$) cannot be reached from $(\ocState, \ocCount)$ in $\graph_\powerIndex$, and in this case, we add $\upProbaVar{\ocState}{\ocCount}{\ocStateB}{2^{\powerIndex+1}}{\powerIndex}=0$ (resp.~$\downProbaVar{\ocState}{\ocCount}{\ocStateB}{0}{\powerIndex}=0$) to the system.

  We prove that the invariant is preserved after this iteration.
  By induction, all variables indexed by $\powerIndex-1$ have only one possible valid assignment.
  Given a solution of the system obtained after processing $\mchain_{\powerIndex}$, the variables indexed by $\powerIndex$ must satisfy an equation system with a unique solution, the coefficients of which are given by the unique valid assignment of the variables indexed by $\powerIndex-1$.
  It follows that there can only be one valuation for the variables indexed by $\powerIndex$ in any solution of this system.
  This, in addition to the inductive hypothesis, guarantees that the invariant holds after the analysis of $\mchain_\powerIndex$.
  In the end, after analysing $\mchain_{\powerMax-1}$, the invariant guarantees that the resulting system has a unique solution.

  To end the proof, it remains to show that the above algorithm respects the asserted complexity bounds.
  For all $\powerIndex\in\integerInterval{\powerMax-1}$, constructing the graph $\graph_\powerIndex$ takes time polynomial in the representation of $\ocmdp$.
  Indeed, for $\graph_0$, to find all successors of a configuration $(\ocState, 1)$, it suffices to iterate over all actions $\ocAction\in\supp{\strat(\ocState, 1)}$ and then build on the set of successors $\supp{\ocTrans(\ocState, \ocAction)}$.
  For the other graphs $\graph_\powerIndex$, their structure is inferred from the analysis of $\graph_{\powerIndex-1}$.
  Each graph $\graph_\powerIndex$ can be analysed in polynomial time by performing a backward reachability analysis from each configuration on the right (i.e., after the arrow) of a variable indexed by $\powerIndex$, and there are $2|\ocStateSpace|$ such configurations per graph.
  As we analyse $\powerMax$ graphs, the overall time required to implement the procedure above respects the announced complexity bounds.

  It remains to prove that the unique solution of the system provided by the procedure above can be computed in polynomial time in the BSS model.
  It suffices to solve linear systems for reachability probabilities in each of the Markov chains $\mchain_{\powerMax}$ for $\powerIndex\in\integerInterval{\powerMax}$.
  This can be done in polynomial time with unit-cost arithmetic: these Markov chains have no more than $5\cdot|\ocStateSpace|$ states each.
\end{proof}

\subsection{Finite representations of compressed Markov chains}\label{section:abstraction:finiteness}

We have not imposed any conditions on the memoryless strategy $\strat$: the compressed Markov chain $\compressChain$ can be defined without assuming that $\strat$ is an OEIS or a CIS.
However, for algorithmic purposes, we require that $\compressChain$ has a finite representation that is amenable to verification algorithms.
In this section, we focus on the representation of the state space of $\compressChain$, as the results of Section~\ref{section:abstraction:transitions} provide a finite representation of transition probabilities for each interval.

By construction, $\compressChain$ is finite if and only if $\intPart$ is finite.
Thus $\compressChain$ can only be finite when $\strat$ is an OEIS.
In the remainder of this section, our goal is to show that $\compressChain$ has a finite representation when $\strat$ is a CIS and $\intPart$ is periodic.
We assume that $\counterUB=\infty$ and that $\strat$ is a CIS that for the remainder of the section.
We let $\period$ denote a common period of $\strat$ and $\intPart$.
We let $\intPartB$ be the partition of $\integerInterval{1, \period}$ induced by $\intPart$.

We claim that $\compressChain$ is induced by a one-counter Markov chain $\cisChain=(\cisChainStateSpace, \cisChainTrans)$  where $\cisChainStateSpace = \compressChainStateSpace\cap(\{\bot\}\cup(\ocStateSpace\times\integerInterval{1, \period}))$  and $\cisChainTrans$ is described below.
We first explain the interpretation of configurations before giving intuition on $\cisChainTrans$.
Let $((\ocState, \ocCount), \ocCountB)$ be a configuration of $\ocChainFin{\cisChain}{\infty}$ such that $\ocCountB\geq 1$ or $\ocCount=\period$ (configurations that do not satisfy these conditions will be unreachable and we ignore them).
This configuration corresponds to the configuration $(\ocState, \period\cdot(\ocCount'-1) +\ocCount)\in\compressChainStateSpace$.
The counter value $\ocCount$ keeps track of where in the period we are and the counter value $\ocCount'$ indicates how many multiples of $\period$ the counter has exceeded.
This correspondence guarantees that the configuration $((\ocState, \period), 0)$ of $\ocChainFin{\cisChain}{\infty}$ represents the configuration $(\ocState, 0)\in\compressChainStateSpace$.

Transitions are defined so that successors in $\ocChainFin{\cisChain}{\infty}$ correspond to successors in $\compressChain$.
We formalise $\cisChainTrans\colon\cisChainStateSpace\to\dist{\cisChainStateSpace\times\{-1, 0, 1\}}$ as follows.
Like before, $\bot$ is absorbing and we give a weight of zero to its self-loop to ensure that we cannot terminate in $\bot$.
In other words, we set $\cisChainTrans(\bot)(\bot, 0)=1$.
Let $\ocConfig=(\ocState, \ocCount)\in\cisChainStateSpace$.
Each transition from $\ocConfig$ in $\compressChain$ to a state in $\cisChainStateSpace$ yields a transition with weight zero in $\cisChain$, i.e., for all $\ocConfig'\in\cisChainStateSpace$, we let $\cisChainTrans(\ocConfig)(\ocConfig', 0) = \compressChainTrans(\ocConfig)(\ocConfig')$.
In particular, all incoming transitions of $\bot$ have weight zero.
Any transition from $\ocConfig$ to a configuration $(\ocStateB, 0)$ in $\compressChain$ induces a transition from $\ocConfig$ to $(\ocStateB, \period)$ in $\cisChain$ with a weight of $-1$, i.e., we let $\cisChainTrans(\ocConfig)((\ocStateB, \period), -1) = \compressChainTrans(\ocConfig)((\ocStateB, 0))$.
Intuitively, in this case, we go back to the previous period.
Finally, any transition from $\ocConfig$ to the configuration $(\ocStateB, \period+1)$ in $\compressChain$ yields a transition with a weight of $1$ in $\cisChain$ from $\ocConfig$ to $(\ocStateB, 1)\in\cisChainStateSpace$ (this configuration is guaranteed to be in $\compressChainStateSpace$ because $1$ is the minimum of the first interval and $\intPart$ has period $\period$), i.e., we let $\cisChainTrans(\ocConfig)((\ocStateB, 1), 1) = \compressChainTrans(\ocConfig)((\ocStateB, \period+1))$.
Intuitively, in this case, we have passed a multiple of $\period$.
We obtain a well-defined transition function with the above: for all counter values $\ocCount$ of configurations in $\cisChainStateSpace$, the successor counter values of $\ocCount$ are a counter value of a configuration in $\cisChainStateSpace$, $0$ or $\period+1$, i.e., the upper and lower bound respectively of the intervals adjacent to $\integerInterval{1, \period}$ in $\intPart\cup\{\integerInterval{0}\}$.

We now show that the termination probabilities in $\compressChain$ and in $\ocChainFin{\cisChain}{\infty}$ match from all initial configurations with the correspondence outlined previously.

\begin{theorem}\label{theorem:cis:compression:ocmc}
  For all $(\ocState, \ocCount)\in\cisChainStateSpace\setminus\{\bot\}$ and $\ocCountB\in\IN$ such that $\ocCountB\geq1$ or $\ocCount=\period$ and all $\ocStateB\in\ocStateSpace$, we have $\probaMCverb{\compressChain}{(\ocState, \period\cdot(\ocCount'-1)+\ocCount)}(\reach{(\ocStateB, 0)}) = \probaMCverb{\ocChainFin{\cisChain}{\infty}}{((\ocState, \ocCount), \ocCount')}(\selectiveTermination{(\ocStateB, \period)})$.
\end{theorem}
\begin{proof}
  We define an injective mapping $f\colon\compressChainStateSpace\setminus\{\bot\}\to\cisChainStateSpace\times\IN$ such that, for any $\ocConfig=(\ocState, \ocCount)\in\compressChainStateSpace$, if $\ocCount$ is divisible by $\period$, we let $f(\ocConfig) = ((\ocState, \period), \frac{\ocCount}{\period})$, and otherwise, we let $f(\ocConfig) = ((\ocState, \ocCount\bmod\period), \lfloor\frac{\ocCount}{\period}\rfloor + 1)$.
  We observe that the image of $f$ is the set of configurations $((\ocState, \ocCount), \ocCountB)$ of $\ocChainFin{\cisChain}{\infty}$ such that $\ocCountB\geq1$ or $\ocCount=\period$.
  The configurations of $\ocChainFin{\cisChain}{\infty}$ with a state other than $\bot$ that are not in the image of $f$ have no incoming transitions in $\ocChainFin{\cisChain}{\infty}$ and are absorbing by construction.
  The statement of the theorem is equivalent to showing that for all $\ocConfig\in\compressChainStateSpace$ and all $\ocStateB\in\ocStateSpace$, we have $\probaMCverb{\compressChain}{\ocConfig}(\reach{(\ocStateB, 0)}) = \probaMCverb{\ocChainFin{\cisChain}{\infty}}{f(\ocConfig)}(\selectiveTermination{(\ocStateB, \period)})$.

  Let $[\cisChainTrans]^{\leq\infty}$ denote the transition function of $\ocChainFin{\cisChain}{\infty}$.
  The crux of the proof is to establish that for all $\ocConfig$, $\ocConfig'\in\compressChainStateSpace$, we have $\compressChainTrans(\ocConfig)(\ocConfig')=[\cisChainTrans]^{\leq\infty}(f(\ocConfig))(f(\ocConfig'))$.
  To refer to this property, we say that $f$ preserves transitions.
  Let $\ocConfig = (\ocState, \ocCount)\in\compressChainStateSpace$.
  If $\ocCount=0$, we have $\compressChainTrans(\ocConfig)(\ocConfig) = 1 = [\cisChainTrans]^{\leq\infty}(f(\ocConfig))(f(\ocConfig)) = [\cisChainTrans]^{\leq\infty}(((\ocState, \period), 0))((\ocState, \period), 0))$ since configurations with counter value zero are absorbing.
  We thus assume that $\ocCount>0$.

  We distinguish two cases below.
  First, we assume that $\period=1$, i.e., the strategy $\strat$ is counter-oblivious.
  It follows that for all $(\ocStateB, \ocCountB)\in\ocStateSpace\times\IN$, we have $f((\ocStateB, \ocCountB)) = ((\ocStateB, 1), \ocCountB)$.
  The successor counter values of $\ocCount$ are $\ocCount-1$ and $\ocCount+1$ because $\period=1$.
  Let $\ocStateB\in\ocStateSpace$, $\weightVal\in\{-1, 1\}$ and $\ocConfig'=(\ocStateB, \ocCount+\weightVal)$.
  By definition of $\cisChain$ and of $\compressChain$ (in particular, its periodic structure), we have
  \[
    [\cisChainTrans]^{\leq\infty}(f(\ocConfig))(f(\ocConfig'))
    = \cisChainTrans((\ocState, 1))((\ocStateB, 1), \weightVal)
    = \compressChainTrans((\ocState, 1))((\ocStateB, 1+\weightVal))
    = \compressChainTrans(\ocConfig)(\ocConfig').
  \]
  This ends the proof that $f$ preserves transitions when $\period=1$.
  
  Now, we assume that $\period > 1$.
  First, we assume that $\ocCount$ is divisible by $\period$, i.e., that $f(\ocConfig) = ((\ocState, \period), \frac{\ocCount}{\period})$.
  Successor counter values of $\ocCount$ are $\ocCount+1$ and $\ocCount-1$, since multiples of $\period$ are the maximum of their interval in $\intPart$.
  Let $\ocStateB\in\ocStateSpace$.
  First, we consider $\ocConfig'=(\ocStateB, \ocCount+1)$.
  In this case, we have $f(\ocConfig') = ((\ocStateB, 1), \frac{\ocCount}{\period}+1)$.
  By definition of $\cisChain$ and of $\compressChain$, we have
  \[
    [\cisChainTrans]^{\leq\infty}(f(\ocConfig))(f(\ocConfig'))
    = \cisChainTrans((\ocState, \period))((\ocStateB, 1), 1)
    = \compressChainTrans((\ocState, \period))((\ocStateB, \period+1))
    = \compressChainTrans(\ocConfig)(\ocConfig').
  \]
  Now, we consider $\ocConfig'=(\ocStateB, \ocCount-1)$.
  We obtain $f(\ocConfig') = ((\ocStateB, \period-1), \frac{\ocCount}{\period})$ and
  \[
    [\cisChainTrans]^{\leq\infty}(f(\ocConfig))(f(\ocConfig'))
    = \cisChainTrans((\ocState, \period))((\ocStateB, \period-1), 0)
    = \compressChainTrans((\ocState, \period))((\ocStateB, \period-1))
    = \compressChainTrans(\ocConfig)(\ocConfig').
  \]
  We have shown that $f$ preserves the transitions from $\ocConfig$ whenever $\ocCount$ is a multiple of $\period$.

  We now assume that $\ocCount$ is not a multiple of $\period$, and thus that $f(\ocConfig) = ((\ocState, \ocCount\bmod\period), \lfloor\frac{\ocCount}{\period}\rfloor+1)$.
  Let $\ocStateB\in\ocStateSpace$, $\ocCountB$ be a successor counter value of $\ocCount$ and $\ocConfig'=(\ocStateB, \ocCountB)$.
  Let $\interval = \integerInterval{\intLB, \intUB}\in\intPart$ such that $\ocCount\in\interval$.
  It follows that $\ocCountB\in\integerInterval{\intLB-1,\intUB+1}$.
  Since multiples of $\period$ are upper bounds of intervals, this implies that $\ocCountB\in\integerInterval{\lfloor\frac{\ocCount}{\period}\rfloor\cdot\period, (\lfloor\frac{\ocCount}{\period}\rfloor +1)\cdot\period+1}$.
  In light of this, we distinguish four cases: $\ocCountB= \lfloor\frac{\ocCount}{\period}\rfloor\cdot\period$, $\ocCountB=(\lfloor\frac{\ocCount}{\period}\rfloor+1)\cdot\period$, $\ocCountB=(\lfloor\frac{\ocCount}{\period}\rfloor+1)\cdot\period+1$ and finally $\ocCountB$ is in none of the previous cases.

  First, we assume that $\ocCountB=\lfloor\frac{\ocCount}{\period}\rfloor\cdot\period$, which implies that $f(\ocConfig') = ((\ocStateB, \period), \lfloor\frac{\ocCount}{\period}\rfloor)$.
  We have
  \[
    [\cisChainTrans]^{\leq\infty}(f(\ocConfig))(f(\ocConfig'))
    = \cisChainTrans((\ocState, \ocCount\bmod\period))((\ocStateB, \period), -1)
    = \compressChainTrans((\ocState, \ocCount\bmod\period))((\ocStateB, 0))
    = \compressChainTrans(\ocConfig)(\ocConfig').
  \]
  Second, we assume that $\ocCountB=(\lfloor\frac{\ocCount}{\period}\rfloor+1)\cdot\period$.
  This implies that $f(\ocConfig') = ((\ocStateB, \period), \lfloor\frac{\ocCount}{\period}\rfloor+1)$.
  It holds that
  \[
    [\cisChainTrans]^{\leq\infty}(f(\ocConfig))(f(\ocConfig'))
    = \cisChainTrans((\ocState, \ocCount\bmod\period))((\ocStateB, \period), 0)
    = \compressChainTrans((\ocState, \ocCount\bmod\period))((\ocStateB, \period))
    = \compressChainTrans(\ocConfig)(\ocConfig').
  \]
  Third, we assume that $\ocCountB=(\lfloor\frac{\ocCount}{\period}\rfloor+1)\cdot\period+1$.
  This implies that $f(\ocConfig') = ((\ocStateB, 1), \lfloor\frac{\ocCount}{\period}\rfloor+2)$.
  It follows that
  \[
    [\cisChainTrans]^{\leq\infty}(f(\ocConfig))(f(\ocConfig'))
    = \cisChainTrans((\ocState, \ocCount\bmod\period))((\ocStateB, 1), 1)
    = \compressChainTrans((\ocState, \ocCount\bmod\period))((\ocStateB, 1))
    = \compressChainTrans(\ocConfig)(\ocConfig').
  \]
  Finally, we assume that none of the previous cases holds.
  We conclude that $f(\ocConfig') = ((\ocStateB, \ocCountB\bmod\period), \lfloor\frac{\ocCount}{\period}\rfloor+1)$.
  We obtain
  \begin{align*}
    [\cisChainTrans]^{\leq\infty}(f(\ocConfig))(f(\ocConfig'))
    & = \cisChainTrans((\ocState, \ocCount\bmod\period))((\ocStateB, \ocCountB\bmod\period), 0) \\
    & = \compressChainTrans((\ocState, \ocCount\bmod\period))((\ocStateB, \ocCountB\bmod\period)) \\
    & = \compressChainTrans(\ocConfig)(\ocConfig').
  \end{align*}
  This ends the proof that $f$ preserves transitions.

  We lift $f$ to histories by letting, for all $\hist=\ocConfig_1\ldots\ocConfig_\indexLast\in\histSet{\compressChain}$ in which $\bot$ does not occur, $f(\hist) = f(\ocConfig_1)\ldots f(\ocConfig_\indexLast)$ and we obtain, since $f$ preserves transitions, $\probaMCverb{\compressChain}{\first{\hist}}(\cyl{\hist})=\probaMCverb{\ocChainFin{\cisChain}{\infty}}{f(\first{\hist})}(\cyl{f(\hist)})$.
  The claim of the theorem follows by writing the objectives as disjoint unions of cylinders and using the fact that $f$ is injective.
\end{proof}

\section{Verification}\label{section:verification}

We present algorithms for the interval strategy verification problem based on the compressed Markov chains of Section~\ref{section:abstraction}.
In Section~\ref{section:verification:algorithms:bounded}, we present a polynomial time algorithm in the BSS model of computation for the verification of OEISs in bounded OC-MDPs.
Sections~\ref{section:verification:algorithms:oeis} and~\ref{section:verification:algorithms:cis} present a reduction from the verification problem for OEISs and CISs respectively to checking the validity of a universal formula in the theory of the reals.
Throughout this section, we invoke the $\mathsf{Refine}$ and $\mathsf{Isolate}$ operators over interval partitions from Section~\ref{section:compression:refinement}.

We consider the following inputs: an OC-MDP $\ocmdp=\ocTuple$, a counter upper bound $\counterUB\in\INposBar$, an OEIS or CIS $\strat$ of $\ocmdpFin{\ocmdp}{\counterUB}$, an initial configuration $\ocConfig_\init=(\ocState_\init, \ocCount_\init)\in\ocStateSpace\times\integerInterval{\counterUB}$, a set of targets $\target\subseteq\ocStateSpace$ and a threshold $\thresProba\in\ccInt{0}{1}\cap\IQ$.

We describe the algorithms in a unified fashion for both the selective termination objective $\selectiveTermination{\target}$ and the state-reachability objective $\reach{\target}$.
We let $\objective\in\{\selectiveTermination{\target}, \reach{\target}\}$ denote the objective.
The major difference between the algorithms for selective termination and state-reachability is with respect to the studied OC-MDP: analysing the state-reachability probabilities requires a (polynomial-time) modification of $\ocmdp$ beforehand (see Theorem~\ref{theorem:ocmdp:probability:reach}).
We assume that this modification has been applied if $\objective=\reach{\target}$.
To further unify notation, we let $\target_\objective = \target\times\{0\}$ if $\objective=\selectiveTermination{\target}$ or $\counterUB=\infty$ and $\target_\objective=\target\times\{0, \counterUB\}$ otherwise.
This choice is motivated by the fact that, for all partitions $\intPart$ of $\integerInterval{1, \counterUB-1}$ for which $\compressChain=(\compressChainStateSpace, \compressChainTrans)$ is well-defined and $\ocConfig_\init\in\compressChainStateSpace$, Theorems~\ref{theorem:ocmdp:probability matching} and~\ref{theorem:ocmdp:probability:reach} ensure that $\probaGverb{\ocmdpFin{\ocmdp}{\counterUB}}{\ocConfig_\init}{\strat}(\objective)=\probaMCverb{\compressChain}{\ocConfig_\init}(\reach{\target_\objective})$.

\subsection{Verification in bounded one-counter Markov decision processes}\label{section:verification:algorithms:bounded}
We assume that $\counterUB\in\INpos$.
We provide a $\ptime^{\posSLP}$ upper bound on the complexity of the OEIS verification problem in bounded OC-MDPs.
We assume that $\counterUB\in\INpos$.
Let $\intPart'$ be the partition of $\integerInterval{1, \counterUB-1}$ given by the description of $\strat$.
We let $\intPart = \mathsf{Refine}(\mathsf{Isolate}(\intPart', \ocCount_\init))$.
It follows that $\strat$ is based on $\intPart$ and that $\ocConfig_\init\in\compressChainStateSpace$ (because $\ocCount_\init$ is a bound of an interval in $\intPart$).

To obtain a $\ptime^{\posSLP}$ complexity upper bound, we need only show that we can decide whether $\probaGverb{\ocmdpFin{\ocmdp}{\counterUB}}{\ocConfig_\init}{\strat}(\objective)\geq\thresProba$ in polynomial time in the BSS model~\cite{DBLP:journals/siamcomp/AllenderBKM09}.
In this model of computation, we can explicitly compute the transition probabilities of $\compressChain$ in polynomial time (by Theorem~\ref{theorem:equations:transitions:unique}) and use them to compute the probability of reaching $\target_\objective$ from $\ocConfig_\init$ in $\compressChain$.
This reachability probability is exactly $\probaGverb{\ocmdpFin{\ocmdp}{\counterUB}}{\ocConfig_\init}{\strat}(\objective)$ by Theorems~\ref{theorem:ocmdp:probability matching} and~\ref{theorem:ocmdp:probability:reach}.
We conclude by comparing it to $\thresProba$.
We obtain the following result.

\begin{theorem}\label{theorem:verification:bounded}
  The OEIS verification problem for state-reachability and selective termination in bounded OC-MDPs is in $\ptime^{\posSLP}$.
\end{theorem}
\begin{proof}
    In this proof, we reason in the BSS model of computation.
  Our goal is to clarify the algorithm outlined above and prove that it runs in polynomial time.
  We let $\intPart=\mathsf{Refine}(\mathsf{Isolate}(\intPart', \ocCount_\init))$ where $\intPart'$ is the interval partition of $\integerInterval{1, \counterUB-1}$ in the representation of $\strat$.
  Lemma~\ref{lemma:ocmpd:interval size} guarantees that $\intPart$ can be computed in polynomial time and has a polynomial-size representation with respect to $\intPart'$, and that $\compressChain$ is well-defined (cf.~Assumption~\ref{assumption:interval size}).
  It follows from Theorem~\ref{theorem:equations:transitions:unique} that the transition probabilities of $\compressChain$ can be computed in polynomial time.

  We have $\probaGverb{\ocmdpFin{\ocmdp}{\counterUB}}{\ocConfig_\init}{\strat}(\objective)=\probaMCverb{\compressChain}{\ocConfig_\init}(\reach{\target_\objective})$ by Theorems~\ref{theorem:ocmdp:probability matching} and~\ref{theorem:ocmdp:probability:reach}.
  It follows that $\probaGverb{\ocmdpFin{\ocmdp}{\counterUB}}{\ocConfig_\init}{\strat}(\objective)$ can be computed in polynomial time by solving a linear system for Markov chain reachability probabilities (with $|\compressChainStateSpace|\leq 2\cdot|\intPart|\cdot |\ocStateSpace|\cdot \log_2(\counterUB)$ variables) and then can be compared to $\thresProba$ in constant time.
  We conclude that the OEIS verification problem for $\objective$ can be solved in polynomial time in the BSS model and thus lies in $\ptime^{\posSLP}$~\cite{DBLP:journals/siamcomp/AllenderBKM09}.
\end{proof}

\subsection{Verifying open-ended interval strategies}\label{section:verification:algorithms:oeis}

We describe a $\coetr$ algorithm for the OEIS verification problem.
Recall that $\coetr$ is the class of decision problems that can be reduced (in polynomial time) to checking whether a universal sentence holds in the theory of the reals and that $\coetr$ is included in $\pspace$~\cite{DBLP:conf/stoc/Canny88}.
The algorithm of Section~\ref{section:verification:algorithms:bounded} provides a finer bound when dealing with bounded OC-MDPs.

We construct logic formulae in the signature of ordered fields to decide the verification problem.
These formulae are also relevant for the realisability problems considered in Section~\ref{section:realisability}.
For this reason, we provide a formula that (essentially) depends only on $\ocmdp$ and a finite partition $\intPart$ of $\integerInterval{1, \counterUB}$ into intervals with respect to which compressed Markov chains are well-defined.
This avoids depending only on $\strat$ which could prove problematic in a realisability setting.
We fix a finite interval partition $\intPart$ of $\integerInterval{1, \counterUB-1}$ satisfying Assumption~\ref{assumption:interval size}, i.e., such that, for all $\interval\in\intPart$, $\log_2(|\interval|+1)\in\IN$.
We build a formula with respect to $\ocmdp$ and $\intPart$ and show that we can answer the verification problem via this formula for all OEISs based on $\intPart$ from any initial configuration in $\compressChainStateSpace$.
We postpone the definition of a relevant partition for $\strat$ and $\ocConfig_\init$ to the end of the section.

Our formula uses three sets of variables.
First, for all $\ocState\in\ocStateSpace$, $\ocAction\in\ocActionSpace(\ocState)$ and $\interval\in\intPart$, we introduce a variable $\varStrat_{\ocState, \ocAction}^\interval$ to represent $\stratB(\ocState,\min\interval)(\ocAction)$ for any OEIS $\stratB$ based on $\intPart$.
For all $\interval\in\intPart$, we let $\varStratI = (\varStrat^\interval_{\ocState, \ocAction})_{\ocState\in\ocStateSpace, \ocAction\in\ocActionSpace(\ocState)}$ and let $\varStratTuple=(\varStratI)_{\interval\in\intPart}$.
We let $\stratB_{\varStratTuple}$ be the formal (parametric) OEIS based on $\intPart$ defined by $\stratB_{\varStratTuple}(\ocState, \min\interval)(\ocAction) = \varStrat^\interval_{\ocState, \ocAction}$ for all $\ocState\in\ocStateSpace$, $\ocAction\in\ocActionSpace(\ocState)$ and $\interval\in\intPart$.
The notation $\stratB_{\varStratTuple}$ allows us to refer to the compressed Markov chain $\compressChainSymbolicVerbose$ parameterised by $\varStratTuple$ in the following.
To lighten notation, we write $\compressChainSymbolic$ instead of $\compressChainSymbolicVerbose$ and $\compressChainTransSymbolic$ instead of $\compressChainTransSymbolicVerbose$.

The second set of variables comes from Theorems~\ref{theorem:equations:termination} and~\ref{theorem:equations:transitions} for each interval of $\intPart$ and are used to represent (and to characterise) the transition probabilities of $\compressChainSymbolic$ from configurations in $\compressChainStateSpace\setminus\compressChainStateSpaceStar$.
We let $\varTransTuple$ denote the vector of all of these variables.
For all configurations $\ocConfig=(\ocState, \ocCount)\in\compressChainStateSpace\setminus\compressChainStateSpaceStar$ and $\ocConfig'=(\ocStateB, \ocCountB)\in\compressChainStateSpace\setminus\{\bot\}$ such that $\ocCountB$ is a successor of $\ocCount$, we let $\varTrans_{\ocConfig, \ocConfig'}$ denote the variable corresponding to $\compressChainTrans(\ocConfig)(\ocConfig')$.
We note that some variables represent outgoing probabilities from two configurations (see Lemma~\ref{lemma:abstraction:transitions:symmetry}).

The last set of variables represents the probability of the counterpart of $\objective$ in $\compressChainSymbolic$ from each configuration.
For all $\ocConfig\in\compressChainStateSpace\setminus\compressChainStateSpaceStar$, we introduce a variable $\varObj_{\ocConfig}$ where $\varObj_{\ocConfig}$ represents $\probaMCverb{\compressChainSymbolic}{\ocConfig}(\reach{\target_\objective})$.
We let $\varObjTuple$ denote the vector of these variables.

We now construct, for all $\ocConfig\in\compressChainStateSpace\setminus\compressChainStateSpaceStar$, a quantifier-free formula such that when substituting $\varStratTuple$ by a vector $\solStratTuple$ that yields a well-defined strategy $\stratB_{\solStratTuple}$ and quantifying the other variables universally, the resulting sentence holds if and only if $\probaGverb{\ocmdpFin{\ocmdp}{\counterUB}}{\ocConfig}{\stratB_{\solStratTuple}}(\objective)\geq\thresProba$.
We rely on universal quantification because we do not have a unique characterisation of the transition probabilities of $\compressChainSymbolic$ when $\counterUB=\infty$.
We construct a quantifier-free conjunction (parameterised by the choices of the strategy) that only holds for (some) over-estimations of $\probaGverb{\ocmdpFin{\ocmdp}{\counterUB}}{\ocConfig}{\stratB_{\solStrat}}(\objective)$.
This allows us to check that $\probaGverb{\ocmdpFin{\ocmdp}{\counterUB}}{\ocConfig}{\stratB_{\solStrat}}(\objective)$ exceeds $\thresProba$ by checking that all of its over-estimations do.

Our formula has two major sub-formulae. First, we define a formula depending on $\varStratTuple$ such that the least vector satisfying it includes the transition probabilities of $\compressChainSymbolic$ from configurations to other configurations (i.e., not to $\bot$).
For each $\interval\in\intPart$, we define $\formulaTransI(\varTransTuple,\varStratI)$ as the conjunction of all the equations in the system characterising the transition probabilities from $\compressChainStateSpace\cap(\ocStateSpace\times\interval)$ in $\compressChainSymbolic$ given by Theorems~\ref{theorem:equations:termination} and~\ref{theorem:equations:transitions} (the invoked theorem depends on whether $\interval$ is finite or not).
We define
\begin{equation}\label{equation:verification:formula:transitions}
  \formulaTrans(\varTransTuple, \varStratTuple) =
  \bigwedge_{\varTrans\in\varTransTuple}\varTrans\geq 0\land
  \bigwedge_{\interval\in\intPart}
    \formulaTransI(\varTransTuple, \varStratI)\land
    \bigwedge_{\ocConfig\in\compressChainStateSpace\setminus\compressChainStateSpaceStar}
    \sum_{\varTrans_{\ocConfig, \ocConfig'}\in\varTrans_{\ocConfig, \cdot}}\varTrans_{\ocConfig, \ocConfig'}\leq 1,
  \end{equation}
  where for all $\ocConfig\in\compressChainStateSpace\setminus\compressChainStateSpaceStar$, $\varTrans_{\ocConfig, \cdot}$ denotes the set of well-defined  variables of the form $\varTrans_{\ocConfig, \ocConfig'}$ ($\ocConfig'\in\compressChainStateSpace$).
The first conjunction ensures that any vector satisfying $\formulaTrans$ is non-negative while the rightmost conjunction ensures that for all $\ocConfig\in\compressChainStateSpace\setminus\compressChainStateSpaceStar$, $\ocConfig'\mapsto\varTrans_{\ocConfig, \ocConfig'}$ is a sub-probability distribution.
It follows that, for all configurations $\ocConfig\in\compressChainStateSpace\setminus\compressChainStateSpaceStar$ and all vectors $\solTransTuple$ and $\solStratTuple$ such that $\formulaTrans(\solTransTuple, \solStratTuple)$ holds, we can define a distribution $\mdpTrans_{\bar{\solTrans}}(\ocConfig)\in\dist{\compressChainStateSpace}$ such that for $\ocConfig'\in\compressChainStateSpace\setminus\{\bot\}$, if $\varTrans_{\ocConfig, \ocConfig'}$ is a well-defined variable then  $\mdpTrans_{\bar{\solTrans}}(\ocConfig)(\ocConfig') = \solTrans_{\ocConfig, \ocConfig'}$ and, otherwise, $\mdpTrans_{\bar{\solTrans}}(\ocConfig)(\ocConfig') = 0$.
We use these distributions in our correctness proof: they allow us to reason on Markov chains over $\compressChainStateSpace$.

The second block of the formula describes the probability of reaching $\target_\objective$ in $\compressChainSymbolic$.
We consider the following formula, derived from the systems for Markov chain reachability probabilities,
\begin{equation}\label{equation:verification:formula:objective}
  \formulaObj(\varTransTuple, \varObjTuple) =
  \bigwedge_{\ocConfig\in\compressChainStateSpace\setminus\compressChainStateSpaceStar}
  \left(
    \varObj_{\ocConfig}\geq 0\land
    \varObj_{\ocConfig} =
    \sum_{\substack{
        \varTrans_{\ocConfig, \ocConfig'}\in\varTrans_{\ocConfig, \cdot}\\
        \ocConfig'\in\compressChainStateSpace\setminus\compressChainStateSpaceStar}
    }
    \varTrans_{\ocConfig, \ocConfig'}\varObj_{\ocConfig'} +
    \sum_{\substack{
        \varTrans_{\ocConfig, \ocConfig'}\in\varTrans_{\ocConfig, \cdot}\\
        \ocConfig'\in\target_\objective}}
    \varTrans_{\ocConfig, \ocConfig'}
  \right).
\end{equation}

We now state that for all well-defined instances $\stratB_{\solStratTuple}$ of $\stratB_{\varStratTuple}$, the conjunction $\formulaTrans(\varTransTuple,\solStratTuple)\land\formulaObj(\varTransTuple, \varObjTuple)$ only holds for over-estimations of the values represented by the variables.
This mainly follows from the construction of the formulae (in particular, by Theorems~\ref{theorem:equations:termination} and~\ref{theorem:equations:transitions}).
\begin{lemma}\label{lemma:verification:oeis:least assignment}
  Let $\solStratTuple$ be a vector such that $\stratB_{\solStratTuple}$ is a well-defined OEIS of $\ocmdpFin{\ocmdp}{\counterUB}$ based on $\intPart$.
  Let $\solTransTuple, \solObjTuple$ be vectors such that $\IR\models\formulaTrans(\solTransTuple,\solStratTuple)\land\formulaObj(\solTransTuple, \solObjTuple)$.
  Then, for all $\ocConfig\in\compressChainStateSpace\setminus\compressChainStateSpaceStar$, we have $\solObj_{\ocConfig}\geq\probaMCverb{\compressChainSymSol}{\ocConfig}(\reach{\target_\objective})$, and, for all $\ocConfig'\in\compressChainStateSpace\setminus\{\bot\}$ such that $\varTrans_{\ocConfig,\ocConfig'}$ is a well-defined variable, $\solTrans_{\ocConfig, \ocConfig'}\geq\compressChainTransSymSol(\ocConfig)(\ocConfig')$.
\end{lemma}
\begin{proof}
  For all $\ocConfig\in\compressChainStateSpace\setminus\compressChainStateSpaceStar$ and all $\ocConfig'\in\compressChainStateSpace\setminus\{\bot\}$ such that $\varTrans_{\ocConfig, \ocConfig'}$ is defined, we have $\solTrans_{\ocConfig, \ocConfig'}\geq\compressChainTransSymSol(\ocConfig)(\ocConfig')$ by construction of $\formulaTrans$ through Theorems~\ref{theorem:equations:termination} and~\ref{theorem:equations:transitions}.
  It remains to show that $\solObj_{\ocConfig}\geq\probaMCverb{\compressChainSymSol}{\ocConfig}(\reach{\target_\objective})$ for all $\ocConfig\in\compressChainStateSpace\setminus\compressChainStateSpaceStar$.

  Our argument relies on the finite Markov chain $\mchain_{\solTransTuple} = (\compressChainStateSpace, \mdpTrans_{\solTransTuple})$ where for all $\ocConfig\in\compressChainStateSpaceStar$, we have $\mdpTrans_{\solTransTuple}(\ocConfig)(\ocConfig)=1$ and for all $\ocConfig\in\compressChainStateSpace\setminus\compressChainStateSpaceStar$ and all $\ocConfig'\in\compressChainStateSpace\setminus\{\bot\}$, we have $\mdpTrans_{\solTransTuple}(\ocConfig)(\ocConfig')=\solTrans_{\ocConfig, \ocConfig'}$ whenever $\varTrans_{\ocConfig, \ocConfig'}$ is defined and direct the probability mass that is not assigned to a successor of $\ocConfig$ in the previous way to $\bot$.

  The least vector satisfying $\formulaObj(\solTransTuple, \varObjTuple)$ is $(\probaMCverb{\mchain_{\solTransTuple}}{\ocConfig}(\reach{\target_\objective}))_{\ocConfig\in\compressChainStateSpace\setminus\compressChainStateSpaceStar}$.
  Therefore, it suffices to show that for all $\ocConfig\in\compressChainStateSpace\setminus\compressChainStateSpaceStar$, we have $\probaMCverb{\mchain_{\solTransTuple}}{\ocConfig}(\reach{\target_\objective})\geq\probaMCverb{\compressChainSymSol}{\ocConfig}(\reach{\target_\objective})$ to end the proof.
  It suffices to establish that for all histories $\mcHist\in\histSet{\compressChainSymSol}$ with $\last{\mcHist}\in\target_\objective$ and no prior configuration in $\target_\objective$, we have $\probaMCverb{\compressChainSymSol}{\first{\mcHist}}(\cyl{\hist})\leq\probaMCverb{\mchain_{\solTransTuple}}{\first{\mcHist}}(\cyl{\hist})$.
  Let $\mcHist\in\histSet{\compressChainSymSol}$ be such a history.
  We assume that $\first{\mcHist}\notin\target_\objective$, as otherwise the result is trivial.
  Since $\bot$ is absorbing (in both Markov chains), it follows that all states along $\mcHist$ are configurations in $\compressChainStateSpace$.
  The desired inequality follows from $\ocTrans_{\solTransTuple}(\ocConfig)(\ocConfig')\geq\compressChainTrans(\ocConfig)(\ocConfig')$ holding for all $\ocConfig, \ocConfig'\in\compressChainStateSpace\setminus\{\bot\}$. \end{proof}

The following theorem provides the formula we use to solve the OEIS verification problem based on the intuition given above.
Its correctness follows from Lemma~\ref{lemma:verification:oeis:least assignment}.
\begin{theorem}\label{theorem:verification:oeis:formula}
  Let $\solStratTuple$ be a vector such that $\stratB_{\solStratTuple}$ is a well-defined OEIS.
  For all $\ocConfig\in\compressChainStateSpace\setminus\compressChainStateSpaceStar$, we have $\probaGverb{\ocmdpFin{\ocmdp}{\counterUB}}{\ocConfig}{\stratB_{\solStratTuple}}(\objective) \geq\thresProba$ if and only if $\IR\models\forall\varTransTuple\,\forall\varObjTuple ((\formulaTrans(\varTransTuple,\solStratTuple)\land\formulaObj(\varTransTuple, \varObjTuple))\implies \varObj_\ocConfig\geq\thresProba)$.
\end{theorem}
\begin{proof}
  Let $\ocConfig\in\compressChainStateSpace\setminus\compressChainStateSpaceStar$.
  By Theorems~\ref{theorem:ocmdp:probability matching} and~\ref{theorem:ocmdp:probability:reach}, we have $\probaGverb{\ocmdpFin{\ocmdp}{\counterUB}}{\ocConfig}{\stratB_{\solStratTuple}}(\objective)= \probaMCverb{\compressChainSymSol}{\ocConfig}(\reach{\target_\objective})$.
  First, we assume that $\probaMCverb{\compressChainSymSol}{\ocConfig}(\reach{\target_\objective})\geq\thresProba$.
  Let $\solTransTuple$ and $\solObjTuple$ such that $\IR\models\formulaTrans(\solTransTuple,\solStratTuple)\land\formulaObj(\solTransTuple, \solObjTuple)$.
  By Lemma~\ref{lemma:verification:oeis:least assignment}, we obtain $\solObj_\ocConfig\geq \probaMCverb{\compressChainSymSol}{\ocConfig}(\reach{\target_\objective})\geq\thresProba$.
  This shows the first implication.

  Conversely, assume that $\IR\models\forall\varTransTuple\,\forall\varObjTuple ((\formulaTrans(\varTransTuple,\solStratTuple)\land\formulaObj(\varTransTuple, \varObjTuple))\implies \varObj_\ocConfig\geq\thresProba)$.
  Let $\solTransTuple$ be the least non-negative satisfying assignment of $\varTransTuple$ in $\formulaTrans(\varTransTuple,\solStratTuple)$.
  The existence of $\solTransTuple$ is guaranteed by Theorems~\ref{theorem:equations:termination} and~\ref{theorem:equations:transitions}, which also imply that for all variables $\varTrans_{\ocConfig', \ocConfig''}$, we have $\solTrans_{\ocConfig', \ocConfig''} = \compressChainTrans(\ocConfig')(\ocConfig'')$.
  We then let $\solObjTuple$ be the least satisfying assignment of $\varObjTuple$ in the formula with parameters $\formulaObj(\solTransTuple,\varObjTuple)$.
  By construction of $\formulaObj$, we conclude that $\solObjTuple$ exists and that $\probaMCverb{\compressChainSymSol}{\ocConfig}(\reach{\target_\objective})=\solObj_\ocConfig\geq\thresProba$.
\end{proof}

We now analyse the size of the formula of Theorem~\ref{theorem:verification:oeis:formula}.
We show that this formula is of size polynomial in the encoding of $\ocmdp$ and the natural representation of $\intPart$, i.e., as a finite set of intervals whose bounds are described in binary.
We use this to show that we can build a formula to solve the verification problem in polynomial time.
This analysis is also relevant to obtain complexity bounds for realisability.
\begin{lemma}\label{lemma:verification:oeis:formula:size}
  The formula $(\formulaTrans(\varTransTuple,\varStratTuple)\land\formulaObj(\varTransTuple, \varObjTuple))\implies \varObj_\ocConfig\geq\thresProba$ has a number of variables and atomic sub-formulae polynomial in $|\ocStateSpace|$, $|\ocActionSpace|$, $|\intPart|$ and the binary encoding of the largest integer bound in $\intPart$.
\end{lemma}
\begin{proof}
    Let $\powerMax=\max_{\interval\in\intPart, |\interval|<\infty}\log_2(|\interval|+1)$ if there is a bounded interval in $\intPart$ and, otherwise, let $\powerMax=1$.
  We note that $\powerMax\leq\log_2(\intUB + 1)$ where $\intUB$ is the largest interval bound of $\intPart$.
  
  First, we have, by definition of $\varStratTuple$ and $\varObjTuple$, $|\varStratTuple|\leq|\intPart|\cdot|\ocStateSpace|\cdot|\ocActionSpace|$ and $|\varObjTuple|\leq |\compressChainStateSpace\setminus\compressChainStateSpaceStar|$.
  By definition of $\compressChainStateSpace$, we have $|\compressChainStateSpace\setminus\compressChainStateSpaceStar| \leq 2 \cdot\powerMax\cdot|\intPart|\cdot|\ocStateSpace|$.
  Second, Lemmas~\ref{lemma:equations:termination:size} and~\ref{lemma:equations:bounded:size} imply that $|\varTransTuple|$ and the number and length of the atomic sub-formulae of $\formulaTrans(\varTransTuple,\varStratTuple)$ derived from Theorems~\ref{theorem:equations:termination} and~\ref{theorem:equations:transitions} are polynomial in $|\ocStateSpace|$, $|\ocActionSpace|$, $|\intPart|$ and $\powerMax$.
  It follows from the above that the number and length of the atomic sub-formulae of $(\formulaTrans(\varTransTuple,\varStratTuple)\land\formulaObj(\varTransTuple, \varObjTuple))\implies \varObj_\ocConfig\geq\thresProba$ is polynomial in $|\ocStateSpace|$, $|\ocActionSpace|$, $|\intPart|$ and $\powerMax$.
\end{proof}

We now assume that $\strat$ is an OEIS and define the interval partition used to construct a verification formula.
Let $\intPart'$ be the interval partition of $\integerInterval{1, \counterUB-1}$ given in the representation of $\strat$.
We let $\intPart = \mathsf{Refine}(\mathsf{Isolate}(\intPart', \ocCount_\init))$.
The partition $\intPart$ satisfies Assumption~\ref{assumption:interval size} and we have $\ocConfig_\init\in\compressChainStateSpace$.
Let $\vectStratTuple$ denote the valuation of $\varStratTuple$ defined by $\varStrat^\interval_{\ocState, \ocAction}=\strat(\ocState, \min\interval)(\ocAction)$ for all $\ocState\in\ocStateSpace$, $\ocAction\in\ocActionSpace(\ocState)$ and $\interval\in\intPart$.
To decide the verification problem, we check whether the formula of Theorem~\ref{theorem:verification:oeis:formula} for $\ocConfig_\init$ holds for this valuation $\vectStratTuple$ of $\varStratTuple$.
We obtain the following complexity result.
\begin{theorem}\label{theorem:verification:oeis:complexity}
  The OEIS verification problem for selective termination and state-reachability objectives is in $\coetr$.
\end{theorem}
\begin{proof}
  We prove that the formula used to answer the verification problem can be constructed in polynomial time.
  The structure of the formula is fixed.
  Therefore, it can be constructed in time polynomial in its size.
  It follows from Lemma~\ref{lemma:verification:oeis:formula:size} that we can construct our formula in polynomial time if $\intPart$ admits a representation of size polynomial in the number of inputs to the verification problem.

    By definition of the $\mathsf{Refine}$ and $\mathsf{Isolate}$ operators, all interval bounds of $\intPart$ are either dominated by a bound in the representation of $\strat$ or by $\ocCount_\init+1$.
  Therefore, all bounds admit a polynomial-size representation.
  Furthermore, Lemma~\ref{lemma:ocmpd:interval size} guarantees that when applying the refinement procedure to obtain $\intPart$, we obtain a partition of size polynomial in the size of the inputs to the verification problem.
\end{proof}

\subsection{Verifying cyclic interval strategies}\label{section:verification:algorithms:cis}

We provide a $\coetr$ algorithm for the CIS verification problem that follows the same ideas as in Section~\ref{section:verification:algorithms:oeis}.
We assume throughout this section that $\counterUB=\infty$.
To analyse CISs, we use the compression approach twice.
First, we use it to define the one-counter Markov chain that induces the infinite compressed Markov chain with respect to our CIS for a well-chosen periodic partition of $\INpos$ (cf.~Theorem~\ref{theorem:cis:compression:ocmc}).
Second, we use the compression approach to reason on the obtained one-counter Markov chain.

The formulae constructed in this section are also relevant for the realisability algorithms of Section~\ref{section:realisability}.
We design formulae that apply to all strategies based on a given periodic partition of $\INpos$.
We let $\period\in\INpos$ be a period, $\intPartB$ be an interval partition of $\integerInterval{1, \period}$ into intervals and let $\intPart$ be the periodic partition generated by $\intPartB$.
We fix a finite interval partition $\intPartC$ of $\INpos$ for the second compression.
For all intervals $\interval\in\intPartB\cup\intPartC$, we assume that $\log_2(|\interval|+1)\in\IN$ to guarantee that compressed Markov chains are well-defined with respect to these partitions.
We design formulae for all CISs based on $\intPart$ whose structure depends only on $\ocmdp$, $\intPartB$ and $\intPartC$.
We let $\bar{\target} = (\target\times\{\period\})\times\{0\}$ denote the target of interest in the compression of the one-counter Markov chain (see Theorems~\ref{theorem:ocmdp:probability matching} and~\ref{theorem:ocmdp:probability:reach}).

Our formula for the verification problem uses four sets of variables; we require a new set of variables comparatively to Section~\ref{section:verification:algorithms:oeis} for the additional compression.
First, we introduce variables for the choices of strategies.
For all $\ocState\in\ocStateSpace$, $\ocAction\in\ocActionSpace(\ocState)$ and $\interval\in\intPartB$, we introduce a variable $\varStrat_{\ocState, \ocAction}^\interval$ to represent $\strat(\ocState,\min\interval)(\ocAction)$.
For all $\interval\in\intPartB$, we let $\varStratI = (\varStrat^\interval_{\ocState, \ocAction})_{\ocState\in\ocStateSpace, \ocAction\in\ocActionSpace(\ocState)}$ and let $\varStratTuple=(\varStratI)_{\interval\in\intPartB}$.
We let $\stratB_{\varStratTuple}$ be the formal CIS of period $\period$ based on $\intPart$ defined by $\stratB_{\varStratTuple}(\ocState, \min\interval)(\ocAction) = \varStrat^\interval_{\ocState, \ocAction}$ for all $\ocState\in\ocStateSpace$, $\ocAction\in\ocActionSpace(\ocState)$ and $\interval\in\intPartB$.
To lighten notation, we write $\compressChainSymbolic$ for the compressed Markov chain $\compressChainSymbolicVerbose$ associated to $\stratB_{\varStratTuple}$ (parameterised by $\varStratTuple$) and $\cisChainSymbolic=\cisChainTupleSymbolic$ for the one-counter Markov chain $\cisChainSymbolicVerbose$ inducing $\compressChainSymbolic$ in the sense of Theorem~\ref{theorem:cis:compression:ocmc}.
We let $\cisChainStateSpaceStar = \cisChainStateSpace\setminus\{\bot\}$.
We let $\compressCisSymbolic = \compressCisTupleSymbolic$ denote the compression of $\ocChainFin{\cisChainSymbolic}{\infty}$ with respect to $\intPartC$.

We then introduce a new set of variables $\varCisTuple$ for the transitions probabilities of $\cisChainSymbolic$ between configurations in $\cisChainStateSpaceStar$; these variables come from the system of Theorem~\ref{theorem:equations:transitions}.
For any two $\ocConfig, \ocConfig'\in\cisChainStateSpaceStar$ and weight $\weightVal\in\{-1, 0, 1\}$, we let $\varCis_{\ocConfig, \ocConfig',\weightVal}$ denote the variable corresponding to $\cisChainTransSymbolic(\ocConfig)(\ocConfig', \weightVal)$ whenever this variable is well-defined.
Third, we consider a set of variables $\varTransTuple$ for the transitions probabilities of $\compressCisSymbolic$ taken from the systems of Theorems~\ref{theorem:equations:termination} and~\ref{theorem:equations:transitions}.
For all $\cisConfig, \cisConfig'\in\compressCisStateSpace$ such that $\cisConfig\in\cisChainStateSpaceStar\times\INpos$, we write $\varTrans_{\cisConfig, \cisConfig'}$ for the variable corresponding to $\compressCisTransSymbolic(\cisConfig)(\cisConfig')$ whenever this variable is defined.
Finally, we introduce a variable $\varObj_{\cisConfig}$ for all configurations $\cisConfig\in\compressCisStateSpace\cap(\cisChainStateSpaceStar\times\INpos)$ to represent the probability $\probaMCverb{\compressCisSymbolic}{\cisConfig}(\reach{\bar{\target}})$.
We let $\varObjTuple=(\varObj_{\cisConfig})_{\cisConfig\in\compressCisStateSpace\cap(\cisChainStateSpaceStar\times\INpos)}$.

We now formulate three sub-formulae of the formula used in our decision procedure.
For all $\interval\in\intPartB$, we let $\formulaCisI(\varCisTuple, \varStratI)$ be the conjunction of the equations obtained by Theorem~\ref{theorem:equations:transitions} for the outgoing transitions of $\cisChainStateSpace\cap(\ocStateSpace\times\interval)$ in $\cisChainStrat{\stratB_{\varStratTuple}}$.
Similarly to Equation~\eqref{equation:verification:formula:transitions}, we define a formula for the transitions of $\cisChainSymbolic$ by
\begin{equation}\label{equation:verification:formula:transitions:cis}
  \formulaCis(\varCisTuple, \varStratTuple) =
  \bigwedge_{\varCis\in\varCisTuple}\varCis\geq 0\land
  \bigwedge_{\interval\in\intPartB}
  \formulaCisI(\varCisTuple, \varStratI)\land
  \bigwedge_{\ocConfig\in\cisChainStateSpace\setminus\{\bot\}}
  \sum_{\varCis_{\ocConfig, \ocConfig', \weightVal}\in\varCis_{\ocConfig, \cdot, \cdot}}\varCis_{\ocConfig, \ocConfig', \weightVal}\leq 1.
\end{equation}

We then construct the counterpart $\formulaCisTrans(\varTransTuple, \varCisTuple)$ of the formula of Equation~\eqref{equation:verification:formula:transitions} for the compressed Markov chain $\compressCis$.
In this case, the sub-formulae derived from the systems of Theorem~\ref{theorem:equations:termination} and Theorem~\ref{theorem:equations:transitions} for each interval of $\intPartC$ depend on $\varCisTuple$ instead of $\varStratTuple$.
We also build a counterpart $\formulaCisObj(\varTransTuple, \varObjTuple)$ of the formula given in Equation~\eqref{equation:verification:formula:objective} for $\compressCis$ with respect to the target $\bar{\target}$.

To decide the verification problem, we rely on a formula similar to that of Theorem~\ref{theorem:verification:oeis:formula}: we check that over-estimations of the probability of interest exceed the threshold $\thresProba$.
To validate this approach, we establish a counterpart of Lemma~\ref{lemma:verification:oeis:least assignment} for the conjunction $\formulaCis(\varCisTuple, \solStratTuple)\land\formulaCisTrans(\varTransTuple, \varCisTuple)\land\formulaCisObj(\varTransTuple,\varObjTuple)$ given a vector $\solStratTuple$ such that $\stratB_{\solStratTuple}$ is a well-defined CIS based on $\intPart$.

\begin{lemma}\label{lemma:verification:cis:least assignment}
  Let $\solStratTuple$ be a vector such that $\stratB_{\solStratTuple}$ is a well-defined CIS of $\ocmdpFin{\ocmdp}{\infty}$ based on $\intPart$.
  Let $\solCisTuple$, $\solTransTuple$ and $\solObjTuple$ be vectors such that $\IR\models\formulaCis(\solCisTuple, \solStratTuple)\land\formulaCisTrans(\solTransTuple, \solCisTuple)\land\formulaCisObj(\solTransTuple,\solObjTuple)$.
  Then, it holds that
  \begin{enumerate}
  \item for all $\cisConfig\in\compressCisStateSpace\cap(\cisChainStateSpaceStar\times\INpos)$, we have $\solObj_{\cisConfig}\geq\probaMCverb{\compressCisSymbSol}{\cisConfig}(\reach{\bar{\target}})$;\label{item:verification:cis:la:1}
  \item for all $\cisConfig\in\compressCisStateSpace\cap(\cisChainStateSpaceStar\times\INpos)$ and all $\cisConfigB\in\compressCisStateSpace$ such that $\varTrans_{\cisConfig, \cisConfigB}$ is defined, we have $\solTrans_{\cisConfig, \cisConfigB}\geq\compressCisTransSymbSol(\cisConfig)(\cisConfigB)$;\label{item:verification:cis:la:2}
  \item for all $\ocConfig$, $\ocConfig'\in\cisChainStateSpace\setminus\{\bot\}$ and $\weightVal\in\{-1, 0, 1\}$ such that $\varCis_{\ocConfig, \ocConfig',\weightVal}$ is defined, we have $\solCis_{\ocConfig, \ocConfig',\weightVal}\geq \cisChainTransSymbSol(\ocConfig)(\ocConfig', \weightVal)$.\label{item:verification:cis:la:3}
  \end{enumerate}
\end{lemma}
\begin{proof}
  Item~\ref{item:verification:cis:la:3} follows directly from the construction of $\formulaCis$ based on Theorem~\ref{theorem:equations:transitions}.
  To prove Items~\ref{item:verification:cis:la:1} and~\ref{item:verification:cis:la:2}, we consider the one-counter Markov chain $\ocChain_{\solCisTuple} = (\cisChainStateSpace, \ocTrans_{\solCisTuple})$ where for all $\ocConfig$, $\ocConfig'\in\cisChainStateSpace$ and all $\weightVal\in\{-1, 0, 1\}$, $\ocTrans_{\solCisTuple}(\ocConfig)(\ocConfig', \weightVal) = \solCis_{\ocConfig, \ocConfig', \weightVal}$ whenever the variable $\varCis_{\ocConfig, \ocConfig', \weightVal}$ is defined and any probability that is not assigned in this way is attributed to $\ocTrans_{\solCisTuple}(\ocConfig)(\bot, 0)$.
  We show that in the compression $\mchain_\intPartC(\ocChain_{\solCisTuple}) = (\compressCisStateSpace, \ocTrans_{\intPartC}[\ocChain_{\solCisTuple}])$ of $\ocChain_{\solCisTuple}$ with respect to $\intPartC$, we have the two following properties: \begin{enumerate}[a.]
  \item for all $\cisConfig, \cisConfigB\in\compressCisStateSpace\cap(\cisChainStateSpaceStar\times\INpos)$, $\ocTrans_{\intPartC}[\ocChain_{\solCisTuple}](\cisConfig)(\cisConfigB)\geq\compressCisTrans(\cisConfig)(\cisConfigB)$;\label{subitem:verification:cis:la:a}
  \item for all $\cisConfig\in\compressCisStateSpace\cap(\cisChainStateSpaceStar\times\INpos)$, we have $\probaMCverb{\ocChain_{\intPartC}(\ocChain_{\solCisTuple})}{\cisConfig}(\reach{\bar{\target}})\geq\probaMCverb{\compressCisSymbSol}{\cisConfig}(\reach{\bar{\target}})$.\label{subitem:verification:cis:la:b}
  \end{enumerate}
  These two properties along with Lemma~\ref{lemma:verification:oeis:least assignment} (with respect to the parameterised formula $\formulaCisTrans(\varTransTuple, \solCisTuple)\land\formulaCisObj(\varTransTuple,\varObjTuple)$) yield Items~\ref{item:verification:cis:la:1} and~\ref{item:verification:cis:la:2}.

  We first prove Item~\ref{subitem:verification:cis:la:a}.
  Let $\cisConfig, \cisConfigB\in\compressCisStateSpace\setminus\{\bot\}$.
  Let $\ocConfig$, $\ocConfig\in\cisChainStateSpace$ and $\ocCount, \ocCount'\in\IN$ such that $\cisConfig=(\ocConfig, \ocCount)$ and $\cisConfigB=(\ocConfig', \ocCount')$.
  If $\ocCount'$ is not a successor counter value of $\ocCount$ with respect to $\intPartC$, then we have $\ocTrans_{\intPartC}[\ocChain_{\solCisTuple}](\cisConfig)(\cisConfigB) = \compressCisTransSymbSol(\cisConfig)(\cisConfigB) = 0$.
  Otherwise, $\compressCisTransSymbSol(\cisConfig)(\cisConfigB)$ is the probability of reaching $\cisConfigB$ from $\cisConfig$ in $\ocChainFin{\cisChainSymbSol}{\infty}$ without visiting another configuration with a successor value of $\ocCount$ beforehand.
  Along the relevant plays for this probability, there are no $\bot$ configurations.
  Since $\ocTrans_{\intPartC}[\ocChain_{\solCisTuple}](\cisConfig)(\cisConfigB)$ is similarly defined, the desired inequality follows from Item~\ref{item:verification:cis:la:3}.

  Item~\ref{subitem:verification:cis:la:a} implies Item~\ref{subitem:verification:cis:la:b}, as there are no $\bot$ configuration that can occur on plays ending in $\bar{\target}$ (of which all states are absorbing).
\end{proof}

We obtain an adaptation of Theorem~\ref{theorem:verification:oeis:formula} for CISs via Lemma~\ref{lemma:verification:cis:least assignment}.
We use the correspondence between configurations of $\cisChainSymbolic$ and $\compressCisSymbolic$ established in Theorem~\ref{theorem:cis:compression:ocmc} in this result.

\begin{theorem}\label{theorem:verification:cis:formula}
  Let $\solStratTuple$ be a vector such that $\stratB_{\solStratTuple}$ is a well-defined CIS of $\ocmdpFin{\ocmdp}{\infty}$ based on $\intPart$.
  For all $\cisConfig = ((\ocState, \ocCount), \ocCountB)\in\compressCisStateSpace\cap(\cisChainStateSpaceStar\times\INpos)$, we have $\probaGverb{\ocmdpFin{\ocmdp}{\infty}}{(\ocState, (\ocCountB-1)\cdot\period+\ocCount)}{\stratB_{\solStratTuple}}(\reach{\bar{\target}})\geq\thresProba$ if and only if $\IR\models\forall\varTransTuple\,\forall\varObjTuple\forall\varCisTuple ((\formulaCis(\varCisTuple, \solStratTuple)\land\formulaCisTrans(\varTransTuple,\varCisTuple)\land\formulaCisObj(\varTransTuple, \varObjTuple))\implies \varObj_{\cisConfig}\geq\thresProba)$.
\end{theorem}
\begin{proof}
  Let $\cisConfig = ((\ocState, \ocCount), \ocCountB)\in\compressCisStateSpace\cap(\cisChainStateSpaceStar\times\INpos)$.
  By Theorems~\ref{theorem:cis:compression:ocmc} and~\ref{theorem:ocmdp:probability matching}, we have $\probaGverb{\ocmdpFin{\ocmdp}{\infty}}{(\ocState, (\ocCountB-1)\period+\ocCount)}{\stratB_{\solStratTuple}}(\reach{\bar{\target}})=\probaMCverb{\ocChainFin{\cisChainSymbSol}{\infty}}{\cisConfig}(\reach{\bar{\target}})=\probaMCverb{\compressCisSymbSol}{\cisConfig}(\reach{\bar{\target}})$

  If $\probaMCverb{\ocChainFin{\cisChain}{\infty}}{\cisConfig}(\reach{\bar{\target}})\geq\thresProba$, then we have $\IR\models\forall\varTransTuple\,\forall\varObjTuple\forall\varCisTuple ((\formulaCis(\varCisTuple, \solStratTuple)\land\formulaCisTrans(\varTransTuple,\varCisTuple)\land\formulaCisObj(\varTransTuple, \varObjTuple))\implies \varObj_{\cisConfig}\geq\thresProba)$ by Lemma~\ref{lemma:verification:cis:least assignment}.

  Conversely, assume that $\IR\models\forall\varTransTuple\,\forall\varObjTuple\forall\varCisTuple((\formulaCis(\varCisTuple, \solStratTuple)\land\formulaCisTrans(\varTransTuple,\varCisTuple)\land\formulaCisObj(\varTransTuple, \varObjTuple))\implies \varObj_{\cisConfig}\geq\thresProba)$.
  Let $\solCisTuple$ be the least vector satisfying $\formulaCis(\varCisTuple, \solStratTuple)$, $\solTransTuple$ be the least vector satisfying $\formulaCisTrans(\varTransTuple, \solCisTuple)$ and $\solObjTuple$ be the least vector satisfying $\formulaCisObj(\solTransTuple, \varObjTuple)$.
  By construction of these three formulae (and Theorems~\ref{theorem:equations:termination} and~\ref{theorem:equations:transitions}), these vectors are well-defined and we obtain $\solObj_{\cisConfig} = \probaMCverb{\ocChainFin{\cisChain}{\infty}}{\cisConfig}(\reach{\bar{\target}})\geq\thresProba$.
\end{proof}

We now study the size of the formula of Theorem~\ref{theorem:verification:cis:formula} for our complexity analysis.
We obtain a conclusion similar to that provided by Lemma~\ref{lemma:verification:oeis:formula:size}.
\begin{lemma}\label{lemma:verification:cis:formula:size}
  The formula $(\formulaCis(\varCisTuple, \varStratTuple)\land\formulaCisTrans(\varTransTuple,\varCisTuple)\land\formulaCisObj(\varTransTuple, \varObjTuple))\implies \varObj_{\cisConfig}\geq\thresProba$ has a number of variables and atomic sub-formulae polynomial in $|\ocStateSpace|$, $|\ocActionSpace|$, $|\intPartB|$, $|\intPartC|$, the binary encoding of $\period$ and the binary encoding of the largest integer bound in $\intPartC$.
    The length of its atomic sub-formulae is polynomial in the same parameters.
\end{lemma}
\begin{proof}
  Lemma~\ref{lemma:verification:oeis:formula:size} implies that the number of variables and atomic formulae of $\formulaCisTrans(\varTransTuple,\varCisTuple)\land\formulaCisObj(\varTransTuple, \varObjTuple)$, well as the length of these atomic formulae, is polynomial in $|\cisChainStateSpaceStar|$, $|\intPartC|$ and the binary encoding of the largest bound in $\intPartC$ (actions are not relevant; remark that we deal with a Markov chain).
  By construction of $\cisChainStateSpaceStar$, we have $|\cisChainStateSpaceStar|\leq2\cdot\log_2(\period+1)\cdot|\intPartB|\cdot|\ocStateSpace|$ (interval bounds of $\intPartB$ are at most $\period$).
  Regarding $\formulaCis(\varCisTuple, \varStratTuple)$, it suffices to adapt the analysis performed in the proof of Lemma~\ref{lemma:verification:oeis:formula:size} for the formula $\formulaTrans$ to obtain the desired bounds.
\end{proof}

We now assume that the input strategy $\strat$ is a CIS.
We close the section by explaining how to construct $\intPartB$ and $\intPartC$ in polynomial time from the representation of $\strat$ and $\ocConfig_\init$ to prove that the verification problem is in $\coetr$.
Let $\intPartB'$ denote the partition of $\integerInterval{1, \period}$ given in the representation of $\strat$.
We let $\intPartB=\mathsf{Refine}(\mathsf{Isolate}(\intPartB, \ocCount_\init\bmod\period))$.
For the partition for the second compression, we let $\intPartC=\mathsf{Refine}(\integerInterval{1, \lfloor\frac{\ocCount_\init}{\period}\rfloor})\cup\{\integerInterval{\lfloor{\frac{\ocCount_\init}{\period}\rfloor+1, +\infty}}\}$ of $\INpos$.
We observe that the counterpart of $\ocConfig_\init$ in $\compressCisStateSpace$ (in the sense of Theorem~\ref{theorem:cis:compression:ocmc}) is guaranteed to exist: if $\ocCount_\init\bmod\period=0$, then we have $((\ocState_\init, \period), \frac{\ocCount_\init}{\period})\in\compressCisStateSpace$, and, otherwise, we have $((\ocState_\init, \ocCount_\init\bmod\period), \lfloor\frac{\ocCount_\init}{\period}\rfloor+1)\in\compressCisStateSpace$.
We let $\vectStratTuple$ denote the substitution of $\varStratTuple$ such that $\varStrat^\interval_{\ocState, \ocAction}$ is set to $\strat(\ocState, \min\interval)(\ocAction)$ for all $\ocState\in\ocStateSpace$, $\ocAction\in\ocActionSpace(\ocState)$ and $\interval\in\intPartB$.
With the partitions $\intPartB$ and $\intPartC$ given above and the formula of Theorem~\ref{theorem:verification:cis:formula} with respect to the counterpart of $\ocConfig_\init$ in $\compressCisStateSpace$ and the parameter $\vectStratTuple$, we can decide our instance of the verification problem.
We thus obtain the following complexity result.
\begin{theorem}\label{theorem:verification:cis:complexity}
  The CIS verification problem for selective termination and reachability objectives is in $\coetr$.
\end{theorem}
\begin{proof}
  We observe, in the same way as in the proof of Theorem~\ref{theorem:verification:oeis:complexity}, that Lemmas~\ref{lemma:ocmpd:interval size} and~\ref{lemma:verification:cis:formula:size} imply that the formula of Theorem~\ref{theorem:verification:cis:formula} can be constructed in polynomial time for the partitions described above.
\end{proof}

\section{Fixed-interval and parameterised realisability}\label{section:realisability}
We provide complexity upper bounds for the fixed-interval and parameterised realisability problems for interval strategies.
Our algorithms are built on the verification techniques presented in Section~\ref{section:verification}.
We first provide a technical result for analysing the complexity of the parameterised realisability problem in Section~\ref{section:realisability:technical}.
In Section~\ref{section:realisability:bounded}, we study our realisability problems in bounded OC-MDPs.
We consider OEISs in general, i.e., we provide an approach applicable for both the bounded and unbounded setting, in Section~\ref{section:realisability:oeis}.
We close the section with CISs in Section~\ref{section:realisability:cis}.

We use similar approaches for all settings.
For fixed-interval realisability for pure strategies, we non-deterministically construct strategies and verify them.
This approach is not viable for fixed-interval realisability for randomised strategies; instead, we quantify existentially over strategy variables in the logical formulae used for verification.
For the parameterised realisability problem, we build on our algorithms for the  fixed-interval case.
The main idea is use non-determinism to find an interval partition compatible with the input parameters and then use fixed-interval algorithms with this partition to answer our problem.
All complexity bounds provided in this section are in $\pspace$.

We consider the following inputs for the whole section: an OC-MDP $\ocmdp=\ocTuple$, a counter upper bound $\counterUB\in\INposBar$, an initial configuration $\ocConfig_\init=(\ocState_\init, \ocCount_\init)\in\ocStateSpace\times\integerInterval{\counterUB}$, a set of targets $\target\subseteq\ocStateSpace$, an objective $\objective\in\{\reach{\target}, \selectiveTermination{\target}\}$  and a threshold $\thresProba\in\ccInt{0}{1}\cap\IQ$.
We specify the other inputs below.
As in Section~\ref{section:verification}, we assume that we work with the modified OC-MDP of Theorem~\ref{theorem:ocmdp:probability:reach} if $\objective=\reach{\target}$.

\subsection{Parameters and compatible interval partitions}\label{section:realisability:technical}
We present algorithms for parameterised interval strategy realisability problems that rely on non-determinism to test all interval partitions compatible with the input parameters to check that there is a well-performing strategy based on one of these interval partitions.
We state a result implying that all of these interval partition admit a representation that is polynomial in the encoding of the input parameters.

\begin{lemma}\label{lemma:partition representation}
  Let $\intNum\in\INpos$, $\intSize\in\INpos$ and $\ocCount\in\INbar$.
  Let $\intPart$ be an interval partition of $\integerInterval{1, \ocCount}$ such that $|\intPart|\leq\intNum$ and for all bounded $\interval\in\intPart$, $|\interval|\leq\intSize$.
  Then $\intPart$ can be explicitly represented in space $\bigo(\intNum\cdot(\log_2(\intNum) + \log_2(\intSize)))$. \end{lemma}
\begin{proof}
  We can represent each interval $\integerInterval{\intLB, \intUB}\in\intPart$ by the pair $(\intLB, \intUB)$ (where $\intUB$ can be $+\infty$).
  We prove that each finite interval bound in these pairs is at most $\intSize\cdot\intNum$.

  First, assume that $\ocCount\in\IN$.
  In this case, the interval $\integerInterval{1, \ocCount}$ is the union of at most $\intNum$ sets of at most $\intSize$ elements (by the assumption on $\intPart$).
  It follows that $\ocCount\leq\intNum\cdot\intSize$, and thus, that all finite interval bounds of $\intPart$ are no more than $\intNum\cdot\intSize$.

  Second, assume that $\ocCount=\infty$.
  Let $\interval_\infty$ denote the unbounded interval in $\intPart$.
  It holds (by the same reasoning as above) that the bounds of all intervals in $\intPart\setminus\{\interval_\infty\}$ are no more than $(\intNum-1)\cdot\intSize$.
  This implies that $\min\interval_\infty - 1\leq (\intNum-1)\cdot\intSize$.
  We obtain that in this second case, all finite interval bounds in $\intPart$ are no more than $(\intNum-1)\cdot\intSize + 1\leq\intNum\cdot\intSize$.
  
  We conclude that, in both cases, we can represent $\intPart$ using no more than $\intNum$ pairs of numbers whose binary encoding is in $\bigo(\log_2(\intNum)+\log_2(\intSize))$.
\end{proof}

Lemma~\ref{lemma:partition representation} implies that, under the assumption that the parameter $\intNum$ for the number of intervals is given in unary, the interval partitions that are compatible with the parameters admit a polynomial-size representation with respect to the inputs to the parameterised interval strategy realisability problems.

\subsection{Realisability in bounded one-counter Markov decision processes}\label{section:realisability:bounded}
Assume that $\counterUB\in\INpos$.
We provide algorithms for the fixed-interval and parameterised OEIS realisability problems in bounded OC-MDPs.
We first discuss the variants of these problems for pure strategies, and then discuss the variants for randomised strategies below.

First, we consider the fixed-interval pure OEIS realisability problem.
Fix an input interval partition $\intPart'$ of $\integerInterval{1, \counterUB-1}$.
We obtain a straightforward non-deterministic algorithm: we guess a pure interval strategy $\strat$ based on the partition $\intPart'$ and then use our verification algorithm for OEISs in bounded OC-MDPs as a $\ptime^{\posSLP}$ sub-procedure (Theorem~\ref{theorem:verification:bounded}) to check whether $\probaGverb{\ocmdpFin{\ocmdp}{\counterUB}}{\ocConfig_\init}{\strat}(\objective)\geq\thresProba$.
This realisability algorithm runs in non-deterministic polynomial time with a $\posSLP$ oracle: we non-deterministically choose $\intNum\cdot|\ocStateSpace|$ actions, i.e., one per state-interval pair, and then run a deterministic polynomial-time algorithm with a $\posSLP$ oracle.
This yields an $\np^{\posSLP}$ upper bound for this problem.

For the parameterised pure OEIS realisability problem in bounded OC-MDPs, we proceed similarly.
Let $\intNum\in\INpos$ and $\intSize\in\INpos$ respectively denote the input parameters for the number and size of intervals.
We non-deterministically guess an interval partition $\intPart'$ of $\integerInterval{1, \counterUB-1}$ that is compatible with $\intNum$ and $\intSize$ (these partitions can be represented in polynomial space by Lemma~\ref{lemma:partition representation}), guess a pure strategy based on $\intPart'$ and verify it.
In this case, by adapting the analysis made above, we also obtain an $\np^{\posSLP}$ upper complexity bound.
We summarise the above upper bounds in the following theorem.

\begin{theorem}\label{theorem:realisability:bounded:pure}
  The fixed-interval and parameterised pure OEIS realisability problems for selective termination and state-reachability objectives in bounded OC-MDPs are in $\np^{\posSLP}$.
\end{theorem}

We now consider the fixed-interval and parameterised randomised OEIS realisability problem and describe an $\np^{\etr}$ algorithm.
We start with the fixed-interval realisability problem.
Let $\intPart' = (\interval_\intIndex)_{\intIndex\in\integerInterval{1, \intNum}}$ denote an input interval partition of $\integerInterval{1, \counterUB-1}$.
Prefacing the formula of Theorem~\ref{theorem:verification:oeis:formula} with existential quantifiers for the strategy probabilities yields a polynomial-space procedure (cf.~Section~\ref{section:realisability:oeis}).
We provide an alternative approach for bounded OC-MDPs which yields a more precise bound.

The key is to rely on a unique characterisation of the transition and reachability probabilities in the compressed Markov chain that we consider.
Theorem~\ref{theorem:equations:transitions:unique} provides the means to do this: it provides systems whose only solution contains the transition probabilities of a compressed Markov chain.
These systems also indicate which transitions have positive probability in the compressed Markov chain.
We can thus refine the reachability probability system (described in the formula in Equation~\eqref{equation:verification:formula:objective}) to have a unique solution.

To adequately refine systems using Theorem~\ref{theorem:equations:transitions:unique}, we must know the supports of the distributions assigned by the considered strategy.
For this, we use non-determinism; we guess the action supports for each state-interval pair.
We then construct an existential formula that is dependent on these supports.
This formula holds if and only if there exists a strategy witnessing a positive answer to the fixed-interval randomised OEIS realisability problem that uses these supports.

We let $\target_\objective=\target\times\{0\}$ if $\objective=\selectiveTermination{\target}$ and $\target_\objective=\target\times\{0, \counterUB\}$ if $\objective=\reach{\target}$.
Let $\intPart=\mathsf{Refine}(\mathsf{Isolate}(\intPart', \ocCount_\init))$.
We note that $\ocConfig_\init\in\compressChainStateSpace$.
For all $\intIndex\in\integerInterval{1, \intNum}$, we let $\intPart_\intIndex = \{\interval\in\intPart\mid\interval\subseteq\interval_\intIndex\}$.
We use the same variables as the formula of Theorem~\ref{theorem:verification:oeis:formula} presented in Section~\ref{section:verification:algorithms:oeis}.
More precisely, we have a variable vector $\varStratTuple$ for the probabilities assigned by strategies and a vector $\varStratI$ for all $\interval\in\intPart$ for the strategy probabilities specific to $\interval$.
We let $\stratB_{\varStratTuple}$ denote a formal strategy given by $\stratB_{\varStratTuple}(\ocState, \min\interval)(\ocAction) = \varStrat^\interval_{\ocState, \ocAction}$ for all $\ocState\in\ocStateSpace$, $\ocAction\in\ocActionSpace(\ocState)$ and $\interval\in\intPart$.
We then have $\varTransTuple$ for the transition probabilities of the compressed Markov chain $\compressChainSymbolic$ and $\varObjTuple$ for the probability of reaching $\target_\objective$ from each configuration in $\compressChainStateSpace\setminus\compressChainStateSpaceStar$.

We call functions $\suppBounded\colon\ocStateSpace\times\integerInterval{1, \intNum}\to\subsets{\ocActionSpace}\setminus\{\emptyset\}$ such that for all $\ocState\in\ocStateSpace$ and $\intIndex\in\integerInterval{1, \intNum}$, the inclusion $\suppBounded(\ocState, \intIndex)\subseteq\ocActionSpace(\ocState)$ holds \textit{support-assigning functions}.
An OEIS $\strat$ based on $(\interval_\intIndex)_{\intIndex\in\integerInterval{1, \intNum}}$ is \textit{compatible} with $\suppBounded$ if for all $\ocState\in\ocStateSpace$ and $\intIndex\in\integerInterval{1, \intNum}$,  we have $\supp{\strat(\ocState, \min\interval_\intIndex)} = \suppBounded(\ocState, \intIndex)$.

We now define the required sub-formulae used in our algorithm.
The first formula checks that the substitution of $\varStratTuple$ results in an interpretation of the symbolic strategy $\stratB_{\varStratTuple}$  that is based on $\intPart'$ and is compatible with $\suppBounded$.
For each $\intIndex\in\integerInterval{1, \intNum}$, we fix $\interval_{\intIndex}^\star\in\intPart_\intIndex$.
We define the formula $\formulaStratB(\varStratTuple)$ by 
\begin{equation}\label{equation:formula:strategy:assigned-support}
    \bigwedge_{\intIndex\in\integerInterval{1, \intNum}}
    \left(
    \bigwedge_{\ocState\in\ocStateSpace}\left(
      \bigwedge_{\ocAction\in\suppBounded(\ocState,\intIndex)}
      \varStrat_{\ocState,\ocAction}^{\interval^\star_\intIndex} > 0\land
      \bigwedge_{\ocAction\notin\suppBounded(\ocState,\intIndex)}
      \varStrat_{\ocState,\ocAction}^{\interval^\star_\intIndex} = 0\land
      \sum_{\ocAction\in\ocActionSpace(\ocState)}    \varStrat_{\ocState,\ocAction}^{\interval^\star_\intIndex} = 1
    \right)\land
    \bigwedge_{\interval\in\intPart_\intIndex}
    \varStratI = \varStratIstar
  \right).
\end{equation}

The other sub-formulae are built under the assumption that we consider an interpretation of $\stratB_{\varStratTuple}$ with the supports described by $\suppBounded$.
The second sub-formula is a parallel of the formula of Equation~\eqref{equation:verification:formula:transitions}.
For each $\interval\in\intPart$, we let $\formulaTransBI(\varTransTuple, \varStratI)$ be the conjunction of the equations in the system with a unique solution obtained from Theorem~\ref{theorem:equations:transitions:unique} for the transitions from $\compressChainStateSpace\cap(\ocStateSpace\times\interval)$ in $\compressChainSymbolic$.
We let $\formulaTransB(\varTransTuple, \varStratTuple) = \bigwedge_{\interval\in\intPart}\formulaTransBI(\varTransTuple, \varStratI)$.
From these equations, we can deduce the transition structure of $\compressChainSymbolic$ and construct a linear system with a unique solution describing the probability of reaching $\target_\objective$ in $\compressChainSymbolic$; we let $\formulaObjB(\varTransTuple, \varObjTuple)$ denote the conjunction of the equations of this system.
Using the fact that these last two formulae have unique satisfying assignments for a valuation of $\varStratTuple$ satisfying $\formulaStratB(\varStratTuple)$, we obtain the following theorem.
\begin{theorem}\label{theorem:synthesis:bounded:formula}
  Let $\suppBounded\colon\ocStateSpace\times\integerInterval{1, \intNum}\to\subsets{\ocActionSpace}\setminus\{\emptyset\}$ be a support-assigning function and $\ocConfig\in\compressChainStateSpace\setminus\compressChainStateSpaceStar$.
  There exists a strategy $\strat$ based on $\intPart'$ that is compatible with $\suppBounded$ such that $\probaGverb{\ocmdpFin{\ocmdp}{\counterUB}}{\ocConfig}{\strat}(\objective) \geq\thresProba$ if and only if $\IR\models\exists\varStratTuple\exists\varTransTuple\,\exists\varObjTuple (\formulaStratB(\varStratTuple)\land\formulaTransB(\varTransTuple,\varStratTuple)\land\formulaObjB(\varTransTuple, \varObjTuple)\land \varObj_\ocConfig\geq\thresProba)$.
\end{theorem}
\begin{proof}
  Assume that there exists a strategy $\strat$ based on $\intPart'$ compatible with $\suppBounded$ such that $\probaGverb{\ocmdpFin{\ocmdp}{\counterUB}}{\ocConfig}{\strat}(\objective) \geq\thresProba$.
  Let $\vectStratTuple$ denote the valuation of $\varStratTuple$ given by $\varStrat^\interval_{\ocState,\ocAction}=\strat(\ocState, \min\interval)(\ocAction)$ for all $\ocState\in\ocStateSpace$, $\ocAction\in\ocActionSpace(\ocState)$ and $\interval\in\intPart$.
  It is easy to see that $\formulaStratB(\vectStratTuple)$ holds by compatibility of $\strat$ with $\suppBounded$.
  By construction of $\formulaTransB(\varTransTuple, \varStratTuple)$ (via Theorem~\ref{theorem:equations:transitions:unique}), there is a unique vector $\solTransTuple$ such that $\formulaTransB(\solTransTuple, \vectStratTuple)$ holds which contains the transition probabilities of $\compressChain$.
  In turn, this implies that there is a unique vector $\solObjTuple$ such that $\formulaObjB(\solTransTuple, \solObjTuple)$ holds, and this vector is $(\probaMCverb{\compressChain}{\ocConfig'}(\reach{\target_\objective}))_{\ocConfig'\in\compressChainStateSpace\setminus\compressChainStateSpaceStar}$.
  By Theorems~\ref{theorem:ocmdp:probability matching} and~\ref{theorem:ocmdp:probability:reach}, we obtain that $\solObj_\ocConfig = \probaMCverb{\compressChain}{\ocConfig}(\reach{\target_\objective}) = \probaGverb{\ocmdpFin{\ocmdp}{\counterUB}}{\ocConfig}{\strat}(\objective) \geq\thresProba$.

  Conversely, let $\solStratTuple$, $\solTransTuple$ and $\solObjTuple$ witnessing that the existential formula above holds and define $\strat = \stratB_{\solStratTuple}$.
  The strategy $\strat$ is well-defined and compatible with $\suppBounded$ because $\formulaStratB(\solStratTuple)$ holds.
  Furthermore, by construction of the formulae $\formulaTransB$ and $\formulaObjB$, we deduce that $\solObj_\ocConfig = \probaMCverb{\compressChain}{\ocConfig}(\reach{\target_\objective})\geq\thresProba$.
  We conclude that $\probaGverb{\ocmdpFin{\ocmdp}{\counterUB}}{\ocConfig}{\strat}(\objective)=\probaMCverb{\compressChain}{\ocConfig}(\reach{\target_\objective}) \geq\thresProba$ by Theorems~\ref{theorem:ocmdp:probability matching} and~\ref{theorem:ocmdp:probability:reach}.
\end{proof}

To decide the fixed-interval randomised OEIS realisability problem, it suffices to check that there exists a support-assigning function $\suppBounded$ (using non-determinism) such that the formula of Theorem~\ref{theorem:synthesis:bounded:formula} for holds for $\ocConfig_\init$ (by construction of $\intPart$, $\ocConfig_\init\in\compressChainStateSpace$).
We thus obtain an $\np^\etr$ upper bound for this variant of the fixed-interval realisability problem.

For the parameterised realisability problem, we obtain an $\np^\etr$ upper bound by altering the fixed-interval algorithm slightly.
In this case, we use non-determinism to guess an interval partition $\intPart'$ that is compatible with the input parameters and a support-assigning function.
We then check the validity of the formula of Theorem~\ref{theorem:synthesis:bounded:formula} for the interval partition $\intPart=\mathsf{Refine}(\mathsf{Isolate}(\intPart', \ocCount_\init))$ and the initial configuration $\ocConfig_\init$.
We obtain the following result.

\begin{theorem}\label{theorem:realisability:bounded:randomised}
    The fixed-interval and parameterised randomised OEIS realisability problems for selective termination and state-reachability objectives in bounded OC-MDPs are in $\np^\etr$.
\end{theorem}
\begin{proof}
  We present a unified argument for both the fixed-interval and parameterised realisability problems.
  The only difference between our complexity analysis for these two problems is how the interval partition $\intPart'$ is obtained.
  In the fixed-interval case, the interval partition $\intPart'$ is a part of the input.
  In the parameterised case, $\intPart'$ is of size polynomial in the input parameters by Lemma~\ref{lemma:partition representation} (recall that the parameter bounding the number of intervals is assumed to be given in unary).
  
  We let $\intPart=\mathsf{Refine}(\mathsf{Isolate}(\intPart', \ocCount_\init))$ and $\suppBounded\colon\ocStateSpace\times\integerInterval{1, |\intPart'|}\to\subsets{\ocActionSpace}\setminus\{\emptyset\}$ be a support-assigning function.
  Lemma~\ref{lemma:ocmpd:interval size} guarantees that $\intPart$ has a representation of size polynomial in that of $\intPart'$ and $\ocCount_\init$.
  The support-assigning function $\suppBounded$ can be explicitly represented in space polynomial in $|\ocStateSpace|$, $|\ocActionSpace|$ and $|\intPart'|$.
  Therefore, to prove that the $\np^\etr$ upper bound holds, it remains to prove that the formula $\formulaStratB(\varStratTuple)\land\formulaTransB(\varTransTuple,\varStratTuple)\land\formulaObjB(\varTransTuple, \varObjTuple)\land \varObj_{\ocConfig_\init}\geq\thresProba$ can be constructed in deterministic polynomial time from $\intPart'$, $\intPart$, $\suppBounded$ and the other inputs to the considered realisability problem.
  
  Lemma~\ref{lemma:verification:oeis:formula:size} implies that the number of variables in this formula is polynomial in $|\ocStateSpace|$, $|\ocActionSpace|$ and $|\intPart|$, and that the formulae $\formulaTransB(\varTransTuple,\varStratTuple)$ and $\formulaObjB(\varTransTuple, \varObjTuple)$ have size polynomial in $|\ocStateSpace|$, $|\ocActionSpace|$ and $|\intPart|$.
  Indeed, for these two formulae, we observe that they can be derived from the original formulae $\formulaTrans$ and $\formulaObj$ for verification by taking their conjunction with atomic propositions requiring that some variables in $\varTransTuple$ and $\varObjTuple$ are equal to zero (these variables can be identified in polynomial time through the algorithm of Theorem~\ref{theorem:equations:transitions:unique} and with a reachability analysis of the compressed Markov chain).
  Finally, we observe that, in the formula $\formulaStratB(\varStratTuple)$, there are no more than $|\intPart'|\cdot (2\cdot|\ocStateSpace| + |\intPart|\cdot|\ocActionSpace|\cdot|\ocStateSpace|)$ atomic formulae of length in $\bigo(|\ocActionSpace|)$.
  We have thus shown that the formula $\formulaStratB(\varStratTuple)\land\formulaTransB(\varTransTuple,\varStratTuple)\land\formulaObjB(\varTransTuple, \varObjTuple)\land \varObj_{\ocConfig_\init}\geq\thresProba$ can be constructed in polynomial time.
\end{proof}
\subsection{Realisability for open-ended interval strategies}\label{section:realisability:oeis}

We now consider the fixed-interval and parameterised realisability problems for OEISs in general OC-MDPs.
For pure strategies, we can adapt the approach of Section~\ref{section:realisability:bounded}.
In the fixed-interval case, we guess a pure OEIS based on the input interval partition of the set of counter values and verify it.
In the parameterised case, we guess an interval partition compatible with the input parameters (its representation is of size polynomial in the representation of the parameters by Lemma~\ref{lemma:partition representation}) and a pure OEIS based on it, then verify it.
Theorem~\ref{theorem:verification:oeis:complexity} thus implies that the fixed-interval and parameterised realisability problems for pure OEISs can be solved by a non-deterministic polynomial-time algorithm that uses a $\coetr$ oracle (which is the same as using an $\etr$ oracle).
We obtain the following upper bound.
\begin{theorem}\label{theorem:realisability:oeis:pure}
  The fixed-interval and parameterised pure OEIS realisability problems for selective termination and state-reachability objectives are in $\np^\etr$.
\end{theorem}

We now consider the variant of our realisability problems for randomised OEISs.
First, we discuss the fixed-interval problem and let $\intPart' =(\interval_{\intIndex})_{\intIndex\in\integerInterval{1, \intNum}}$ be an input partition.
Let $\intPart=\mathsf{Refine}(\mathsf{Isolate}(\intPart', \ocCount_\init))$.
We let, for all $\intIndex\in\integerInterval{1, \intNum}$, $\intPart_\intIndex = \{\interval\in\intPart\mid\interval\subseteq\interval_\intIndex\}$.
Next, we consider the sets of variables $\varStratTuple$, $\varTransTuple$ and $\varObjTuple$ and the formulae $\formulaTrans(\varTransTuple, \varStratTuple)$ and $\formulaObj(\varTransTuple, \varObjTuple)$ from Section~\ref{section:verification:algorithms:oeis}.
We introduce a new formula to constrain the variables $\varStratTuple$ similarly to the formula of Equation~\eqref{equation:formula:strategy:assigned-support}: we let
\begin{equation}\label{equation:formula:strategy:free-support}
  \formulaStrat(\varStratTuple) =
  \bigwedge_{\intIndex\in\integerInterval{1, \intNum}}
  \bigwedge_{\interval\in\intPart_\intIndex}
  \left(
  \bigwedge_{\ocState\in\ocStateSpace}\left(
    \bigwedge_{\ocAction\in\ocActionSpace(\ocState)}
      \varStrat_{\ocState,\ocAction}^{\interval} \geq 0\land
      \sum_{\ocAction\in\ocActionSpace(\ocState)}\varStrat_{\ocState,\ocAction}^{\interval} = 1
    \right)\land
    \bigwedge_{\interval'\in\intPart_\intIndex}
    \varStratI = \varStratIprime
  \right).
\end{equation}
Any vector $\solStratTuple$ satisfying $\formulaStrat(\solStratTuple)$ defines a strategy $\stratB_{\solStratTuple}$ based on $\intPart'$ and any such strategy induces such a vector.
We obtain the following corollary of Theorem~\ref{theorem:verification:oeis:formula}.
\begin{theorem}\label{theorem:synthesis:oeis:formula}
  Let $\ocConfig\in\compressChainStateSpace\setminus\compressChainStateSpaceStar$.
  There exists an OEIS $\strat$ based on the partition $\intPart'$ such that $\probaGverb{\ocmdpFin{\ocmdp}{\counterUB}}{\ocConfig}{\strat}(\objective) \geq\thresProba$ if and only if $\IR\models\exists\varStratTuple\forall\varTransTuple\,\forall\varObjTuple (\formulaStrat(\varStratTuple)\land((\formulaTrans(\varTransTuple,\varStratTuple)\land\formulaObj(\varTransTuple, \varObjTuple))\implies \varObj_\ocConfig\geq\thresProba))$.
\end{theorem}

We obtain that the fixed-interval realisability problem for randomised OEISs can be reduced to deciding a sentence in $\IR$ with two blocks of quantifiers.
This shows that the problem is decidable in polynomial space~\cite[Rmk.~13.10]{BPR2006}.
For the parameterised randomised OEIS realisability problem, we obtain an $\npspace = \pspace$~\cite{DBLP:journals/jcss/Savitch70} upper bound through the following algorithm: we use non-determinism to obtain a partition $\intPart'$ of $\integerInterval{1, \counterUB-1}$ compatible with the input parameters and then check the validity of the formula of Theorem~\ref{theorem:synthesis:oeis:formula} for this partition.
We obtain the following.
\begin{theorem}\label{theorem:realisability:oeis:randomised}
  The fixed-interval and parameterised randomised OEIS realisability problems for selective termination and state-reachability objectives are in $\pspace$.
\end{theorem}
\begin{proof}
  Let $\intPart'$ denote the input interval partition of $\integerInterval{1, \counterUB-1}$ if we consider the fixed-interval realisability problem or the partition obtained using non-determinism if we consider the parameterised realisability problem.
  In the latter, the representation of $\intPart'$ is of size polynomial in that of the input parameters by Lemma~\ref{lemma:partition representation}.

  Whether a formula with two blocks of quantifiers holds in the theory of the reals can be decided in polynomial space~\cite[Rmk.~13.10]{BPR2006}.
  Therefore, to obtain the claim of the theorem, it suffices to show that the formula $\formulaStrat(\varStratTuple)\land((\formulaTrans(\varTransTuple,\varStratTuple)\land\formulaObj(\varTransTuple, \varObjTuple))\implies \varObj_\ocConfig\geq\thresProba)$ can be constructed in polynomial time with respect to the representation of $\ocmdp$ and $\intPart'$.

  Lemma~\ref{lemma:ocmpd:interval size} implies that $\intPart=\mathsf{Refine}(\mathsf{Isolate}(\intPart', \ocCount_\init))$ admits a representation of size polynomial in the representation of the inputs to the fixed-interval realisability problem.
  It follows from Lemma~\ref{lemma:verification:oeis:formula:size} that the sub-formula $(\formulaTrans(\varTransTuple,\varStratTuple)\land\formulaObj(\varTransTuple, \varObjTuple))\implies \varObj_\ocConfig\geq\thresProba$ can be constructed in polynomial time.
  For $\formulaStrat(\varStratTuple)$, we observe that it is a conjunction of no more than $|\intPart|\cdot(|\ocStateSpace|\cdot(|\ocActionSpace|+1) + |\intPart|\cdot|\ocStateSpace|\cdot|\ocActionSpace|)$ atomic formulae of length in $\bigo(|\ocActionSpace|)$.
\end{proof}

\subsection{Realisability for cyclic interval strategies}\label{section:realisability:cis}

We now consider the fixed-interval and parameterised realisability problems for CISs.
We assume that $\counterUB=\infty$ for the remainder of the section.
We adapt the techniques of Section~\ref{section:realisability:oeis}.

First, we provide non-deterministic algorithms for the variants of these problems for pure CISs.
In the fixed-interval case, it suffices to non-deterministically select an action for each state of the OC-MDP and interval from the input interval partition and then verify the resulting CIS.
We now consider the parameterised case.
Let $\intNum\in\INpos$ and $\intSize\in\INpos$ respectively denote the parameter bounding the number of intervals and the size of intervals for the desirable interval partitions.
To solve the parameterised realisability problem, we guess a period $\period\leq\intNum\cdot\intSize$, an interval partition $\intPartB'$ of $\integerInterval{1, \period}$ that is compatible with $\intNum$ and $\intSize$ and actions for all pairs in $\ocStateSpace\times\intPartB'$, then verify the obtained CIS.
Theorem~\ref{theorem:verification:cis:complexity}, along with Lemma~\ref{lemma:partition representation} in the parameterised case, imply that both of these problems can be solved in non-deterministic polynomial time with an $\etr$ oracle.
\begin{theorem}\label{theorem:realisability:cis:pure}
  The fixed-interval and parameterised pure CIS realisability problems for selective termination and state-reachability objectives are in $\np^\etr$.
\end{theorem}

We now consider the problem variants for randomised strategies.
As before, we first focus on the fixed-interval case.
Let $\period\in\INpos$ denote the input period and $\intPartB' = (\interval_{\intIndex})_{\intIndex\in\integerInterval{1, \intNum}}$ denote the input interval partition of $\integerInterval{1, \period}$.
We  define $\intPartB=\mathsf{Refine}(\mathsf{Isolate}(\intPartB', \ocCount_\init\bmod\period))$ of $\integerInterval{1, \period}$.
We let, for all $\intIndex\in\integerInterval{1, \intNum}$, $\intPartB_\intIndex = \{\interval\in\intPartB\mid\interval\subseteq\interval_\intIndex\}$.
We also let $\intPartC = \mathsf{Refine}(\integerInterval{1, \lfloor\frac{\ocCount_\init}{\period}\rfloor})\cup\{\integerInterval{\lfloor\frac{\ocCount_\init}{\period}\rfloor+1, \infty}\}$.
These choices guarantee that the counterpart in the sense of Theorem~\ref{theorem:cis:compression:ocmc} of $\ocConfig_\init$ in $\compressCisStateSpace$ exists.

Next, we reintroduce the sets of variables $\varStratTuple$, $\varCisTuple$, $\varTransTuple$ and $\varObjTuple$ and the formulae $\formulaCis(\varCisTuple, \varStratTuple)$, $\formulaCisTrans(\varTransTuple, \varCisTuple)$ and $\formulaCisObj(\varTransTuple, \varObjTuple)$ from Section~\ref{section:verification:algorithms:cis}.
We formulate an adaptation of the formulae of Equations~\eqref{equation:formula:strategy:assigned-support} and~\eqref{equation:formula:strategy:free-support} with respect to $\intPartB$: we let
\begin{equation}
  \formulaStratCis(\varStratTuple) =
  \bigwedge_{\intIndex\in\integerInterval{1, \intNum}}
  \bigwedge_{\interval\in\intPartB_\intIndex}
  \left(
  \bigwedge_{\ocState\in\ocStateSpace}\left(
    \bigwedge_{\ocAction\in\ocActionSpace(\ocState)}
      \varStrat_{\ocState,\ocAction}^{\interval} \geq 0\land
      \sum_{\ocAction\in\ocActionSpace(\ocState)}\varStrat_{\ocState,\ocAction}^{\interval} = 1
    \right)\land
    \bigwedge_{\interval'\in\intPartB_\intIndex}
    \varStratI = \varStratIprime
  \right).
\end{equation}
Once again, we obtain a natural correspondence between vectors $\solStratTuple$ satisfying $\formulaStratCis(\solStratTuple)$ and the CISs considered for our realisability problem.
The following theorem is therefore implied by Theorem~\ref{theorem:verification:cis:formula}.
\begin{theorem}\label{theorem:synthesis:cis:formula}
  Let $\cisConfig\in\compressCisStateSpace\cap(\cisChainStateSpace\times\INpos)$.
  There exists a CIS $\strat$ based on the periodic partition generated by $\intPartB'$ such that $\probaMCverb{\ocChainFin{\cisChain}{\infty}}{\cisConfig}(\reach{\bar{\target}})\geq\thresProba$ if and only if $\IR\models\exists\varStratTuple\forall\varTransTuple\,\forall\varObjTuple\forall\varCisTuple (\formulaStratCis(\varStratTuple)\land ((\formulaCis(\varCisTuple, \varStratTuple)\land\formulaCisTrans(\varTransTuple,\varCisTuple)\land\formulaCisObj(\varTransTuple, \varObjTuple))\implies \varObj_{\cisConfig}\geq\thresProba))$.
\end{theorem}

We thus obtain that the fixed-interval realisability problem for randomised CISs is reducible in polynomial time to deciding whether a sentence with two quantifier blocks holds in the theory of the reals.
We obtain a $\pspace$ upper bound~\cite[Rmk.~13.10]{BPR2006}.
For the parameterised case, we obtain an $\npspace = \pspace$~\cite{DBLP:journals/jcss/Savitch70} upper bound; we non-deterministically guess an interval partition $\intPartB'$ as we have done for the parameterised pure CIS realisability problem, then reduce to checking the validity of the formula of Theorem~\ref{theorem:synthesis:cis:formula} as in the fixed-interval case.
The following theorem summarises our complexity bounds.

\begin{theorem}\label{theorem:realisability:cis:randomised}
  The fixed-interval and parameterised randomised CIS realisability problems for selective termination and state-reachability objectives are in $\pspace$.
\end{theorem}
\begin{proof}
  This theorem follows from an adaptation of the proof of Theorem~\ref{theorem:realisability:oeis:randomised}.
  The major difference, in this case, is that we refer to Lemma~\ref{lemma:verification:cis:formula:size} instead of Lemma~\ref{lemma:verification:oeis:formula:size} for bounds on the size of the sentence used in to solve the realisability problem.
\end{proof}

\section{Lower bounds}\label{section:hardness}
We present lower complexity bounds for the interval strategy verification problem, the fixed-interval realisability problem and the parameterised realisability problem.
In Section~\ref{section:hardness:sqs}, we prove the square-root-sum hardness of all variants of these problems.
We show that our interval strategy realisability problems are $\np$-hard when considering selective termination in Section~\ref{section:hardness:np}.

\subsection{Square-root-sum hardness}\label{section:hardness:sqs}
The square-root-sum problem consists in comparing a sum of square roots of natural numbers to some integer bound.
Formally, the inputs to the problems are natural numbers $\sqsx_1, \ldots, \sqsx_\sqsn$ and a natural number $\sqsy$.
We formalise the \textit{square-root-sum problem} as the problem of determining whether $\sum_{\sqsi=1}^\sqsn\sqrt{\sqsx_\sqsi}\geq\sqsy$.

\begin{remark}
  The square-root-sum problem is typically formulated as having to decide whether $\sum_{\sqsi=1}^\sqsn\sqrt{\sqsx_\sqsi}\leq\sqsy$, i.e., with the opposite inequality.
  For the sake of illustrating the hardness of a problem, both problems can be seen as equally suitable.

  We argue this by briefly showing that an efficient solution to either variant of the problem yields an efficient solution to the other one.
  We observe that these two variants are almost the complement of one another.
  The only case in which the two problems have the same solution for the same inputs is when $\sum_{\sqsi=1}^\sqsn\sqrt{\sqsx}_\sqsi= \sqsy$.
  Deciding whether $\sum_{\sqsi=1}^\sqsn\sqrt{\sqsx}_\sqsi= \sqsy$ can be done in polynomial time~\cite{DBLP:journals/jsc/BorodinFHT85}. Therefore, an efficient decision procedure for one variant of the square-root-sum problem would entail an efficient one for the other.
  \hfill$\lhd$
\end{remark}

The square-root-sum problem is not known to be solvable in polynomial time in the Turing model of computation.
It is known that the square-root-sum problem can be solved in polynomial time in the BSS model~\cite{DBLP:journals/jc/Tiwari92}.
In particular, the square-root-sum problem is in $\ptime^\posSLP$~\cite{DBLP:journals/siamcomp/AllenderBKM09} and thus in the counting hierarchy.

Our hardness result rely on an existing reduction from the square-root-sum problem to (a variant of) the verification for one-counter Markov chains from~\cite{DBLP:journals/pe/EtessamiWY10}.
We recall this reduction (adapted to our formalism) in Section~\ref{section:hardness:unbounded} and derive square-root-sum hardness for our problems in unbounded OC-MDPs from it.
In Section~\ref{section:hardness:bounded}, we adapt the same reduction to the bounded setting by proving that we can, in polynomial time, compute a large enough counter bound so the bounded one-counter Markov chain approximates the one used in the unbounded reduction well enough for the reduction to still be valid.

\subsubsection{Unbounded one-counter Markov decision processes}\label{section:hardness:unbounded}

We show that the reduction from~\cite{DBLP:journals/pe/EtessamiWY10} from the square-root-sum problem to a verification problem for selective termination in one-counter Markov chains can be used to obtain lower complexity bounds for our interval strategy problems in unbounded OC-MDPs.
We fix inputs $\sqsx_1, \ldots, \sqsx_\sqsn$ and $\sqsy$ to the square-root-sum problem, and let $\sqsm = \max_{1\leq\sqsi\leq\sqsn}\sqsx_\sqsi$ and $\sqsxVect = (\sqsx_1,\ldots, \sqsx_\sqsn)$.
In~\cite{DBLP:journals/pe/EtessamiWY10}, the authors provide a one-counter Markov chain based on $\sqsxVect$ such that the probability of terminating in a given state $\ocStateC$ is $\frac{1}{\sqsn\sqsm}\sum_{\sqsi=1}^\sqsn\sqrt{\sqsx_\sqsi}$.
We reformulate this one-counter Markov chain as an equivalent OC-MDP $\chainX$ with one action.

We depict the fragment of $\chainX$ that is associated with $\sqsx_\sqsi$ in Figure~\ref{figure:sqs:gadget}.
Formally, we define $\chainX=(\ocStateSpace_{\sqsxVect}, \{\ocAction\}, \ocTrans_{\sqsxVect}, \weight_{\sqsxVect})$ where $\ocStateSpace_{\sqsxVect} = \{\ocState_\init, \ocStateC\}\cup\bigcup_{1\leq\sqsx\leq\sqsn}\{\ocState_\sqsi,\ocState_\sqsi^+,\ocState_\sqsi^-\}$ and, for all $\sqsi\in\integerInterval{1, \sqsn}$, transitions and weights to and from the state $\ocState_\sqsi$, $\ocState_\sqsi^+$, $\ocState_\sqsi^-$ and $\ocStateC$  match those in the illustration.
For any $\counterUB\in\INposBar$, we let $\ocChainFin{\chainX}{\counterUB}$ denote the Markov chain induced by the sole strategy of $\ocmdpFin{\chainX}{\counterUB}$.
We have the following theorem (it can also be seen as a corollary of Theorem~\ref{theorem:equations:termination}).
  \begin{figure}
    \centering
    \begin{tikzpicture}
      \node[state, align=center] (q0) {$\ocState_\init$};
      \node[stochasticc, right = of q0] (q0mid) {};
\node[state, align=center, right = of q0mid] (q) {$\ocState_\sqsi$};
      \node[stochasticc, right = of q] (qt) {};
      \node[state, align=center, right = of qt] (qplus) {$\ocState_\sqsi^+$};
      \node[state, align=center, below = of qt] (qminus) {$\ocState_\sqsi^-$};
      \node[stochasticc, below = of q] (qtm) {};
      \node[state, left = of qtm] (t) {$\ocStateC$};

      \path[-] (q0) edge node[above] {$\ocAction\mid0$} (q0mid);
\path[->] (q0mid) edge node[above] {$\frac{1}{\sqsn}$} (q);
      \path[-] (q) edge node [below] {$\ocAction\mid 0$} (qt);
      \path[-] (qminus) edge node[below] {$\ocAction\mid -1$} (qtm);
      \path[->] (qt) edge node[below, align=center] {$\frac{1}{2}$} (qplus);
      \path[->] (qt) edge node[right, align=center] {$\frac{1}{2}$} (qminus);
      \path[->] (qplus) edge[bend right] node[above, align=center] {$\ocAction\mid 1$} (q);
      \path[->] (qtm) edge node[left, align=center] {$1 - \frac{\sqsx_\sqsi}{\sqsm^2}$} (q);
      \path[->] (qtm) edge node[below, align=center] {$\frac{\sqsx_\sqsi}{\sqsm^2}$} (t);
      \path[->] (t) edge[loop left] node[left, align=center] {$\ocAction\mid -1$} (t);
    \end{tikzpicture}
    \caption{A fragment of $\chainX$. Transition probabilities are well-defined: $\ocState_\init$ has $\sqsn$ successors and its outgoing transition share the same probabilities and, for the states $\ocState_\sqsi^-$, we have $\sqsx_\sqsi\leq\sqsm$.}\label{figure:sqs:gadget}
  \end{figure}
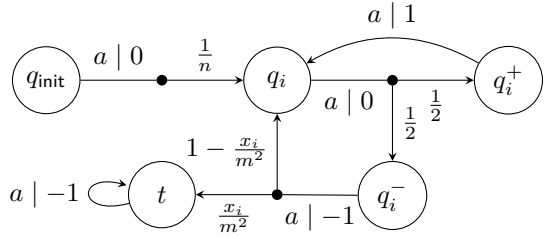

  \begin{theorem}[\cite{DBLP:journals/pe/EtessamiWY10}]\label{theorem:reduction:probability:chain}
    We have $\probaMCverb{\ocChainFin{\chainX}{\infty}}{(\ocState_\init, 1)}(\selectiveTermination{\ocStateC}) = \frac{1}{\sqsn\sqsm}\sum_{\sqsi=1}^\sqsn\sqrt{\sqsx_\sqsi}$ and, for all $1\leq\sqsi\leq\sqsn$, $\probaMCverb{\ocChainFin{\chainX}{\infty}}{(\ocState_\sqsi, 1)}(\selectiveTermination{\ocStateC}) = \frac{1}{\sqsm}\sqrt{\sqsx_\sqsi}$.
\end{theorem}

Due to the structure of $\chainX$, reaching $\ocStateC$ and terminating in $\ocStateC$ are equivalent.
Thus, Theorem~\ref{theorem:reduction:probability:chain} implies that we can reduce the square-root sum instance fixed above to the verification problem for selective termination and state-reachability on $\chainX$ for the unique (counter-oblivious) strategy of $\ocmdpFin{\chainX}{\infty}$, which is both an OEIS and a CIS, and the threshold $\thresProba=\frac{\sqsy}{\sqsn\sqsm}$.
Furthermore, as there is only a single strategy, the answer to the verification problem is the same as the answer to the realisability problem for (pure or randomised) counter-oblivious strategies, which is a special case of the fixed-interval and parameterised interval strategy realisability problems for OEISs and CISs.

\begin{theorem}\label{verification:hardness:unbounded}
  The interval strategy verification, fixed-interval realisability and parameterised realisability problems for state-reachability and selective termination are square-root-sum-hard in unbounded OC-MDPs.
\end{theorem}
\begin{proof}
  The OC-MDP $\chainX$ and the threshold $\thresProba = \frac{\sqsy}{\sqsn\sqsm}$ can be computed in polynomial time.
  Therefore, we need only comment on the correctness of the reduction.
  
  By Theorem~\ref{theorem:reduction:probability:chain} and due to the structure of $\chainX$, we have $\probaMCverb{\ocChainFin{\chainX}{\infty}}{(\ocConfig_\init, 1)}(\reach{\ocStateC}) = \probaMCverb{\ocChainFin{\chainX}{\infty}}{(\ocConfig_\init, 1)}(\selectiveTermination{\ocStateC}) = \frac{1}{\sqsn\sqsm}\sum_{\sqsi=1}^\sqsn\sqrt{\sqsx_\sqsi}$.
  Let $\objective\in\{\selectiveTermination{\ocStateC}, \reach{\ocStateC}\}$.
  We clearly have $\probaMCverb{\ocChainFin{\chainX}{\infty}}{(\ocConfig_\init, 1)}(\objective)\geq\thresProba$ if and only if $\sum_{\sqsi=1}^\sqsn\sqrt{\sqsx_\sqsi}\geq\sqsy$ in light of the above, and thus the reduction is correct.
\end{proof}

\subsubsection{Bounded one-counter Markov decision processes}\label{section:hardness:bounded}

We now establish the square-root sum hardness for the interval strategy verification, fixed-interval realisability and parameterised realisability problems  in bounded OC-MDPs.
Intuitively, in most cases, the reduction consists in adding a counter upper bound to the reduction of Section~\ref{section:hardness:unbounded}.
We fix inputs $\sqsx_1$, \ldots, $\sqsx_\sqsn$ and $\sqsy$ to the square-root-sum problem for the remainder of the section and let $\sqsm = \max_{1\leq\sqsi\leq\sqsn}\sqsx_\sqsi$ and $\sqsxVect = (\sqsx_1,\ldots, \sqsx_\sqsn)$.
Let $\thresProba = \frac{\sqsy}{\sqsn\sqsm}$ denote the threshold used for the reduction.
We assume that for all $\sqsi\in\integerInterval{1, \sqsn}$, $\sqsx_\sqsi\neq 0$, and thus that $\sqsm\geq1$.

We aim to determine a bound $\counterUB\in\INpos$ such that $\probaMCverb{\ocChainFin{\chainX}{\counterUB}}{(\ocState_\init, 1)}(\selectiveTermination{\ocStateC})\geq\thresProba$ if and only if $\sum_{\sqsi=1}^\sqsn\sqrt{\sqsx_\sqsi}\geq\sqsy$.
For all $\counterUB\in\INpos$, let $\eleError{\counterUB} = \probaMCverb{\ocChainFin{\chainX}{\infty}}{(\ocState_\init, 1)}(\reach{\ocStateSpace\times\{\counterUB\}})$.
Using Theorem~\ref{theorem:reduction:probability:chain}, we obtain that for all $\counterUB\in\INpos$, we have
\begin{equation*}
  \probaMCverb{\ocChainFin{\chainX}{\counterUB}}{(\ocState_\init, 1)}(\selectiveTermination{\ocStateC})
  = \probaMCverb{\ocChainFin{\chainX}{\infty}}{(\ocState_\init, 1)}(\selectiveTermination{\ocStateC})
  - \eleError{\counterUB}
  = \frac{1}{\sqsn\sqsm}\sum_{\sqsi=1}^\sqsn\sqrt{\sqsx_\sqsi}
  -\eleError{\counterUB}.
\end{equation*}
Therefore, we require a bound $\counterUB\in\INpos$ with a polynomial-size representation such that the positive error term $\eleError{\counterUB}$ above is small enough to ensure the correctness of the reduction.

There is a particular case for which it is clear that no suitable $\counterUB$ exists: whenever $\sum_{\sqsi=1}^\sqsn\sqrt{\sqsx_\sqsi} = \sqsy$ holds.
Since this equality can be decided in polynomial time~\cite{DBLP:journals/jsc/BorodinFHT85}, we consider a reduction that is conditioned on this equality.
First, we check if the equality holds in polynomial time.
If it does, we reduce to a fixed positive instance of the considered interval strategy problem (in bounded OC-MDPs).
Otherwise, we mirror the reduction of the unbounded case with a well-chosen counter upper bound $\counterUB$.

From here, we assume that $\sum_{\sqsi=1}^\sqsn\sqrt{\sqsx_\sqsi} \neq \sqsy$ and try to determine a suitable bound $\counterUB$.
We first formulate a lower bound on the distance $|\sum_{\sqsi=1}^\sqsn\sqrt{\sqsx_\sqsi} - \sqsy|$.
We adapt~\cite[Lemma~3]{DBLP:journals/jc/Tiwari92} to suit our formulation of the square-root-sum problem.
For the sake of completeness, we provide an adapted proof in Appendix~\ref{appendix:square root sum}.
\begin{restatable}{lemma}{lemmaSqrtSumError}\label{lemma:sqrt-sum:error}
  Let $\lambda$ be the sum of the bit-sizes of $\sqsx_1, \ldots, \sqsx_\sqsn$ and $\sqsy$.
  If $\sum_{\sqsi=1}^\sqsn\sqrt{\sqsx_\sqsi} \neq\sqsy$, then $|\sum_{\sqsi=1}^\sqsn\sqrt{\sqsx_\sqsi}  - \sqsy| > 2^{-2^{\sqsn}(\lambda+1)}$.
\end{restatable}

Lemma~\ref{lemma:sqrt-sum:error} implies that our reduction is correct whenever we ensure that $\sqsn\cdot\sqsm\cdot\eleError{\counterUB}\leq 2^{-2^{\sqsn}(\lambda+1)}$ where $\lambda$ denotes the sum of bit-sizes of the inputs to the square-root-sum problem.
We claim that choosing $\counterUB = 2^\sqsn\sqsm\cdot(\lambda+1)+\sqsn\sqsm^2+1$ is sufficient.
This bound can be computed in polynomial time as $\sqsn$ is the number of inputs.
To show that this choice of $\counterUB$ is suitable, we first bound $\eleError{\counterUB}$ by an explicit function in $\counterUB$ and $\sqsm$.

\begin{lemma}\label{lemma:sqrt-sum:bound:domination}
  For all $\counterUB\in\INpos$, $\eleError{\counterUB}\leq (\frac{\sqsm}{\sqsm+1})^{\counterUB-1}$.
\end{lemma}
\begin{proof}
  We have $\eleError{1}=1 = (\frac{\sqsm}{\sqsm+1})^0$, and thus the inequality holds trivially for $\counterUB=1$.
  For the remainder of the proof, we prove properties that hold for all $\counterUB\geq 2$ to obtain the general result.
  
  For all $\counterUB\geq 2$ and $\sqsi\in\integerInterval{1, \sqsn}$, let $\eleB{\counterUB} = \probaMCverb{\ocChainFin{\chainX}{\infty}}{(\ocState_\sqsi, 1)}(\reach{(\ocState_\sqsi, \counterUB)})$ denote the probability of hitting counter value $\counterUB$ from $(\ocState_\sqsi, 1)$.
  For all $\counterUB\geq 2$, we have $\eleError{\counterUB} = \frac{1}{\sqsn}\sum_{\sqsi=1}^\sqsn\eleB{\counterUB}$ due to the topology of $\chainX$.
  To obtain the lemma, it suffices to show that for all $\sqsi\in\integerInterval{1, \sqsn}$, we have $\eleB{\counterUB}\leq(\frac{\sqsm}{\sqsm+1})^{\counterUB-1}$.
  We fix $\sqsi\in\integerInterval{1, \sqsn}$.

  For all $\counterUB\geq 2$, we let $\eleD{\counterUB} =  \probaMCverb{\ocChainFin{\chainX}{\infty}}{(\ocState_\sqsi, \counterUB-1)}(\reach{(\ocState_\sqsi, \counterUB)})$ denote the probability of reaching counter value $\counterUB$ from $(\ocState_\sqsi, \counterUB-1)$.
  To conclude this proof, we establish three statements.
  First, we show that for all $\counterUB\geq 2$, we have $\eleB{\counterUB+1} = \eleB{\counterUB}\cdot\eleD{\counterUB+1}$.
  Second, we show that for all $\counterUB\geq 2$, we have $\eleD{\counterUB}\leq\frac{\sqsm}{\sqsm+1}$.
  Finally, we combine these two properties to conclude using an inductive argument.

  For all $\counterUB\in\INpos$ and $\ocCount\in\integerInterval{1, \counterUB}$, we let $\histPart_{\ocCount\to\counterUB}$ denote the set of histories of $\ocmdpFin{\chainX}{\infty}$ starting in $(\ocState_\sqsi,\ocCount)$ and ending in $(\ocState_\sqsi, \counterUB)$ with only one occurrence of this last configuration.
  The sets $\histPart_{\ocCount\to\counterUB}$ are prefix-free; the cylinders of their elements are pairwise disjoint.

  We now show the first claim.
  Let $\counterUB\geq 2$.
  All histories of $\histPart_{1\to\counterUB}$ can be written as the concatenation of a history of $\histPart_{1\to\counterUB-1}$ and a history of $\histPart_{\counterUB-1\to\counterUB}$.
  We obtain
  \begin{align*}
    \eleB{\counterUB}
    & = \probaMCverb{\ocChainFin{\chainX}{\infty}}{(\ocState_\sqsi, 1)}(\cyl{\histPart_{1\to\counterUB}}) \\
    & = \sum_{\hist_1\in\histPart_{1\to\counterUB-1}}\sum_{\hist_2\in\histPart_{\counterUB-1\to\counterUB}}\probaMCverb{\ocChainFin{\chainX}{\infty}}{(\ocState_\sqsi, 1)}(\cyl{\histConcat{\hist_1}{\hist_2}}) \\
    & = \left(\sum_{\hist_1\in\histPart_{1\to\counterUB-1}}\probaMCverb{\ocChainFin{\chainX}{\infty}}{(\ocState_\sqsi, 1)}(\cyl{\hist_1})\right)\cdot\left(\sum_{\hist_2\in\histPart_{\counterUB-1\to\counterUB}}\probaMCverb{\ocChainFin{\chainX}{\infty}}{(\ocState_\sqsi, \counterUB-1)}(\cyl{\hist_2})\right) \\
    & = \eleB{\counterUB-1}\cdot\eleD{\counterUB}.
  \end{align*}
  This proves the first claim.

  For the second claim, we show that the sequence $\seqD$ is increasing and convergent, and that $\lim_{\counterUB\to\infty}\eleD{\counterUB}\leq\frac{\sqsm}{\sqsm+1}$.
  Let us prove that $\seqD$ is increasing.
  Let $\counterUB\geq 2$.
  We consider the mapping $f_{+1}\colon\histPart_{\counterUB-1\to\counterUB}\to\histPart_{\counterUB\to\counterUB+1}$ that increases all counter values along a history by $1$.
  This mapping is injective.
  Furthermore, for all $\hist\in\histPart_{\counterUB-1\to\counterUB}$, we have $\probaMCverb{\ocChainFin{\chainX}{\infty}}{(\ocState_\sqsi, \counterUB-1)}(\cyl{\hist}) = \probaMCverb{\ocChainFin{\chainX}{\infty}}{(\ocState_\sqsi, \counterUB)}(\cyl{f_{+1}(\hist)})$.
  It follows that
  \[\eleD{\counterUB} =
    \probaMCverb{\ocChainFin{\chainX}{\infty}}{(\ocState_\sqsi, \counterUB-1)}(\cyl{\histPart_{\counterUB-1\to\counterUB}})  \leq
    \probaMCverb{\ocChainFin{\chainX}{\infty}}{(\ocState_\sqsi, \counterUB)}(\cyl{\histPart_{\counterUB\to\counterUB+1}}) =
    \eleD{\counterUB+1}.\]
  This shows that $\seqD$ is increasing.

  The sequence $\seqD$ is bounded and increasing, thus it converges.
  We prove that $\lim_{\counterUB\to\infty}\eleD{\counterUB}=\frac{\sqsm}{\sqsm+\sqrt{\sqsx_\sqsi}}$.
  To this end, we establish an inductive relation on the elements of this sequence: we prove that for all $\counterUB\geq 2$, we have $\eleD{\counterUB+1} = \frac{1}{2} + (1-\frac{\sqsx_\sqsi}{\sqsm^2})\eleD{\counterUB}\cdot \eleD{\counterUB+1}$.
  Let $\counterUB\geq 2$.
  By separating histories for which a counter increment occurs first from those for which a counter decrement occurs first, we obtain that
  \[\eleD{\counterUB+1} =
    \frac{1}{2} +
    \frac{1}{2}\cdot \left(1-\frac{\sqsx_\sqsi}{\sqsm^2}\right)\cdot
    \probaMCverb{\ocChainFin{\chainX}{\infty}}{(\ocState_\sqsi, \counterUB-1)}(\cyl{\histPart_{\counterUB-1\to\counterUB+1}}).\]
  We obtain $\eleD{\counterUB+1} = \frac{1}{2} + (1-\frac{\sqsx_\sqsi}{\sqsm^2})\eleD{\counterUB-1}\cdot \eleD{\counterUB}$ by observing that elements of $\histPart_{\counterUB-1\to\counterUB+1}$ can be written as concatenations of elements of $\histPart_{\counterUB-1\to\counterUB}$ and $\histPart_{\counterUB\to\counterUB+1}$ and following the same reasoning as for $\eleB{\counterUB}$ above.

  By taking the limits on both sides of the above inductive relation, we obtain that $\lim_{\counterUB\to\infty}\eleD{\counterUB} = \frac{1}{2} + (1-\frac{\sqsx_\sqsi}{\sqsm^2})\cdot(\lim_{\counterUB\to\infty}\eleD{\counterUB})^2$.
  If $\sqsm=1$, then $\sqsx_\sqsi=1$ (inputs are positive) and we directly obtain $\lim_{\counterUB\to\infty}\eleD{\counterUB} = \frac{1}{2} = \frac{\sqsm}{\sqsm+\sqrt{\sqsx_\sqsi}}$.
  Otherwise, if $\sqsm\geq 2$, we have $\frac{\sqsx_\sqsi}{\sqsm^2}<1$ and we can solve a quadratic equation to deduce that $\lim_{\counterUB\to\infty}\eleD{\counterUB}\in\{\frac{\sqsm}{\sqsm +\sqrt{\sqsx_\sqsi}}, \frac{\sqsm}{\sqsm - \sqrt{\sqsx_\sqsi}}\}$.
  It follows from $\frac{\sqsm}{\sqsm - \sqrt{\sqsx_\sqsi}}> 1$ and $\seqD$ being a sequence of probabilities that $\lim_{\counterUB\to\infty}\eleD{\counterUB}=\frac{\sqsm}{\sqsm +\sqrt{\sqsx_\sqsi}}$.
  To end the proof of the second claim, we observe that since $\seqD$ is increasing and $\sqsx_\sqsi\geq 1$, we have, for all $\counterUB\geq 2$, $\eleD{\counterUB}\leq\frac{\sqsm}{\sqsm+\sqrt{\sqsx_\sqsi}}\leq\frac{\sqsm}{\sqsm+1}$.

  We now combine the two claims to provide an inductive proof that $\eleB{\counterUB}\leq(\frac{\sqsm}{\sqsm+1})^{\counterUB-1}$ for all $\counterUB\geq 2$.
  For $\counterUB=2$, we have $\eleB{2}=\frac{1}{2}\leq \frac{\sqsm}{\sqsm+1}$ (because $\sqsm\geq 1$).
  We now assume that $\eleB{\counterUB}\leq(\frac{\sqsm}{\sqsm+1})^{\counterUB-1}$ holds.
  Via the two claims, we conclude that $\eleB{\counterUB+1} = \eleB{\counterUB}\cdot\eleD{\counterUB+1}\leq(\frac{\sqsm}{\sqsm+1})^{\counterUB}$.
\end{proof}

The following result is a technical inequality required to show that the candidate for $\counterUB$ for the reduction is well-chosen.
We separate it from the main proof for the sake of clarity.
\begin{lemma}\label{lemma:sqrt-sum:bound:ub:inequality}
  It holds that $\frac{1}{\log_2(\frac{\sqsm+1}{\sqsm})}\leq\sqsm$.
\end{lemma}
\begin{proof}
  The inequality above is equivalent to $1+\frac{1}{\sqsm}\geq 2^{\frac{1}{\sqsm}}$.
  To prove this equivalent formulation, we show that the function $f\colon\coInt{1}{+\infty}\to\IR\colon z\mapsto 1 + \frac{1}{z} - 2^{\frac{1}{z}}$ is non-negative.
  Let $z_0 = -\log_2(\ln(2))^{-1} > 1$.
  We show that $f$ is increasing on $\ccInt{1}{z_0}$ and decreasing on $\coInt{z_0}{+\infty}$.
  This property implies that $f$ is non-negative.
  Indeed, on the one hand, we have $f(1) = 0$, implying that $f$ is non-negative on $\ccInt{1}{z_0}$.
  On the other hand, because $\lim_{z\to+\infty}f(z) = 0$, $f$ is necessarily non-negative on the interval $\coInt{z_0}{+\infty}$.

  We study the sign of the derivative of $f$ to determine its intervals of monotonicity.
  We define $g\colon\coInt{1}{+\infty}\to\IR$ such that, for all $z\in\coInt{1}{+\infty}$, $g(z) = \ln(2)\cdot 2^{\frac{1}{z}}-1$.
  For all $z\in\coInt{1}{+\infty}$, we have $f'(z) = \frac{1}{z^2}g(z)$.
  We obtain that for all $z\in\coInt{1}{+\infty}$, the sign of $f'(z)$ depends only on the sign of $g(z)$.
  The function $g$ is a decreasing function, because $\ooInt{1}{+\infty}\to\IR\colon z\mapsto 2^{\frac{1}{z}}$ is a restriction of the composition of the decreasing function $\ooInt{0}{+\infty}\to\IR\colon z\mapsto\frac{1}{z}$ and the increasing function $\ooInt{0}{+\infty}\to\IR\colon z\mapsto 2^z$.
  Because $g(1) = 2\ln(2) - 1>0$ and $g(z_0) = 0$, it follows that $f'$ is positive on the interval $\ocInt{1}{z_0}$ and negative on the interval $\ooInt{z_0}{+\infty}$.
  This implies the desired property for $f$, ending the proof.
\end{proof}

We can now show that choosing $\counterUB=2^\sqsn\sqsm\cdot(\lambda+1)+\sqsn\sqsm^2+1$ (where $\lambda$ is the sum of bit-sizes of the inputs to our square-root sum instance) is sufficient to achieve the precision given by Lemma~\ref{lemma:sqrt-sum:error} that ensures the validity of the reduction.

\begin{lemma}\label{lemma:sqrt-sum:bound:ub:error}
  Let $\lambda\in\INpos$.
  For all $\counterUB \geq 2^\sqsn\sqsm\cdot(\lambda+1)+\sqsn\sqsm^2+1$, we have $\eleError{\counterUB}\leq \frac{1}{\sqsn\sqsm}2^{-2^\sqsn(\lambda+1)}$.
\end{lemma}
\begin{proof}
  We claim that it is sufficient to show that, for all $\counterUB \geq 2^\sqsn\sqsm\cdot(\lambda+1)+\sqsn\sqsm^2+1$,
  \begin{equation}\label{equation:sqrt-sum:bound:ub:error:alt}
    \left(\frac{\sqsm}{\sqsm+1}\right)^{\counterUB-1}\leq
    2^{-2^\sqsn(\lambda+1)-\sqsn\sqsm}.
  \end{equation}
  We observe that $\sqsm$, $\sqsn\in\INpos$ implies that $ 2^{-\sqsn\sqsm}\leq\frac{1}{\sqsn\sqsm}$.
  Combining this with Lemma~\ref{lemma:sqrt-sum:bound:domination} implies that, for all $\counterUB\in\INpos$ such that Equation~\eqref{equation:sqrt-sum:bound:ub:error:alt} holds,
  \[\eleError{\counterUB} \leq
    \left(\frac{\sqsm}{\sqsm+1}\right)^{\counterUB-1}\leq
    2^{-2^\sqsn(\lambda+1)-\sqsn\sqsm}\leq
    \frac{1}{\sqsn\sqsm}2^{-2^\sqsn(\lambda+1)}.\]
  This guarantees that establishing Equation~\eqref{equation:sqrt-sum:bound:ub:error:alt} for the relevant upper bounds $\counterUB$ is sufficient.

  For all $\counterUB\in\INpos$, Equation~\eqref{equation:sqrt-sum:bound:ub:error:alt} is equivalent  (by applying $\log_2$ on both sides then using algebraic manipulations) to
    \begin{equation}\label{equation:sqrt-sum:bound:ub:error:2}
      \counterUB-1\geq \frac{2^\sqsn(\lambda+1)+\sqsn\sqsm}{\log_2\left(\frac{\sqsm+1}{\sqsm}\right)}.
    \end{equation}
    By Lemma~\ref{lemma:sqrt-sum:bound:ub:inequality}, Equation~\eqref{equation:sqrt-sum:bound:ub:error:2} is guaranteed to hold whenever $\counterUB-1\geq\sqsm(2^\sqsn(\lambda+1)+\sqsn\sqsm)$, and thus the same applies to Equation~\eqref{equation:sqrt-sum:bound:ub:error:alt}.
    This ends the proof of this lemma.
\end{proof}

We can now formulate and prove our square-root-sum hardness result for our OEIS-related problems in bounded OC-MDPs.

\begin{theorem}\label{verification:hardness:bounded}
  The interval strategy verification, fixed-interval realisability and parameterised realisability problems for state-reachability and selective termination are square-root-sum-hard in bounded OC-MDPs.
\end{theorem}
\begin{proof}
  The reduction only differ slightly between the three considered problem; we discuss the verification problem and comment on the additional steps for the other two problems below.
  The reduction is the same for state reachability and selective termination, and thus we only mention the target in the following without specifying the objective.
  We consider inputs $\sqsx_1$, \ldots, $\sqsx_\sqsn$ and $\sqsy$ to the square-root-sum problem.
  We assume that for all $\sqsi\in\integerInterval{1, \sqsn}$, $\sqsx_\sqsi\neq 0$.
  Let $\sqsm = \max_{1\leq\sqsi\leq\sqsn}\sqsx_\sqsi$ and $\sqsxVect = (\sqsx_1,\ldots, \sqsx_\sqsn)$.
  We describe the reduction and prove its correctness.
  
  First, we check in polynomial time whether $\sum_{\sqsi=1}^\sqsn\sqrt{\sqsx_\sqsi}=\sqsy$.
  If this equality holds, we construct an OC-MDP $\ocmdp$ with a single state $\ocState$ and a single action $\ocAction$ where the self-loop of $\ocState$ labelled by $\ocAction$ has weight $-1$.
  We reduce our instance of the square-root-sum problem to the verification problem on $\ocmdp$ with counter upper bound $\counterUB=2$, initial configuration $(\ocState, 1)$, target $\{\ocState\}$ and threshold $\thresProba=1$.
  The reduction is trivially correct in this case and is in polynomial time.

  If $\sum_{\sqsi=1}^\sqsn\sqrt{\sqsx_\sqsi}\neq\sqsy$, we construct the OC-MDP $\chainX$.
  Let $\counterUB = 2^\sqsn\sqsm\cdot(\lambda+1)+\sqsn\sqsm^2+1$ where $\lambda$ is the sum of bit-sizes of the inputs of the square-root-sum instance.
  We reduce our instance of the square-root-sum problem to the verification problem on $\chainX$ with counter upper bound $\counterUB$, initial configuration $(\ocState_\init, 1)$, target $\{\ocStateC\}$ and threshold $\thresProba = \frac{\sqsy}{\sqsn\sqsm}$.
  This reduction is in polynomial time, and thus it remains to prove its correctness.

  Let $\eleError{\counterUB} = \probaMCverb{\ocChainFin{\chainX}{\infty}}{(\ocState_\init, 1)}(\reach{\ocStateSpace\times\{\counterUB\}})$.
  It follows from Theorem~\ref{theorem:reduction:probability:chain} that we must show that $\frac{1}{\sqsn\sqsm}\sum_{\sqsi=1}^\sqsn\sqrt{\sqsx_\sqsi} - \eleError{\counterUB}\geq\thresProba$ if and only if $\sum_{\sqsi=1}^\sqsn\sqrt{\sqsx_\sqsi}\geq\sqsy$.
  It is direct that $\frac{1}{\sqsn\sqsm}\sum_{\sqsi=1}^\sqsn\sqrt{\sqsx_\sqsi} - \eleError{\counterUB}\geq\thresProba$ implies $\sum_{\sqsi=1}^\sqsn\sqrt{\sqsx_\sqsi}\geq\sqsy$.
  We prove the converse implication.
  Assume that $\sum_{\sqsi=1}^\sqsn\sqrt{\sqsx_\sqsi}\geq\sqsy$.
  By Lemmas~\ref{lemma:sqrt-sum:error} and~\ref{lemma:sqrt-sum:bound:ub:error}, it holds that $\eleError{\counterUB}\leq\frac{1}{\sqsn\sqsm}|\sum_{\sqsi=1}^\sqsn\sqrt{\sqsx_\sqsi}-\sqsy| = \frac{1}{\sqsn\sqsm}(\sum_{\sqsi=1}^\sqsn\sqrt{\sqsx_\sqsi}-\sqsy$).
  We obtain that $\frac{1}{\sqsn\sqsm}\sum_{\sqsi=1}^\sqsn\sqrt{\sqsx_\sqsi} - \eleError{\counterUB}\geq \frac{1}{\sqsn\sqsm}\cdot\sqsy = \thresProba$.
  This shows that the reduction is correct.

  It remains to comment on how to adapt the above reduction to the fixed-interval and parameterised realisability problems.
  Instead of specifying the strategy as an input, we specify the interval partition $\intPart = \integerInterval{1, \counterUB-1}$ for the fixed-interval case, and the parameters $\intNum=1$ for the number of intervals and $\intSize=\counterUB-1$ for the size of intervals in the parameterised case.
  These inputs are such that we check the existence of a well-performing counter-oblivious strategy (with respect to the threshold $\thresProba$ specified above).
\end{proof}

\subsection{\texorpdfstring{$\np$}{NP}-hardness of fixed-interval and parameterised realisability}\label{section:hardness:np}

The goal of this section is to prove that the realisability problem for counter-oblivious strategies is $\np$-hard for the selective termination objective.
This implies the $\np$-hardness of the fixed-interval and parameterised realisability problems: counter-oblivious strategies are single-interval OEISs and are also CISs with a period of one.
We prove this hardness result by a reduction from the problem of deciding if a directed graph has a Hamiltonian cycle.

Let $\graph = (\vertexSet, \edgeSet)$ denote a directed graph where $\vertexSet$ is a finite set of vertices and $\edgeSet\subseteq\vertexSet^2$ is a set of edges.
A \textit{Hamiltonian cycle} of $\graph$ is a simple cycle $\vertex_0\vertex_1\ldots\vertex_\indexLast$ such that $\indexLast=|\vertexSet|$, i.e., a cycle that passes through all vertices exactly once, except the first vertex which is visited twice.
Deciding whether $\graph$ has a Hamiltonian cycle is $\np{}$-complete (e.g.,~\cite{DBLP:books/fm/GareyJ79}).

We sketch a reduction from the problem of deciding if a graph has a Hamiltonian cycle to the counter-oblivious strategy realisability problem for targeted termination.
Let $\graph = (\vertexSet, \edgeSet)$ be a directed graph.
We fix an initial vertex $\vertex_\init$.
We derive an OC-MDP $\ocmdp$ with deterministic transitions from $\graph$ by adding vertices and redirecting transitions.
We add a copy $\vertex_\init'\notin\vertexSet$ of the initial vertex and a fresh absorbing state $\ocState\notin\vertexSet$.
All incoming transitions of $\vertex_\init$ are redirected to $\vertex_\init'$ and the only successor of $\vertex_\init'$ is set to be $\ocState$.
All transitions are given a weight of $-1$.
We assume a counter upper bound $\counterUB\in\{|\vertexSet|+1, \infty\}$ that exceeds the initial counter value chosen below.

We claim that there is a Hamiltonian cycle in $\graph$ if and only if there is a strategy guaranteeing (almost-)sure termination in $\vertex_\init'$ from the initial configuration $(\vertex_\init, [\vertexSet|)$.
Intuitively, all cycles of $\graph$ from $\vertex_\init$ with $\ocCount$ edges (i.e., $\ocCount+1$ vertices) are equivalent to a history from $(\vertex_\init, |\vertexSet|)$ to $(\vertex_\init', |\vertexSet|-\ocCount)$ in $\ocmdpFin{\ocmdp}{\counterUB}$.
If there is a Hamiltonian cycle in $\graph$, because it is simple, we obtain a history of length $|\vertexSet|$ in $\ocmdpFin{\ocmdp}{\counterUB}$ that can be obtained via a pure counter-oblivious strategy, and this strategy provides a positive answer to the realisability problem.
For the converse, we observe that any other simple cycle from $\vertex_\init$ of $\graph$ yields a history terminating in the additional state $\ocState$.
Therefore, if there is no Hamiltonian cycle in $\graph$, all (randomised) counter-oblivious strategies have a history consistent with them that either terminates in $\ocState$ if $\vertex_\init'$ is reached in under $|\vertexSet|$ steps or terminates in $\vertexSet$.

\begin{theorem}\label{theorem:realisability:np-hardness}
  The problem of deciding whether there exists a counter-oblivious (pure or randomised) strategy ensuring almost-sure selective termination is $\np$-hard.
  In particular, the fixed-interval and parameterised realisability problems for selective termination are $\np$-hard.
\end{theorem}
\begin{proof}
  We provide a reduction from the $\np$-complete problem of deciding whether a directed graph contains a Hamiltonian cycle.
  We fix a directed graph $\graph=(\vertexSet, \edgeSet)$ and an initial vertex $\vertex_\init\in\vertexSet$ for the remainder of the proof.

  We consider $\ocmdp = \ocTuple$ such that $\ocStateSpace = \vertexSet\cup\{\vertex_\init', \ocState\}$ (where $\vertex'_\init$, $\ocState\notin\vertexSet$), $\mdpActionSpace = \vertexSet$.
  The transition function is deterministic: we view it as a function $\ocTrans\colon\ocStateSpace\times\ocActionSpace\to\ocStateSpace$.
  We formalise $\ocTrans$ as follows.
  First, for all $(\vertex, \vertex')\in\edgeSet$ such that $\vertex'\neq\vertex_\init$, we let $\ocTrans(\vertex, \vertex') = \vertex'$.
  Second, for all $\vertex\in\vertexSet$ such that $(\vertex, \vertex_\init)\in\edgeSet$, we let $\ocTrans(\vertex, \vertex_\init) = \vertex_\init'$.
  Finally, for all $\vertex\in\vertexSet$, we let $\ocTrans(\vertex_\init', \vertex) = \ocTrans(\ocState, \vertex) = \ocState$.
  All weights are $-1$.
  Recall that counter-oblivious strategies can be seen as functions $\strat\colon\ocStateSpace\to\ocActionSpace$.

  We show that the three following assertions are equivalent:
  \begin{enumerate}
  \item there exists a Hamiltonian cycle of $\graph$;\label{item:fixed-interval:np-hardness:1}
  \item there exists a pure counter-oblivious strategy $\strat$ of $\ocmdp$ such that $\probaG{(\vertex_\init, |\vertexSet|)}{\strat}(\selectiveTermination{\vertex_\init'})=1$;\label{item:fixed-interval:np-hardness:2}
  \item there exists a counter-oblivious strategy $\strat$ of $\ocmdp$ such that $\probaG{(\vertex_\init, |\vertexSet|)}{\strat}(\selectiveTermination{\vertex_\init'})=1$;\label{item:fixed-interval:np-hardness:3}
  \end{enumerate}
  We prove that Item~\ref{item:fixed-interval:np-hardness:1} implies Item~\ref{item:fixed-interval:np-hardness:2} and that Item~\ref{item:fixed-interval:np-hardness:3} implies Item~\ref{item:fixed-interval:np-hardness:1}.
  The implication from Item~\ref{item:fixed-interval:np-hardness:2} to Item~\ref{item:fixed-interval:np-hardness:3} is direct.

  We assume that there exists a Hamiltonian cycle $\vertex_0\vertex_1\ldots\vertex_{|V|}$ of $\graph$.
  Assume without loss of generality that $\vertex_0 = \vertex_\init$.
  It is easy to see that the pure counter-oblivious strategy $\strat$ such that $\strat(\vertex_\indexPosition) = \vertex_{\indexPosition+1}$ for all $\indexPosition\in\integerInterval{|\vertexSet|-1}$ ensures that $\probaG{(\vertex_\init, |\vertexSet|)}{\strat}(\selectiveTermination{\vertex_\init'})=1$.
  The strategy $\strat$ is well-defined because $\vertex_0\vertex_1\ldots\vertex_{|V|-1}$ is a simple path.
  This shows that Item~\ref{item:fixed-interval:np-hardness:1} implies Item~\ref{item:fixed-interval:np-hardness:2}.

  We now prove the contrapositive of the implication from Item~\ref{item:fixed-interval:np-hardness:3} to Item~\ref{item:fixed-interval:np-hardness:1}.
  Assume that there is no Hamiltonian cycle in $\graph$.
  Let $\strat$ be a counter-oblivious strategy.
  We show that termination occurs in a state other than $\vertex_\init'$ with positive probability.
  If $\probaG{(\vertex_\init, |\vertexSet|)}{\strat}(\selectiveTermination{\vertex_\init'})=0$, then the claim is direct.
  We assume that $\probaG{(\vertex_\init, |\vertexSet|)}{\strat}(\selectiveTermination{\vertex_\init'})>0$.
  Thus, there exists a history $\hist = (\vertex_0, |\vertexSet|)\vertex_1(\vertex_1, |\vertexSet|-1)\ldots (\vertex_{|\vertexSet|-1}, 1)\vertex_\init(\vertex_\init', 0)$ consistent with $\strat$ such that $\vertex_0 = \vertex_\init$.
  Since there are no Hamiltonian cycles in $\graph$ and $\hist$ induces a cycle of $\graph$, there must be a state other than $\vertex_\init$ (due to the structure of $\ocmdp$) that is repeated in this induced cycle.
  Let $0 < \indexPosition < \indexPosition' < |\vertexSet|$ such that $\vertex_\indexPosition = \vertex_{\indexPosition'}$.
  It is easy to see that the history starting in $(\vertex_0, |\vertexSet|)$ that follows $\hist$ up to index $\indexPosition'$ and then loops in the cycle between $\vertex_\indexPosition$ and $\vertex_{\indexPosition'}$ until termination is consistent with $\strat$ and therefore $\probaG{(\vertex_\init, |\vertexSet|)}{\strat}(\selectiveTermination{\vertex_\init'})<1$.
  This ends the proof that Item~\ref{item:fixed-interval:np-hardness:3} implies Item~\ref{item:fixed-interval:np-hardness:1}.

  To conclude, we note that $\ocmdp$ can be constructed in polynomial time.
  This ends our $\np$-hardness proof.
\end{proof}

\bibliography{master_references}

\newpage

\appendix

\section{Bounds on sufficient approximations for square-root sum}\label{appendix:square root sum}
The goal of this section is to prove Lemma~\ref{lemma:sqrt-sum:error}.
Let $\sqsx_1, \ldots, \sqsx_\sqsn\in\IN$ and $\sqsy\in\IN$ be fixed for the remainder of this section.
This section has three parts.
First, we recall some field-theoretic notions that are required for the proof.
We refer the reader to~\cite{lang2012algebra} for a reference on field theory.
Second, we show that all roots of the minimal polynomial of $\sum_{\sqsi=1}^\sqsn\sqrt{\sqsx_\sqsi}-\sqsy$ are of a certain form.
We end with a proof of Lemma~\ref{lemma:sqrt-sum:error}.

\subparagraph*{Field-theoretic background.}
A complex number is \textit{algebraic} (over $\IQ$) if it is the root of a polynomial with rational coefficients.
An \textit{algebraic extension} of $\IQ$ is a field $K$ such that $\IQ\subseteq K\subseteq\IC$ such that all elements of $K$ are algebraic.
Let $K\subseteq L\subseteq\IC$ be algebraic extensions of $\IQ$.
We write $L/K$ as shorthand to mean that $L$ is an extension of $K$.
The \textit{minimum polynomial} of $\alpha\in L$ over $K$ is the unique monic polynomial with coefficients in $K$ of minimum degree that has $\alpha$ as a root.
An algebraic number is an \textit{algebraic integer} if its minimal polynomial over $\IQ$ has integer coefficients.
Algebraic integers form a sub-ring of the algebraic closure of $\IQ$.

Given algebraic numbers $\alpha_1, \ldots, \alpha_\indexPosition$, we let $K(\alpha_1, \ldots, \alpha_\indexPosition)$ be the smallest algebraic extension of $K$ containing $\alpha_1$, \ldots, $\alpha_\indexPosition$.
The \textit{degree} of the extension $L/K$, denoted by $[L:K]$, is the dimension of $L$ as a vector space over $K$, and if $L = K(\alpha)$, $[L:K]$ is the degree of the minimal polynomial of $\alpha$ over $K$.
Degrees of successive extensions multiply, in the sense that, given $F/L$, it holds that $[F:K] = [F:L]\cdot [L:K]$.

An extension $L/K$ is \textit{Galois} if and only if any embedding of $K$ in the algebraic closure of $\IQ$ induces an automorphism of $K$.
Given a polynomial $P\in K[X]$, the \textit{splitting field} of $P$ is the smallest algebraic extension of $K$ that contains all of the complex roots of $P$.
If $L/K$ is of finite degree, then $L/K$ is Galois if and only if it is the splitting field of a polynomial in $K[X]$.
If $L/K$ is Galois, the Galois group $\mathsf{Gal}(L/K)$ of $L/K$ is the group formed by the (field) automorphisms of $L$ whose restriction to $K$ is the identity function over $K$.
The order of the Galois group of a Galois extension is the degree of the extension.

\subparagraph*{Minimal polynomials.}
The following lemma provides a set that is guaranteed to contain all roots of the minimal polynomial of $\sum_{\sqsi=1}^\sqsn\sqrt{\sqsx_\sqsi}-\sqsy$.
Through this lemma, we can bound the coefficients of the minimal polynomial of $\sum_{\sqsi=1}^\sqsn\sqrt{\sqsx_\sqsi}-\sqsy$, which is the crux of the proof of Lemma~\ref{lemma:sqrt-sum:error}.

\begin{lemma}\label{lemma:sqrt-sum:roots}
  Let $\beta = \sum_{\sqsi=1}^\sqsn\sqrt{\sqsx_\sqsi} - \sqsy$.
  The minimal polynomial of $\beta$ over $\IQ$ has at most $2^\sqsn$ roots and all are included in the set $\{\sum_{\sqsi=1}^{\sqsn}(-1)^{b_\sqsi}\sqrt{\sqsx_\sqsi} - \sqsy\mid(b_1, \ldots, b_\sqsn)\in\{0, 1\}^\sqsn\}$.
\end{lemma}
\begin{proof}
  Let $P_\beta$ denote the minimum polynomial of $\beta$ and let $K = \IQ(\sqrt{\sqsx_1},\ldots, \sqrt{\sqsx_\sqsn})$.
  To bound the number of roots of $P_\beta$, it suffices to bound its degree, i.e., $[\IQ(\beta):\IQ]$.
  We have $\IQ(\beta)\subseteq K$ because $\beta\in K$ by definition of $K$.
  It follows that $[\IQ(\beta):\IQ]$ is a divisor of $[K:\IQ]$.
  We have that $[K:\IQ]$ is at most $2^\sqsn$ because
  \[[K:\IQ] = \prod_{0\leq\sqsi\leq\sqsn-1}[\IQ(\sqrt{\sqsx_1},\ldots, \sqrt{\sqsx_{\sqsi+1}}):\IQ(\sqrt{\sqsx_1},\ldots, \sqrt{\sqsx_\sqsi})]\]
  and the degrees in the product are one or two; for all $1\leq\sqsi\leq\sqsn$, $\sqrt{\sqsx_\sqsi}$ is a root of $X^2 - \sqsx_\sqsi$.

  We now show that all roots of $P_\beta$ are in $\{\sum_{\sqsi=1}^{\sqsn}(-1)^{b_\sqsi}\sqrt{\sqsx_\sqsi} - \sqsy\mid(b_1, \ldots, b_\sqsn)\in\{0, 1\}^\sqsn\}$.
  First, we note that $K$ is the splitting field of the polynomial $\prod_{\sqsi=1}^\sqsn (X^2 - \sqsx_\sqsi)$.
  Therefore, $K/\IQ$ is Galois (as its degree is finite).
  
  We determine the Galois group of $K/\IQ$.
  Let $R\subseteq\{\sqrt{\sqsx_1}, \ldots, \sqrt{\sqsx_\sqsn}\}$ be a minimal set such that $K = \IQ(R)$.
  Assume that $R = \{\sqrt{\sqsx_1}, \ldots, \sqrt{\sqsx_{\sqsn'}}\}$.
  For all $1\leq\sqsi\leq\sqsn'$, $[K:\IQ(R\setminus\{\sqrt{\sqsx_\sqsi}\})] = 2$ and $K/\IQ(R\setminus\{\sqrt{\sqsx_\sqsi}\})$ is Galois.
  It follows that the group $\mathsf{Gal}(K/\IQ)$ contains, for all $1\leq\sqsi\leq\sqsn'$, the automorphism of $K$ in $\mathsf{Gal}(K/\IQ(R\setminus\{\sqrt{\sqsx_\sqsi}\}))$ that is such that $\sqrt{\sqsx_\sqsi}\mapsto -\sqrt{\sqsx_\sqsi}$ that leaves other elements of $R$ unchanged.
    These different automorphisms commute.
    It follows that these automorphisms generate the Galois group $\mathsf{Gal}(K/\IQ)$ whose order is $2^{\sqsn'}$.

    Let $L$ denote the splitting field of $P_\beta$ (thus $L/\IQ$ is Galois).
    It holds that $L\subseteq K$, because $K/\IQ$ is Galois and $P_\beta$ has a root in $K$.
    On the one hand, elements of $\mathsf{Gal}(L/\IQ)$ are the restrictions of elements of $\mathsf{Gal}(K/\IQ)$.
    On the other hand, the action of $\mathsf{Gal}(L/\IQ)$ on the set of roots of $P_\beta$ is transitive (because $P_\beta$ is irreducible in $\IQ[X]$).
    It follows that the roots are all of the claimed form.
\end{proof}

\subparagraph*{Proof of Lemma~\ref{lemma:sqrt-sum:error}.}
We now provide a proof of Lemma~\ref{lemma:sqrt-sum:error}.

\lemmaSqrtSumError*
\begin{proof}
  Assume that $\sum_{\sqsi=1}^\sqsn\sqrt{\sqsx_\sqsi} \neq\sqsy$.
  Let $P_\beta$ denote the minimal polynomial of $\beta = \sum_{\sqsi=1}^\sqsn\sqrt{\sqsx_\sqsi}-\sqsy$.
  It has integer coefficients, because $\beta$ is an algebraic integer.
  Indeed, square roots of integers are algebraic integers and algebraic integers are a sub-ring of the algebraic closure of $\IQ$.
  We bound the coefficients of $P_\beta$ to conclude via the following result: the non-zero roots of a non-zero polynomial with $k$-bit integer coefficients are greater than $2^{-k}$ in absolute value \cite{householder1970numerical,DBLP:journals/jc/Tiwari92}.

  Let $d$ denote the degree of $P_\beta$.
  We have $d\leq 2^{\sqsn}$ by Lemma~\ref{lemma:sqrt-sum:roots}.
  By the same result, it follows that the roots of $P_\beta$ are of absolute value at most $\sum_{\sqsi=1}^\sqsn\sqsx_\sqsi + \sqsy < 2^\lambda$.
  From the decomposition of $P_\beta$ in linear factors, we obtain that its coefficients are sums of at most $2^d$ products of roots of $P_\beta$, and thus are strictly less than $2^d\cdot (2^\lambda)^d =2^{d(\lambda+1)}$ in absolute value, i.e., their bit-size is at most $d(\lambda+1)$.
  We obtain that $|\beta|> 2^{-d(\lambda+1)}\geq 2^{-2^{\sqsn}(\lambda+1)}$.
\end{proof}

\end{document}